%% file: main.tex
\documentclass[11pt]{article}
\usepackage[letterpaper,margin=1in]{geometry}
\usepackage[dvipsnames]{xcolor}
\definecolor{DarkRed}{rgb}{0.65,0,0}
\usepackage{amsmath,amssymb,amsthm}
\usepackage{graphicx}
\DeclareTextFontCommand{\textsc}{\normalfont\scshape}
\usepackage{algorithm}
\usepackage{algpseudocode}
\usepackage[colorlinks=true,linkcolor=DarkRed,linktoc=page,citecolor=green!50!black,urlcolor=blue]{hyperref}
\usepackage[capitalize]{cleveref}
\hypersetup{bookmarksopen=true,bookmarksnumbered=true,bookmarksdepth=2}
\newcommand{\email}[1]{\href{mailto:#1}{#1}}
\numberwithin{equation}{section}
\numberwithin{algorithm}{section}
\makeatletter
\providecommand{\theHALG@line}{}
\renewcommand{\theHALG@line}{\theHalgorithm.\arabic{ALG@line}}
\makeatother
\crefname{section}{section}{sections}
\crefname{subsection}{subsection}{subsections}
\Crefname{figure}{Figure}{Figures}
\crefformat{equation}{\textup{#2(#1)#3}}
\input{header}

\title{Local Important Separators with Applications\\ to Unbreakable Decomposition and Multicut}
\input{metadata}
\date{}

\begin{document}

\pagenumbering{roman}
\pagestyle{empty}

\maketitle
\thispagestyle{empty}

\input{abstract}

\clearpage
\tableofcontents
\clearpage
\pagenumbering{arabic}
\pagestyle{plain}

\input{intro}
\section{Technical Overview.}

\paragraph{Local Important Separators.} We are given an undirected graph $G$ and vertices $s$ and $t$. The goal is to enumerate all important $s$-$t$ separators with volume at most $\nu$ (on the $s$-side) and size at most $k$.

The natural starting point is to try and use the local min-cut algorithm of~\cite{forster2020computing}, which has the following guarantee: it either finds \emph{some $s$-$t$ cut} of volume at most $\nu \lambda$ and size at most $\lambda$  containing $s$ in time $\Otil(\nu \lambda^2)$, or declares that there is no $s$-$t$ cut of volume at most $\nu$ and size at most $\lambda$, with high probability. Let us refer to this subroutine as $\LocalCut(G, s,t,\nu, \lambda)$. Let us call this procedure and obtain its output, say a cut $(L,R)$.  This cut may not even be an important separator: In fact, a standard fact, which follows from submodularity, is that for any important separator $(L,R)$ and the furthest min-cut $(L^*,R^*)$ between $s$ and $t$, $L^* \subseteq L$. But the returned cut from $\LocalCut(s,t,\nu,\lambda)$ may not even be a furthest $s$-$t$ min-cut.

Therefore, we first develop a subroutine $\LocalMincut(s,t,\nu,\lambda)$, whose goal is to find the furthest $s$-$t$ min-cut $(L^*,R^*)$, if \emph{we are promised} that the furthest $s$-$t$ min-cut has volume at most $\nu$ and size equal to $\lambda$.\footnote{The formal algorithm does not require this promise: if the furthest min-cut has volume greater than $\nu$, no relevant important separator remains in that branch.} We obtain this as follows. For simplicity, assume that we know $\lambda$.

First, we call $\LocalCut{(s,t,\nu,\lambda)}$, which returns a cut $(L,R)$. This cut must be of size exactly $\lambda$, since we assumed that the furthest $s$-$t$ min-cut has size equal to $\lambda$. We pick any cut vertex of this cut, and see if the cut can be ``pushed'' further to the right without increasing the size. More concretely, we pick a cut vertex $x$, contract $L \setminus R \cup \{x\}$ into a single vertex $s'$ to obtain $G'$, and then recursively call $\LocalMincut{(G', s',t,\nu',\lambda)}$ where $\nu'$ is appropriately adjusted. 
 If $\LocalMincut(G', s',t,\nu',\lambda)$ succeeds, we simply return its output. If it does not succeed, it means that our guess was wrong, and this means that $x$ must be a cut vertex, which cannot be contracted with $L \setminus R$. In this case, using this information, we can decrease $\lambda$ by $1$ - that is, we call $\LocalMincut(G'', s'',t,\nu'',\lambda - 1)$ where again the bound $\nu''$ is appropriately computed, and $s''$ is obtained by contracting $L \setminus R$ into a single vertex and $G''$ is the resulting graph, where we also delete the vertex $x$. This recursive algorithm enables us to find the furthest $s$-$t$ min-cut.

Once we know how to  ``reach'' the furthest $s$-$t$ min-cut, $(L^*,R^*)$ of size $\lambda$, we can now try to emulate the branching  algorithm that enumerates important separators. This is the goal of the subroutine \LocalImpSep{$(G, s,t,\nu,k,\lambda)$}, which aims to find all important separators of volume at most $\nu$ and size at most $k$, where we are promised that the minimum $s$-$t$ cut in $G$ is of size $\lambda$. First, we compute the maximal min-cut by calling $\LocalMincut(G,s,t,\nu,\lambda)$. Let $(L^*,R^*)$ be the returned cut. We pick a cut vertex $x \in (L^* \cap R^*)$. Then, we guess if this vertex is reachable from $s$ or not, branching two different ways.  If we guess that it is reachable, we contract it and continue. In this case, since $(L^*,R^*)$ was already the furthest min-cut, after contraction, the size of the min-cut must increase. Hence we can increase $\lambda$ by $1$ and call $\LocalImpSep(G',s',t,\nu', k, \lambda + 1)$, where $s'$ is obtained by contracting $L^* \setminus R^*$ and $x$ into a single vertex. The other case is the branch where this vertex is on the boundary of the important separator, and hence we can reduce $k \gets k - 1$ - thus we call $\LocalImpSep(G'', s'', t, \nu'', k - 1, \lambda-1)$, where $s''$ is the single vertex obtained by contracting $L^* \setminus R^*$, $G''$ is the resulting graph after contraction and deleting $x$, and $\nu''$ is appropriately modified. 

In this short overview, we do not discuss the full details of the running time, but we note that the branching algorithm clearly has at most $4^k$ branches, as in each recursive call the parameter $2k - \lambda$ decreases by at least $1$. It remains to bound the time of a single invocation of $\LocalMincut$. Recall that $\LocalMincut$ repeatedly tries to push the current cut past a cut vertex $x$. Each successful push contracts $L \setminus R \cup \{x\}$ and reduces the volume parameter by at least $1$, so there are at most $\nu$ successful pushes. Charging each of them the worst-case running time of $\LocalCut$ would give a quadratic dependence on $\nu$; instead, we charge each successful push to the volume it contracts, and these volumes sum to at most $\nu$. An unsuccessful push decreases $\lambda$ by $1$, so it happens at most $\lambda \leq k$ times. Thus, one invocation of $\LocalMincut$ takes $\Otil(\nu\poly(k))$ time, giving $\Otil(\nu 4^k\poly(k))$ time overall.

\paragraph{Unbreakable Decomposition.} On a high level, our algorithm follows the approach of~\cite{anand2025unbreakable}. To make our explanation clear, we use two definitions. Given a graph $G$, an $(X,v,q,k)$ witness is a vertex cut $(L,R)$, such that $X \subseteq R$, $v \in L \setminus R$, $|L| > q$ and $|L \cap R| \leq k$.
An $(X,v,q,k)$ witness for a terminal set $T \subseteq V(G)$ is a natural generalization: it is a vertex cut $(L,R)$ such that $X \subseteq R$, $v \in L \setminus R$, $|L \cap T| > q$ and $|L \cap R| \leq k$.

From~\cite{anand2025unbreakable}, it is sufficient to solve the following problem: Given a graph $G$, a set $X \subseteq V(G)$ such that $X$ is $(k,k)$ unbreakable in $G$, find a set $T$ satisfying $X \subseteq T \subseteq V(G)$, such that there is no $(X,v,q',2k)$ witness for $T$ in $G$. If we do this, then it is easy to get an $(O(q'),k)$ unbreakable decomposition with adhesion $O(q')$. Thus it suffices to find a set $T$ to ensure that there is no $(X,v,2k,2k)$ witness for $T$ in $G$ and we focus on this goal for the rest of this overview. Let $k' = 2k$. 

We could adapt the natural approach, followed in essentially every previous paper on constructing unbreakable decomposition, by setting $T = V(G)$ first, and then refining $T$. Then, if there is an $(X,v,k',k')$ witness for $T$ in $G$, $(L,R)$, we update $T$ to $T'$ where $T' = T \setminus L \cup (L \cap R)$. Unfortunately, such an approach can be quite slow, due to the following problem: the set $T$ keeps changing, but it is not monotone: while some vertices are removed from $T$, they can be added back to $T$ in some later iteration. Thus even if we could find witnesses in time proportional to their volume, which we can hope for thanks to our local algorithm, the overall running time may not be linear, as a unit of volume may be processed many times.

Instead, we try to achieve a reduction to the case when $G$ has no $(X,v,q',k')$ witness for $q' = 2^{O(k)}$. If this happens, then now we know that any witness $(L,R)$ must have $|L| \leq q'$ and hence its volume must be at most $\nu \leq q'^2 = 2^{O(k)}$. This is where our local algorithm proves useful. Notice that if we assume connectivity of witness cuts, then given $T$, if there is an $(X,v,q,k)$ witness for $T$ in $G$, it can be ``pushed'' towards $X$ as much as possible. Formally, if there is an $(X,v,q,k)$ witness for $T$, there must be an $(X,v,q,k)$ witness that is an important $v$-$X$ separator.

If we now use our local algorithm, then we can find an $(X,v,k',k')$ witness in $\Otil(\nu 2^{O(k)}) = \Otil(2^{O(k)})$ time. When we update $T \gets T \setminus L \cup (L \cap R)$, since $|L \cap R| < |L \cap T|$, we decrease the size of $T$ by at least $1$. This means that there are at most $m$ such steps of finding witnesses and updates, and hence we would finish in only $\Otil(2^{O(k)}m)$ time. However, we must handle disconnected witnesses, and this is accomplished by an additional branching step so that the running time now becomes $\Otil(k^{O(k)}m)$.

But in order to use this approach, we must first reduce to the case when $G$ has no $(X,v,q',k')$ witness.
This concept of reducing to a case where the graph is ``weakly unbreakable'' to facilitate local search is similar to the approach adopted by~\cite{korhonen2025linear}. But the way we obtain this reduction is quite different and conceptually simpler and again leverages our local algorithm.

So how do we accomplish this reduction? Given $G$, we obtain a graph $G'$ with $V(G') \subseteq V(G)$, such that for any $T \subseteq V(G')$ if there is no $(X,v,k',k')$ witness for $T$ in $G'$, then there is no $(X,v,k',k')$ witness for $T$ in $G$. $G'$ must additionally satisfy the property we want: that there is no $(X,v,q',k')$ witness in $G'$. To construct this graph $G'$, instead of adopting the approach of~\cite{korhonen2025linear}, we replace this by a single step which only guarantees something weaker: that there is no $(X,v,q',k')$ \emph{connected witness} in $G'$. Here an $(X,v,q',k')$  connected witness is an $(X,v,q',k')$  witness $(L,R)$ where $G[L \setminus R]$ is connected and $N_G(L \setminus R)=L\cap R$. It turns out that eliminating such connected witnesses is enough, and this weaker guarantee is enough to accomplish the last step similar to the above argument.

Thus, the problem now reduces to eliminating all connected $(X,v,q,k)$ witnesses in $G$ which is the subject of \Cref{sec:connectedwitness}. Here, we use our local algorithm again, but to circumvent the problem of repeatedly examining volume elements, we use the notion of downward well-linkedness introduced by~\cite{korhonen2025linear}. A downward well-linked set $S$ satisfies for any separation $(L,R)$ of $G[S \cup N(S)]$, either $|N(S) \cap L| \leq |L \cap R|$ or $|N(S) \cap R| \leq |L \cap R|$. In simple language, what this means is that, if the goal is to separate vertices of $N(S)$ by deleting vertices from $S \cup N(S)$, then it is essentially always cheaper to cut vertices of $N(S)$ - any cut inside $S$ can always be \emph{uncrossed} to not cut any vertex of $S$, while still achieving the same separation of $N(S)$.

\paragraph{Organization.}
In \Cref{sec:local} we describe our local algorithm for important separators. Sections \ref{sec:connectedwitness} through \ref{sec:Decomp} deal with our result on unbreakable decomposition. \Cref{sec:connectedwitness} describes how to eliminate $(X,v,q',k')$ connected witnesses from a given graph. \Cref{sec:disconnectedwitness}, when given as input a graph with no $(X,v,q',k')$ connected witness, gives the algorithm for removing all $(X,v,k',k')$ witnesses to obtain the desired set $T$. \Cref{sec:keysubroutine} gives the key subroutine we need for unbreakable decomposition by combining the results of \Cref{sec:connectedwitness} and \Cref{sec:disconnectedwitness}. Finally,~\Cref{sec:Decomp} gives our unbreakable decomposition algorithm. In \Cref{sec:singlesink}, we show our result on single-sink important separators. In \Cref{sec:multicut} we show our result on {\textsc{Multicut}}. 

\section{Preliminaries.}\label{sec:preliminaries}

Throughout this paper, we use $n$ to denote the number of vertices and $m$ to denote the number of edges in a graph. All graphs are undirected, unweighted, and connected unless otherwise stated, and hence we shall assume $m \ge n-1$. Given a graph $G$ and a subset of vertices $U \subseteq V(G)$, $G[U]$ denotes the subgraph induced on the set of vertices $U$ and $E_G(U)$ denotes the set of edges in $E(G)$ whose both endpoints are in $U$, i.e. the set of edges in $G[U]$. The set of neighbours of a vertex $v$ in $G$ is denoted by $N_G(v)$. We denote by $N_G(U):= \bigcup_{v \in U} N_G(v) \setminus U$ the set of neighbours of $U$. We omit the subscripts when the graph is clear from the context. 

\paragraph{Separators/Vertex Cuts.} 
A \emph{separator} $(L,R)$ is such that $L \cup R = V(G)$, $L\setminus R$ and $R\setminus L$ are not empty, and there is no edge between $L\setminus R$ and $R\setminus L$. The \emph{size} of the vertex cut $(L,R)$ is $|L\cap R|$. We emphasize that, we will view $(L,R)$ as an \emph{ordered pair}, because in some definitions, the order of $(L,R)$ does matter.
Throughout, we only consider vertex cuts (not edge cuts), so we usually write cuts as an abbreviation of vertex cuts. We further have the following definition of $(s,t)$-separators, which separate two given vertices $s$ and $t$. While it shall be clear from the context, typically we represent a separator by only one side of it: thus usually $(L,R)$ is represented only by $L$. This is crucial to obtain bounds on the running time for local algorithms, as the volume bounds apply only to one side of the cut. 

\begin{definition}[$(s,t)$-separators]%
Let $G$ be an undirected graph. For two vertices $s,t\in V(G)$, two vertex subsets $L,R\subseteq V(G)$ form an $(s,t)$-separator $(L,R)$ if $s\in L\setminus R$, $t\in R\setminus L$ and there is no edge connecting $L\setminus R$ and $R\setminus L$.
\end{definition}

\begin{definition}[Connected separators]%
A separator in $G$ is \emph{connected} if $G[L\setminus R]$ is connected.
\end{definition}

Finally, observe that a separator $(L,R)$ in $G$ can be uniquely fixed once we define $L\setminus R$ and $L\cap R$. Hence, our algorithms will generally represent (and store) $(L,R)$ as $L\setminus R$ and $L\cap R$, which is important to achieve local running time.

\paragraph{Balance, Adhesion and Unbreakability.}

For any vertex set $X\subseteq V(G)$, $X$ is \emph{$\alpha$-balanced} if each connected component $C$ of $G\setminus X$ has size $|C|\le \alpha n$. The \emph{adhesion} $\sigma_G(X)$ of $X$ in $G$, or simply the adhesion of $X$, is the maximum, over connected components $C$ of $G\setminus X$, of the quantity $|N_G(C)|$. Sometimes we will also refer to $|N_G(C)|$ as the adhesion of component $C$.

\begin{definition}[Unbreakability]%
A vertex set $X \subseteq V(G)$ is $(q,k)$-\emph{unbreakable in $G$} if every vertex cut $(L,R)$ of size at most $k$ satisfies $|L\cap X|\leq q$ or $|R\cap X|\leq q$. A $(q,k)$-\emph{breakable witness of $X$ in $G$} is a vertex cut $(L,R)$ of $G$ of size at most $k$ satisfying $|L\cap X| >  q$ and $|R\cap X|> q$.
\end{definition}
By definition, $X$ is $(q,k)$-unbreakable in $G$ iff there is no $(q,k)$-breakable witness of $X$ in $G$. Note that any set $X$ of size at most $q$ is vacuously $(q,k)$-unbreakable.

\begin{definition}
Given a graph $G$ and a vertex set $X$, an $(X,v,q,k)$ witness is a vertex cut $(L,R)$, such that $X \subseteq R$, $v \in L \setminus R$, $|L| > q$ and $|L \cap R| \leq k$.
An $(X,v,q,k)$ witness for a terminal set $T \subseteq V(G)$ is a natural generalization: it is a vertex cut $(L,R)$ such that $X \subseteq R$, $v \in L \setminus R$, $|L \cap T| > q$ and $|L \cap R| \leq k$. We say that an $(X,v,q,k)$ witness (for $T$) is connected if $G[L \setminus R]$ is connected and $N_G(L \setminus R)=L\cap R$.

\end{definition}

\paragraph{Unbreakable Decomposition.} Next, we define unbreakable decompositions.
 
\begin{definition}[Tree Decomposition]%
\label{def:TreeDecomp}%
A \emph{tree decomposition} of a graph $G$ is a pair $(T,\beta)$, where $T$ is a tree and $\beta: V(T) \rightarrow 2^{V(G)}$ is a mapping that assigns to every tree node $t$ a subset $\beta(t)\subseteq V(G)$, called a \emph{bag}. Furthermore, $(T,\beta)$ satisfies the following.
\begin{itemize}
    \item For each vertex $v\in V(G)$, the set $\{t \mid v\in \beta(t)\}$ induces a connected subtree of $T$.
    \item For each edge $\{u,v\}\in E(G)$, there is a tree node $t$ where $u,v\in \beta(t)$.
\end{itemize}
\end{definition}

A \emph{rooted} tree decomposition is a tree decomposition $(T, \beta)$ together with a designated root node $r \in V(T)$. For any node $t\in V(T)$ with parent $t'$, the adhesion of a tree node $t$ is $\sigma(t) = \beta(t) \cap \beta(t')$. We define $\sigma(r) = \emptyset$ for the root $r$. The \emph{adhesion} of $T$ is $\max_{t \in T}|\sigma(t)|$. 
For every $t\in V(T)$, we also define the sets 
\[
\gamma(t) = \bigcup_{\text{descendants $s$ of $t$}} \beta(s)\text{ and }G_t = G[\gamma(t)] - E_G(\sigma(t)).
\]

\begin{definition}[Unbreakable Decomposition]\label{def:UnbreakableDecomp}A $(q,k)$-unbreakable decomposition of $G$ is a rooted tree decomposition $(T,\beta)$ where each bag $\beta(t)$ is $(q,k)$-unbreakable in $G$. The decomposition admits the stronger \emph{subtree unbreakability} property if each bag $\beta(t)$ is $(q,k)$-unbreakable in $G_{t}$.
\end{definition}

The main goal of this paper (\Cref{sec:Decomp}) is a fast algorithm for computing an unbreakable decomposition with  subtree unbreakability property and  small adhesion.

For each tree node $t$, let $\alpha(t) = \gamma(t)\setminus\sigma(t)$.
We say that a rooted tree decomposition $(T, \beta)$ is \emph{compact} if for every node $t \in V(T)$ for which $\sigma(t) \neq \emptyset$, $G[\alpha(t)]$ is connected and $N_G(\alpha(t)) = \sigma(t)$. Compactness is a useful additional property to have for unbreakable decompositions for various applications.

\input{local_cut}

\input{connected_witness_elimination}

\input{phase2}

\input{decomposition}

\appendix

\input{appendix_fast_multicut}

\bibliographystyle{alpha}
\bibliography{references}
\end{document}

%% file: header.tex
\usepackage{enumerate}
\usepackage{multirow}
\usepackage{array}
\usepackage{thm-restate}
\usepackage{float}
\usepackage{tablefootnote}
\usepackage{longtable}
\usepackage{xspace}
\usepackage{enumitem}

\theoremstyle{plain}
\newtheorem{theorem}{Theorem}[section]
\newtheorem{lemma}[theorem]{Lemma}

\newtheorem{observation}[theorem]{Observation}
\newtheorem{claim}[theorem]{Claim}
\theoremstyle{definition}
\newtheorem{definition}[theorem]{Definition}
\theoremstyle{remark}

\AddToHook{cmd/appendix/after}{\crefalias{section}{appendix}}

\newcommand{\poly}{\mathrm{poly}\xspace}
\newcommand{\polylog}{\mathrm{polylog}\xspace}
\newcommand{\vol}{\mathrm{vol}}

\newcommand{\OO}{O}

\newcommand{\LocalImpSep}{{\textsc{LocalImpSep}}}
\newcommand{\LocalMincut}{{\textsc{LocalMincut}}}
\newcommand{\LocalCut}{{\textsc{LocalCut}}}

\newcommand{\eps}{\epsilon}

\global\long\def\val{{\rm val}}

\global\long\def\Decomp{\textsc{Decomp}}

\global\long\def\final{{\rm final}}
\global\long\def\Pop{\textsc{Pop}}
\global\long\def\WellLink{\textsc{WellLink}}
\global\long\def\Clique{{\rm Clique}}
\global\long\def\nnew{{\rm new}}

\global\long\def\old{{\rm old}}
\global\long\def\Torso{{\rm Torso}}
\global\long\def\MakeLean{\textsc{MakeLean}}

\newcommand{\wtilde}{\widetilde}

\newcommand{\Otil}{\wtilde{O}}

\DeclareMathOperator{\dist}{dist}
\DeclareMathOperator{\OPT}{OPT}

\newcommand{\R}{\mathbb R}

%% file: metadata.tex
\author{Aditya Anand\thanks{University of Michigan (\email{adanand@umich.edu}).}
\and Euiwoong Lee\thanks{University of Michigan (\email{euiwoong@umich.edu}). Supported in part by NSF award CCF-2236669.}
\and Yaowei Long\thanks{University of Michigan (\email{yaoweil@umich.edu}).}
\and Thatchaphol Saranurak\thanks{University of Michigan (\email{thsa@umich.edu}). Supported by NSF Grant CCF-2238138 and a Sloan Fellowship.}}
\hypersetup{
  pdftitle={Local Important Separators with Applications to Unbreakable Decomposition and Multicut},
  pdfauthor={Aditya Anand, Euiwoong Lee, Yaowei Long, Thatchaphol Saranurak}
}

%% file: abstract.tex
\begin{abstract}
    We present the first \emph{local FPT} algorithm for finding important separators: Given an undirected graph $G$ with $n$ vertices and $m$ edges, vertices $s,t$ and parameters $k$ and $\nu$, we can enumerate all $s$-$t$ important separators of size at most $k$ and volume at most $\nu$ in time $\tilde{O}(\nu 2^{O(k)})$. Known approaches spend either $\Omega(m)$ time (not local) or $\nu^{\Omega(k)}$ time (not FPT).

Using this local algorithm, we show the following applications. 
\begin{enumerate}
\setlength{\itemsep}{0pt}
\setlength{\parsep}{0pt}
\item The first is an improved construction of unbreakable decomposition, a powerful generalization of tree decompositions.
For any $k$, we can compute a $(5k,k)$-unbreakable decomposition with adhesion $4k$ in $\Otil(k^{O(k)}m)$ time.
This strictly subsumes the result of Anand, Lee, Li, Long, and Saranurak (SODA 2025) and is faster than the recent $O(k^{O(k^2)}m)$-time algorithm by Korhonen (STOC 2025) when $k = \omega\Bigl(\sqrt{\frac{\log \log n}{\log \log \log n}}\Bigr)$. Our algorithm is also simpler than that of Korhonen (STOC 2025).
\item We give a near-linear time algorithm for \emph{single-sink important separators}. For any sink $t$, we can list all $v$-$t$ important separators of size at most $k$ across all $v \in V(G)$ in time $\Otil(2^{O(k)}m)$. Known approaches trivially compute  $v$-$t$ important separators for each $v$, taking $\Omega(4^k k mn)$ time.
\item Using the above result on single-sink important separators, we speed up undirected shadow removal and give the first almost-linear time FPT algorithm for \textsc{Edge Multicut}.
\end{enumerate}
\end{abstract}

%% file: intro.tex
\section{Introduction.}
In the field of parameterized algorithms, problems on graphs form one of the most studied classes of problems. In particular, cut problems, where the goal is to separate some vertices subject to conditions on size, sparsity and connectivity requirements, have received a lot of attention - some of these problems being {\textsc{Multiway Cut}}, {\textsc{Multicut}}, {\textsc{Minimum Bisection}}, {\textsc{Feedback Vertex Set}}, {\textsc{Balanced Separator}}~\cite{marx2006parameterized,fm06,chen2008fixed,marx2011fixed,cygan2014minimum,cygan2020randomized,li2022detecting}.

Perhaps one of the most valuable tools in studying cut problems has been the technique of \emph{important separators}, introduced by Marx~\cite{marx2006parameterized}, finding applications in various settings, see for instance~\cite{marx2006parameterized,chen2008fixed,chen2009improved, chitnis2013fixed, lokshtanov2013clustering,lokshtanov2021fpt,anand2025allsubsetsimportantseparatorsapplications}. Given a graph $G$ and two disjoint subsets of vertices $A$ and $B$, a vertex set $X \subseteq V(G)$ is called an \emph{important separator} (a) if it is an $A$-$B$ separator, that is, there is no path between $A$ and $B$ in $G \setminus X$, (b) it is inclusion-wise minimal, so that no proper subset of $X$ is an $A$-$B$ separator, and (c) for no other separator $X'$ with $|X'| \leq |X|$, $R(X') \supsetneq R(X)$, where $R(X)$ is the set of vertices reachable from $A$ in $G \setminus X$. Intuitively, it captures right-most $A$-$B$ separators, so that the size of an $A$-$B$ separator to the right of this separator (pushed further towards $B$) can only be larger. It is well-known that given two sets $A$ and $B$, there are at most $4^k$ $A$-$B$ important separators of size at most $k$, and they can be enumerated in time $O(4^k k (m + n))$~\cite{chen2009improved}.

While this seems like the best possible running time bound for the enumeration of important separators, for many applications, one needs several invocations of important separator enumeration between many different pairs of sets, leading to a super-linear running time. In the field of fast graph algorithms, a similar story has evolved with the computation of $s$-$t$ min-cuts: many algorithms need multiple computations of $s$-$t$ min-cuts/max-flows, and hence methods have been developed to further reduce the cost of such computation. In particular, the techniques of \emph{isolating cuts}~\cite{li2025deterministic}, which computes multiple pair-wise min-cuts using only polylogarithmic max-flows, and \emph{local min-cuts}~\cite{chechik2017faster,forster2020computing,nalam2023maximal} stand out as two of the successful approaches which have gained traction in recent years. In both these techniques, local computation allows one to compute several $s$-$t$ max-flows but on smaller graphs, whose \emph{total size} is small, the key observation being that if the running time for each computation only depends on the size of the graph, the total running time can still be (almost) linear.

In this paper, we focus on local cut algorithms viewed through the lens of parameterized complexity. We generalize local computation of $s$-$t$ min-cuts to $s$-$t$ important separators, which generalize $s$-$t$ minimum cuts. Our main result is such a local algorithm for finding important separators. Equipped with our local algorithm, we use it to improve the best known construction of unbreakable decompositions, a key tool for obtaining FPT algorithms for cut problems in parameterized complexity. We also use our local algorithm to solve the single-sink important separators problem. We discuss each of these results in more detail in the following paragraphs.

\paragraph{Local Important Separators.} 
Local computation has been a central theme in the graph algorithms literature. In particular, local algorithms for graph cuts have been studied before. Given a graph $G$, the volume of a set $S$ is the sum of the degrees of the vertices in $S$. The local version of the minimum cut problem can then be phrased as follows: Given a graph $G$, a source $s$ and parameters $k,\nu$, is there a subset $S \subseteq V(G)$ with $s \in S$, such that the volume of $S$ is at most $\nu$, and the cut-size of $S$, $|\delta_G(S)| \leq k$? This problem was first studied by~\cite{chechik2017faster}, who showed that it is possible to detect if such a set exists in time $O(\nu k^{O(k)})$. Later, \cite{nanongkai2019breaking} and   \cite{forster2020computing} removed the exponential dependency on $k$, solving the same problem in $\Otil(\nu^{1.5} k)$ and $\Otil(\nu k^2)$ time, respectively\footnote{Throughout the paper, we use $\tilde{O}(\cdot)$ to hide $\polylog(n)$ factors.}.

In this work, we consider local algorithms from a parameterized perspective. We show a \emph{local} algorithm for enumerating important $s$-$t$ separators for fixed vertices $s$ and $t$. Given $s,t,k$ and a \emph{volume parameter} $\nu$, our algorithm finds a list of  at most $4^k$ $s$-$t$ separators of size at most $k$ and guarantees that every $s$-$t$ important separator of size at most $k$ and volume at most $\nu$ is contained in this list. Here the volume of an $s$-$t$ cut $(L,R)$ is defined as the total degree of the vertices in $L \setminus R$.\footnote{In this paper, we deal with vertex cuts, and they are represented as $(L,R)$, where $L \cap R$ is the set of cut vertices, separating $L \setminus R$ from $R \setminus L$. For a separator $X$, let $S$ be the vertices reachable from $s$ in $G-X$; its separation representation is $(S\cup X,V(G)\setminus S)$.} The running time of our algorithm depends only on $\nu$ and $k$ and can be bounded by $\Otil(\nu  4^k  \poly(k))$. The previous best known bound for this problem is an algorithm running in time $O(\nu^{O(k)})$, implicit in~\cite{korhonen2025linear}. Our algorithm can thus be viewed as the first \emph{local FPT} algorithm for this problem. We believe this notion could be interesting for other problems too.

\begin{theorem}[Local Important Separators]\label{thm:breakable}%
Let $G$ be an undirected graph and $s,t \in V(G)$ and  parameters $\nu, k$ be given. Then there is a randomized algorithm that outputs a list of at most $4^k$ $(s,t)$-separators, each of size at most $k$ and volume at most $\nu$, such that every important $s$-$t$ separator of size at most $k$ and volume at most $\nu$ is contained in this list with high probability. The algorithm runs in time $\Otil(\nu 4^{k} \poly(k))$.
\end{theorem}

This leads to much faster algorithms in settings where one needs to compute important separators between different pairs of vertices. As a first application of our result, we use our local algorithm to show a new improved algorithm for computing unbreakable decompositions.

\paragraph{Unbreakable Decomposition.}

Unbreakable Decomposition has been a very important tool to solve cut problems. It was introduced by~\cite{cygan2014minimum}, leading to the first FPT algorithm for \textsc{Minimum Bisection}{} and has been the key tool to obtain FPT algorithms for many cut problems, including the \textsc{Minimum $p$-Way Cut} problem~\cite{cygan2020randomized,lokshtanov2022parameterized}, model checking~\cite{pilipczuk2021algorithms,lokshtanov2018reducing,schirrmacher2024model}, connectivity oracles under vertex failures~\cite{pilipczuk2021algorithms}, \textsc{Multiway Node Hubs}~\cite{saurabh202364}, \textsc{Judicious Partitions}~\cite{lokshtanov2019balanced}, \textsc{Fair Bisection}~\cite{inamdar2023parameterized}, and deletion to bounded degree graphs~\cite{agrawal2022fixed}. See~\cite{anand2025unbreakable} for a more detailed discussion.
A vertex set $T$ is said to be $(q,k)$ unbreakable in graph $G$ if for every vertex cut $(L,R)$ in $G$ with cut-size $|L \cap R| \leq k$, either $|L \cap T| \leq q$ or $|R \cap T| \leq q$. A $(q,k)$ unbreakable decomposition with adhesion $\sigma$ is a tree decomposition where every bag is $(q,k)$ unbreakable, and for any parent-child bags the intersection between them is at most $\sigma$.~\cite{cygan2014minimum} showed an algorithm for obtaining a $(2^{O(k)},k)$ unbreakable decomposition. Later,~\cite{cygan2020randomized} showed an optimal $(k,k)$ unbreakable decomposition algorithm with adhesion $k$, with a running time of $k^{O(k)} n^{O(1)}$.

More recently, \emph{fast algorithms} for unbreakable decomposition, where the dependency on $n$ is close to linear, have been of significant interest. \cite{anand2025unbreakable} gave the first such progress: they showed a close-to-linear time algorithm that constructs a $(O\bigl(\frac{k}{\epsilon}\bigr), k)$ unbreakable decomposition with adhesion $O\bigl(\frac{k}{\epsilon}\bigr)$ running in time $k^{O(k)}m^{1 + \epsilon}$ for any $\epsilon > 0$. Later,~\cite{korhonen2025linear} studied $k$-lean decompositions. $k$-lean decomposition is a natural and powerful generalization of $(k,k)$ unbreakable decomposition.~\cite{korhonen2025linear} gave a linear time algorithm for every fixed $k$ to construct $k$-lean decompositions - the algorithm runs in time $k^{O(k^2)}m$. The dependence on the graph size $m$ here is linear which is ideal, but the dependence on $k$ does not seem optimal and raises the question of if it is possible to do better. Further, the algorithm of~\cite{korhonen2025linear} is complex with several non-trivial steps combining various observations.

In this work, we show a $(5k,k)$  unbreakable decomposition with adhesion $4k$ in time $k^{O(k)}m\,\mathrm{polylog}(n)$. Our result is obtained by combining the approaches in~\cite{anand2025unbreakable} and~\cite{korhonen2025linear} in a new way, with our local algorithm being the main tool for every step. The dependence of the running time of our algorithm as a function of $k$ matches the best known $k^{O(k)}$ for unbreakable decomposition, as opposed to $k^{O(k^2)}$ in~\cite{korhonen2025linear}, at the cost of the factor of 5 in the unbreakability (we note that many applications of unbreakable decomposition remain unaffected by a constant factor overhead). Our algorithm is simpler than that of~\cite{korhonen2025linear}, and is faster when $k = \omega\bigl({\sqrt{\frac{\log \log n}{\log \log \log n}}}\bigr)$.

\begin{theorem}
Given a graph $G$ and a positive integer $k$, there is a randomized algorithm that computes a $(5k,k)$ unbreakable decomposition with adhesion $4k$ in time $k^{O(k)} m\,\polylog(n)$ with high probability.
\end{theorem}

As a corollary, we also immediately and directly improve the algorithm for finding $p$-Way Cuts\footnote{Here the goal is to partition the graph into $p$ parts such that the cut-size is at most $k$. This problem is more popularly known as $k$-Way-Cut or $k$-Cut in the literature, which we avoid for confusion with the cut-size parameter $k$} of size $k$ in~\cite{anand2025unbreakable} from $2^{O(k \log k)}m^{1 + o(1)}$ to $2^{O(k \log k)}m\,\polylog(n)$.
Our algorithm also inherits the desirable properties of logarithmic depth and sub-tree unbreakability as in~\cite{anand2025unbreakable} (see~\Cref{sec:preliminaries} for details).

Conceptually, our algorithm overcomes two barriers in~\cite{anand2025unbreakable} and~\cite{korhonen2025linear}. It improves the local search scheme of~\cite{korhonen2025linear} (see the discussion below) using our new local algorithm and avoids the use of iterative refinement in~\cite{anand2025unbreakable} by exploiting the concept of downward well-linked sets introduced by~\cite{korhonen2025linear}.

We believe that the ideal and final goal in this line of work would be to obtain a $k$-lean decomposition in time $k^{O(k)}m$ or even $2^{O(k)} m$. We believe that the local algorithm introduced in this paper has the potential to be useful towards such a goal and is likely to find applications beyond obtaining this particular result for unbreakable decomposition. We discuss one reason for this below.

\paragraph{Addressing a key bottleneck: Finding witness cuts in unbreakable graphs.}
On a high level, one of the main conceptual bottlenecks in~\cite{korhonen2025linear}, which explains the $k^{O(k^2)}$ in the running time, is solving the following problem. In this problem, we are given a $(q,k)$ unbreakable graph $G$ - where for every cut $(L,R)$ of size $|L \cap R| \leq k$, we must have $|L| \leq q$ or $|R| \leq q$. We are also given a terminal set $T \subseteq V(G)$, and we want to find a \emph{witness cut} $(L',R')$ - a cut having at least $k + 1$ terminals on both sides of the cut, so that $|L' \cap T| \geq k + 1, |R' \cap T| \geq k + 1$ and cut-size $|L' \cap R'| \leq k$.~\cite{korhonen2025linear} employs a local search algorithm that works by branching and solves this problem in $q^{O(k)}$ time. Such local search procedures are run with $q = k^{O(k)}$ which leads to the final running time. On the other hand, in such a setting, we can do better.

Our algorithm detects an important separator of volume $\nu$ and size at most $k$ between two nodes in time $\Otil(\nu 2^{O(k)})$. It is not too hard to show that if there exists a witness cut $(L',R')$ such that $G[L' \setminus R']$ is connected and $N_G(L' \setminus R')=L'\cap R'$, then there must be another witness cut $(L'',R'')$ that is a $v$-$X$ important separator, where $ v\in L' \setminus R'$, and $X \subseteq T \cap R'$ is some fixed set of terminals of size $k + 1$. Since we are dealing with a $(q,k)$  unbreakable graph, we also know that $|L''| \leq q$ or $|R''| \leq q$. For simplicity let us assume $|L''| \leq q$. Then the volume of the cut $(L'',R'')$, which is the total degree of vertices in $L'' \setminus R''$, is at most $q(q + k) \leq 2q^2$. Applying our local algorithm with volume parameter $\nu = 2q^2$, we would discover this cut in time $q^2 2^{O(k)}$, an improvement over $q^{O(k)}$.~\Cref{sec:disconnectedwitness} concretely implements this approach, removing the assumption of connectivity of $L' \setminus R'$, at the cost of $k^{O(k)}$ in the running time instead of $2^{O(k)}$.

Thus in principle, one could hope to convert the algorithm of~\cite{korhonen2025linear} to run in time $k^{O(k)}m\,\polylog (n)$. However, this is far from obvious due to the various steps which require careful treatment, many of which deal with disconnected cuts. For instance, some steps in~\cite{korhonen2025linear} require separations in tangle-unbreakable graphs, which are not necessarily unbreakable, where it is less clear if we can obtain a similar running time for finding witness cuts. Nevertheless, we believe that it is plausible that our local algorithm could prove helpful in this case too.

\paragraph{Single Sink Important Separators.} Consider the \emph{single-sink important separator} problem, where given a graph $G$, an integer $k$, and a sink vertex $t$, the goal is to enumerate $(v,t)$ important separators across \emph{all $v \in V(G)$}. We show an algorithm that achieves this, running in time $2^{O(k)}m\,\mathrm{polylog(n)}$ by using a simple reduction to local important separators. This sub-routine arises in many applications, including the first algorithm for unbreakable decomposition~\cite{cygan2014minimum}, and implicitly in our algorithm (see~\Cref{sec:connectedwitness}), and therefore may be of independent interest.

\begin{theorem}\label{thm:singlesink}
Given a graph $G$ and a designated vertex $t$, one can compute with high probability the list of all $(v,t)$ important separators of size at most $k$ for all $v \in V(G), v \neq t$ in time $2^{O(k)}m\,\polylog(n)$.

\end{theorem}

\paragraph{Edge Multicut.} The \textsc{Edge Multicut} problem is a fundamental problem in parameterized complexity. Here we are given a graph $G$, a set $\mathcal{T}$ of pairs $(s_i, t_i)$ and an integer $k$. The goal is to check if there is an edge cut of size at most $k$ in $G$ that can disconnect each pair $(s_i,t_i)$ simultaneously. Using our result on single sink important separators, we can speed up yet another important technique called \emph{shadow removal}, which is an important component of many algorithms, and in particular for \textsc{Multicut}. Together with a new $O(\log k)$ approximate rounding algorithm given an LP solution, recent developments in multicommodity flows~\cite{haeupler2024low}, and following other steps similar to recent work in~\cite{chu2026faster}, we get the first almost-linear time FPT algorithm for this problem.

\begin{theorem}
Given an instance $(G,\mathcal T,k)$ of \textsc{Edge Multicut}, let
\(
N=|V(G)|+|E(G)|+|\mathcal T|.
\)
Then there is a randomized algorithm that, with high probability,
decides whether $G$ has an edge multicut of size at most $k$ in
time
\(
f(k)N^{1+o(1)}.
\)

\end{theorem}

%% file: local_cut.tex
\section{\texorpdfstring{Local algorithm: Proof of \Cref{thm:breakable}.}{Local algorithm.}}\label{sec:local}

The following \Cref{lemma:LocalImpSep} is the main theorem of this section, see \Cref{sect:SectPrelimLocal} for some preliminaries. An adjacency list representation is $A$-respecting if, for every $v\ne s$, its list contains at most $A(v)$ entries, including entries for deleted endpoints.

\begin{restatable}{theorem}{LocalImportantSeparator}
\label{lemma:LocalImpSep}
Let $G$ be an undirected graph with source vertex $s$, sink vertex $t$, a node-weighting $A$ and an $A$-respecting adjacency list representation. Given a volume parameter $\nu$ and a cut-size parameter $k\geq 1$, there is an algorithm that outputs a collection ${\cal S}$ of $4^{k}$ many $(s,t)$-separators, satisfying that
\begin{enumerate}
\item each $(L,R)\in{\cal S}$ has size $\leq k$ and $A$-volume $\leq k\nu$, and
\item with high probability, ${\cal S}$ includes all important $(s,t)$-separators with size $\leq k$ and $A$-volume $\leq \nu$.
\end{enumerate}
The running time is $O(4^{k}\cdot \poly(k)\cdot\nu\cdot \log n)$.

\end{restatable}

Theorem~1.1 follows by taking $A(v)=\deg_G(v)$ and filtering the output list,
retaining only separators of volume at most $\nu$. This filtering does not
remove any important separator of volume at most $\nu$.

\subsection{Section Preliminaries.}
\label{sect:SectPrelimLocal}

\paragraph{Node-Weightings.} Consider an undirected graph $G$ with a source vertex $s$ and an adjacency list representation. A \emph{node-weighting $A$} is a function that assigns each vertex $v\in V(G)\setminus \{s\}$ a positive integer $A(v)$ and sets $A(s) = 0$ for the source vertex $s$. All node-weightings in this section are in fact \emph{degree-respecting}, meaning that $A(v)\geq \deg_{G}(v)$ for all $v\in V(G)\setminus \{s\}$. However, we do not state this explicitly each time, as this property is implicitly enforced through the notion of \emph{$A$-respecting adjacency list representations} introduced below.

\paragraph{Adjacency List Representations.} Because we discuss local algorithms in this section, it is important to specify the model to access the input graph. Concretely, we access an input graph $G$ via its \emph{adjacency list representation}. Normally, a standard adjacency list representation of $G$ simply stores, for each vertex $v\in V(G)$, the list of its incident edges. However, since the algorithm may modify the graph by contracting or deleting vertices and we still want to maintain the adjacency list representation efficiently, our actual definition is slightly more involved than the standard one. We note that this somewhat technical definition is not only for shaving polylog factors in the running time, but readers could skip it on a first reading for convenience.

Precisely, for an undirected graph $G$ with a specified source vertex $s$ and a node-weighting $A$, its adjacency list representation, denoted by ${\cal L}_{G}$, stores the following.
\begin{enumerate}
\item\label{Item:List1} For each vertex $v\in V(G)$, a doubly-linked list $\Gamma_{G}(v) = \{(v,u_{1}),(v,u_{2}),...\}$, which contains all $v$'s incident edges but may additionally include some garbage incident edges $(v,x)$ whose endpoints $x$ are not inside $V(G)$. Specially, the list $\Gamma_{G}(s)$ of the source vertex contains no garbage edge.
\item\label{Item:List2} A \emph{forbidden set} $X_{G}$ (without multiplicities) that collects the invalid endpoints $x$ of all garbage incident edges in all lists $\Gamma_{G}(v)$. This set is represented as an \emph{indicator array} (the IDs of such $x$ are moderately small as we will see, so this is affordable), allowing us to check in $O(1)$ time whether an edge $(v,u)$ in a list $\Gamma_{G}(v)$ is garbage or not.
\item\label{Item:List3} For each edge $(u,v)\in E(G)$, its two occurrences in $\Gamma_{G}(u)$ and $\Gamma_{G}(v)$ are associated with cross-pointers to each other.
\end{enumerate}

\Cref{Item:List3} is only used when we update the adjacency list after modifying the graph. \Cref{Item:List1} and \Cref{Item:List2} allow us to access the adjacency list representation as if it were the standard one, except that we may spend extra time skipping over garbage edges. Hence, we introduce the concept of \emph{$A$-respecting adjacency list representation ${\cal L}_{G}$}, meaning that $|\Gamma_{G}(v)|\leq A(v)$ for all $v\in V(G)\setminus \{s\}$. This allows us to use $A$ to bound access time to ${\cal L}_{G}$ when analyzing local algorithms. Note that the existence of an $A$-respecting adjacency list representation ${\cal L}_{G}$ implicitly implies that $A$ is degree-respecting.

\paragraph{Separators.} Throughout this section, all $(s,t)$-separators (in both lemma statements and proofs) are connected. To avoid clutter, we usually omit the qualifier \emph{connected}. Furthermore, all the connected $(s,t)$-separators $(L,R)$ mentioned in this section will ensure $N(L\setminus R) = L\cap R$. For two connected $(s,t)$-separators $(L^{+},R^{+})$ and $(L,R)$, we write $(L^{+},R^{+})\succeq (L,R)$ if $L^{+}\setminus R^{+}\supseteq L\setminus R$, and similarly write $(L^{+},R^{+})\succ (L,R)$ if $L^{+}\setminus R^{+}\supset L\setminus R$.

\begin{definition}[Important $(s,t)$-separators]%
Let $G$ be an undirected graph with source vertex $s$ and sink vertex $t$. An $(s,t)$-separator $(L,R)$ is important if it is connected, and there is no connected $(s,t)$-separator $(L^{+},R^{+})$ satisfying both $(L^{+},R^{+})\succ (L,R)$ and $|L^{+}\cap R^{+}|\leq |L\cap R|$.
\end{definition}

The above definition of important $(s,t)$-separators may seem different from the textbook definition, but in fact, they are equivalent.

\begin{definition}[Size and $A$-volume]Let $G$ be an undirected graph with source vertex $s$, sink vertex $t$ and node-weighting $A$. For an $(s,t)$-separator $(L,R)$, its size is $|L\cap R|$ and its $A$-volume is $A(L\setminus R)$.
\end{definition}

\paragraph{The Cut-Advanced Graph and The Cut-Vertex-Fixed Graph.} Let $G$ be an undirected graph with source vertex $s$, sink vertex $t$, and a node-weighting $A$. Fix an $(s,t)$-separator $(L,R)$ in $G$. We define the \emph{cut-advanced graph} $G_{1}$ and the \emph{cut-vertex-fixed graph} $G_{2}$, as well as their corresponding node-weightings, source vertices and sink vertices, as follows. Let $x$ be the vertex in $L\cap R$ with the smallest id.  
\begin{itemize}
\item The cut-advanced graph $G_{1}$ is obtained from $G$ by contracting $L\setminus R\cup\{x\}$ into a single vertex $s_{1}$, and then taking the corresponding simple graph. The source vertex is $s_{1}$, and the sink vertex remains to be $t$. The node-weighting $A_{1}$ of $G_{1}$ is the same as $A$ at each $v_{1}\in V(G_{1})\setminus \{s_{1}\}$, and we set $A_{1}(s_{1}) = 0$.
\item The cut-vertex-fixed graph $G_{2}$ is obtained from $G$ by contracting $L\setminus R$ into a single vertex $s_{2}$, deleting $x$, and then taking the corresponding simple graph. The source vertex is $s_{2}$, and the sink vertex remains to be $t$. The node-weighting $A_{2}$ of $G_{2}$ is the same as $A$ at each $v_{2}\in V(G_{2})\setminus \{s_{2}\}$, and we set $A_{2}(s_{2}) = 0$.
\end{itemize}

\paragraph{Computing Adjacency List Representations of $G_{1},G_{2}$.} 
Next, for each $i\in \{1,2\}$, we discuss how to obtain an $A_{i}$-respecting adjacency list representation of $G_{i}$ from an $A$-respecting adjacency list representation of $G$.

\medskip

\noindent\underline{The Adjacency List Representation of $G_{1}$.} We first ignore the vertex $s_{1}$ and its incident edges. Then for each vertex $v_{1}\in V(G_{1})\setminus \{s_{1}\}$, its incident-edge list $\Gamma_{G_{1}}(v_{1})$ can be obtained from $\Gamma_{G}(v_{1})$ by removing all incident edges whose other endpoint lies in $L\setminus R\cup\{x\}$. To do this, we can scan through all edges incident to $L\setminus R\cup\{x\}$ by accessing $\Gamma_{G}(u)$ for all $u\in L\setminus R\cup\{x\}$. Once an edge $(u,v)\in \Gamma_{G}(u)$ with $v\in V(G_{1})\setminus \{s_{1}\} = V(G)\setminus (L\setminus R\cup\{x\})$ is found, use the cross-pointer to locate $(u,v)\in \Gamma_{G}(v)$ and remove it. This process clearly takes time $O(\sum_{v\in L\setminus R\cup\{x\}}|\Gamma_{G}(v)|)$.

The second step is to add each incident edge $(s_{1},v_{1})$ of $s_{1}$ to both $\Gamma_{G_{1}}(s_{1})$ and $\Gamma_{G_{1}}(v_{1})$ ($\Gamma_{G_{1}}(s_{1})$ is empty initially). To do this, we only need to know $N_{G_{1}}(s_{1})$, which is exactly $N_{G}(L\setminus R\cup\{x\})$, and it can be computed in $O(\sum_{v\in L\setminus R\cup\{x\}}|\Gamma_{G}(v)|)$ time by accessing $\Gamma_{G}(u)$ for all $u\in L\setminus R\cup\{x\}$.

We can easily observe that the above adjacency list representation of $G_{1}$ is $A_{1}$-respecting: for each $v_{1}\in V(G_{1})\setminus \{s_{1}\}$, if its list $\Gamma_{G_{1}}(v_{1})$ receives a new edge in the second step, it must lose at least one edge in the first step (also note that $A_{1}(v_{1}) = A(v_{1})\geq |\Gamma_{G}(v_{1})|$)

The running time for obtaining the adjacency list representation of $G_{1}$ from that of $G$ is
\[
O\bigl(\sum_{v\in L\setminus R\cup\{x\}}|\Gamma_{G}(v)|\bigr) \leq O(|\Gamma_{G}(s)| + A(L\setminus R\cup\{x\}))\leq O(|L\cap R| + A(L\setminus R\cup\{x\})).
\]
The first inequality is because $A$ is a list-capacity node-weighting (but $A(s) = 0$). The second inequality is because $|\Gamma_{G}(s)|= |N_{G}(s)|\leq |L\setminus R| + |L\cap R|\leq A(L\setminus R\cup\{x\}) + |L\cap R|$ (recall that the list $\Gamma_{G}(s)$ has no garbage incident edge).

\medskip

\noindent\underline{The Adjacency List Representation of $G_{2}$.} Recall that $G_{2}$ is obtained from $G$ by contracting $L\setminus R$ into $s_{2}$ and then removing $x$. The adjacency list representation right after contracting $L\setminus R$ can be obtained in a similar way in time $O(|L\cap R|+A(L\setminus R))$. To further handle the removal of $x$, we simply put $x$ into the forbidden set $X_{G_{2}}$.

\paragraph{Projections Between Separators in $G$ and $G_{1},G_{2}$.} For both $i\in\{1,2\}$, we define the \emph{projection} and \emph{back-projection} between $(s,t)$-separators in $G$ and $(s_{i},t)$-separators in $G_{i}$ in the following natural way.

\medskip

\noindent\underline{Between separators in $G$ and $G_{1}$.}
Let ${\cal S}_{1}$ collect all $(s,t)$-separators $(L^{+},R^{+})$ in $G$ such that $(L^{+},R^{+})\succeq (L,R)$ and $x\in L^{+}\setminus R^{+}$. For each $(L^{+},R^{+})\in {\cal S}_{1}$, it can be projected onto $G_{1}$ by contracting the subset $L\setminus R\cup\{x\}\subseteq L^{+}\setminus R^{+}$ into the vertex $s_{1}$. For each $(s_{1},t)$-separator $(L_{1},R_{1})$, it can be back-projected onto $G$ by expanding the vertex $s_{1}\in L_{1}\setminus R_{1}$ back to $L\setminus R\cup\{x\}$. The projection and back-projection establish a bijection between $(s,t)$-separators $(L^{+},R^{+})\in {\cal S}_{1}$ in $G$ and $(s_{1},t)$-separators $(L_{1},R_{1})$ in $G_{1}$.

\begin{observation}
\label{ob:project1}
For each $(L^{+},R^{+})\in{\cal S}_{1}$ and its projected $(s_{1},t)$-separator $(L_{1},R_{1})$ in $G_{1}$, they have the same size, and $A(L^{+}\setminus R^{+}) = A(L\setminus R\cup\{x\}) + A_{1}(L_{1}\setminus R_{1})$.
\end{observation}

\noindent\underline{Between separators in $G$ and $G_{2}$.} Let ${\cal S}_{2}$ collect all $(s,t)$-separators $(L^{+},R^{+})$ in $G$ such that $(L^{+},R^{+})\succeq (L,R)$ and $x\in L^{+}\cap R^{+}$. For each $(L^{+},R^{+})\in {\cal S}_{2}$, it can be projected onto $G_{2}$ by contracting the subset $L\setminus R\subseteq L^{+}\setminus R^{+}$ into the vertex $s_{2}$ and removing the vertex $x$ from $L^{+}\cap R^{+}$. For each $(s_{2},t)$-separator $(L_{2},R_{2})$, it can be back-projected onto $G$ by expanding the vertex $s_{2}\in L_{2}\setminus R_{2}$ back to $L\setminus R$, and adding back $x$ into the cut. Again, this establishes a bijection between $(L^{+},R^{+})\in {\cal S}_{2}$ and $(s_{2},t)$-separators $(L_{2},R_{2})$ in $G_{2}$.

\begin{observation}
\label{ob:project2}
For each $(L^{+},R^{+})\in{\cal S}_{2}$ and its projected $(s_{2},t)$-separator $(L_{2},R_{2})$ in $G_{2}$, the size of $(L^{+},R^{+})$ is $1$ larger than the size of $(L_{2},R_{2})$, and $A(L^{+}\setminus R^{+}) = A(L\setminus R) + A_{2}(L_{2}\setminus R_{2})$.
\end{observation}

\subsection{Local Cuts.}

We will use the following local cuts subroutine from \cite{forster2020computing}.

\begin{lemma}[Local Cuts, modified Corollary 4.2 in \cite{forster2020computing}]%
\label{lemma:LocalCuts}%
    Let $G$ be an undirected graph with source vertex $s$, sink vertex $t$, a node-weighting $A$ and an $A$-respecting adjacency list representation ${\cal L}_{G}$. Given a volume parameter $\nu$ and a cut-size parameter $\lambda\geq 0$, there is an algorithm that outputs either
    \begin{enumerate}
    \item\label{Item:Localcut1} an $(s,t)$-separator $(L, R)$ with size $\leq \lambda$ and $A$-volume $\leq \lambda\nu$ if $\lambda \geq 1$, and $\leq \nu$ if $\lambda = 0$, or
    \item\label{Item:Localcut2} a symbol $\bot$,
    \end{enumerate}
    with the following guarantee:
    if there exists an $(s,t)$-separator with 
    size $\leq \lambda$ and $A$-volume $\leq \nu$,
    then with high probability, the algorithm outputs an $(L,R)$ in \Cref{Item:Localcut1}. %
    The algorithm always runs in $O(\nu(\lambda + 1)^{2}\log n)$ time. If the output is a separator $(L,R)$, the running time can also be bounded by $O((A(L\setminus R)+\lambda + 1)(\lambda + 1)^{2}\log n)$\footnote{The $+1$ terms are only to handle the corner case $\lambda = 0$}.
\end{lemma}

We make some remarks on the difference between \Cref{lemma:LocalCuts} and the original Corollary 4.2 in \cite{forster2020computing} (in addition to those tiny differences, for example we talk about $A$-volumes and have one more $\log n$ factor to ensure high probability, which can be easily handled). 

\medskip

\noindent\underline{Difference 1.}
First, our local cuts algorithm aims to detect $(s,t)$-separators, while the algorithm in \cite{forster2020computing} is for detecting $(s,\cdot)$-separator (i.e., allow an arbitrary sink). This is not a fundamental difference. The algorithm in \cite{forster2020computing} roughly keeps guessing a vertex $v$ outside the target cut from a large set of candidate vertices generated by a DFS starting from $s$ (and then reverse the $s$-$v$ path). We can simply take $v=t$ once the DFS search reaches $t$, in order to guarantee that the returned separator separates $s$ and $t$. 

\medskip

\noindent\underline{Difference 2.}
Second, we ensure that any separator returned by \Cref{lemma:LocalCuts} has $A$-volume $\leq \lambda\nu$, while this is not explicitly stated in Corollary 4.2 in \cite{forster2020computing}. In \cite{forster2020computing}, they obtain a local vertex cut algorithm by a simple reduction from a local \emph{directed} edge cut algorithm, and they give volume bounds for the latter (intuitively, since the algorithm runs in local time and $(L,R)$ is returned by explicitly outputting $L\setminus R$ and $L\cap R$, the $A$-volume of $L\setminus R$ can thus be bounded). Also, this volume bound is not very crucial for our later algorithms, since we can simply reject all separators with volume larger than $\nu$.

\medskip

\noindent\underline{Difference 3.} \cite{forster2020computing} only guarantees $O(\nu\lambda^{2})$ running time for $\lambda \geq 1$. To additionally guarantee a running time $O((A(L\setminus R) + \lambda)\lambda^{2})$, we can simply run the $O(\nu\lambda^{2})$-time version for $\nu' \in \{2,4,8,...,\nu\}$, and stop once the algorithm returned a separator.

\subsection{Local Furthest Min-Cuts.}

To understand the following local furthest min-cut subroutine better, we point out that we generally invoke this subroutine in the setting where, among $(s,t)$-separators with $A$-volume $\leq\nu$, the minimum-size one has size exactly $\lambda$. We are expecting an $(s,t)$-separator $(L,R)$ that is pushed furthest with $A$-volume $\leq \nu$.

\begin{lemma}[Local Furthest Min-Cuts]%
\label{lemma:LocalMincuts}%
    Let $G$ be an undirected graph with source vertex $s$, sink vertex $t$, a list-capacity node-weighting $A$, and an $A$-respecting adjacency list representation. Given a volume parameter $\nu$ and a cut-size parameter $\lambda\geq 0$, there is an algorithm $\LocalMincut(G,A,s,t,\nu,\lambda)$ which outputs either 
    \begin{enumerate}
    \item\label{Item:LocalMincut1} an $(s,t)$-separator $(L, R)$ with size $\leq \lambda$ and $A$-volume $\leq \lambda\nu$ if $\lambda \geq 1$, and $\leq \nu$ if $\lambda = 0$, or
    
    \item\label{Item:LocalMincut2} a symbol $\bot$,
    \end{enumerate}
    with the following guarantee:
        if there exists an $(s,t)$-separator with 
        $A$-volume $\leq \nu$
        and size $\leq \lambda$, then with high probability, the output is an $(L,R)$ in \Cref{Item:LocalMincut1} further satisfying that either
        \begin{itemize}
        \item it has size $<\lambda$ or $A$-volume $>\nu$, or
        \item any $(s,t)$-separator $(L^{+},R^{+})\succ (L,R)$ with $A$-volume $\leq \nu$ has size $\geq \lambda+1$.
        \end{itemize}
    The algorithm runs in $\nu\cdot\poly(\lambda)\cdot \log n$ time.
\end{lemma}

\paragraph{The Algorithm.} The algorithm $\LocalMincut(G,A,s,t,\nu,\lambda)$ is as follows.
\begin{enumerate}
\item\label{Step:LocalMincut1} Invoke $\LocalCut(G,A,s,t,\nu,\lambda)$. If the output is the symbol $\bot$, return this symbol.

\item\label{Step:LocalMincut2} Otherwise, let $(L,R)$ be the $(s,t)$-separator returned by $\LocalCut(G,A,s,t,\nu,\lambda)$. If $(L,R)$ has size $<\lambda$ or $A$-volume $>\nu$, return the same $(L,R)$ immediately. Therefore, from now, $(L,R)$ has size exactly $\lambda$ and $A$-volume $\leq \nu$. If $\lambda = 0$, we again return $(L,R)$ directly.

\item\label{Step:LocalMincut3} Let $x$ be the vertex in $L\cap R$ with the smallest id (when defining the cut-advanced graph). 

If $\nu - A(L\setminus R\cup\{x\})<0$, skip this step. 

Otherwise, let $G_{1},A_{1},s_{1},t$ be the cut-advanced graph and its corresponding node-weighting, source vertex and sink vertex. Let $\nu_{1} =\nu - A(L\setminus R\cup\{x\})$ and invoke 
\[
\LocalMincut(G_{1},A_{1},s_{1},t,\nu_{1},\lambda).
\]
If the output is an $(s_{1},t)$-separator in $G_{1}$, return its back-projected $(s,t)$-separator in $G$. If not, then  proceed.

\item\label{Step:LocalMincut4} 
If $\lambda = 1$, return $(L,R)$.

Otherwise, let $G_{2},A_{2},s_{2},t$ be the cut-vertex-fixed graph and its corresponding node-weighting, source vertex and sink vertex. Let $\nu_{2} = \nu - A(L\setminus R)$ and invoke 
\[
\LocalMincut(G_{2},A_{2},s_{2},t,\nu_{2},\lambda - 1).
\]
If the output is an $(s_{2},t)$-separator in $G_{2}$, return its back-projected $(s,t)$-separator in $G$, otherwise return $(L,R)$.

\end{enumerate}

\paragraph{Correctness.} First, it is relatively easy to see that, if the algorithm returns an $(s,t)$-separator $(L,R)$, it must have size $\leq \lambda$ and $A$-volume $\leq\lambda\nu$. This follows from the guarantee of the subroutine $\LocalCut$ given in \Cref{lemma:LocalCuts}, together with \Cref{ob:project1} and \Cref{ob:project2} on the back-projection of separators.

It remains to prove that the key guarantee holds for any input instance $(G,A,s,t,\nu,\lambda)$. For simplicity, we assume that, when running $\LocalMincut$ for this input instance, the high-probability event in \Cref{lemma:LocalCuts} happens for all calls to the $\LocalCut$ subroutine. By the following \Cref{claim:LocalMincutRecursiveSteps}, the number of calls to $\LocalCut$ is bounded by $\poly(\nu,\lambda)$, so this assumption holds with high probability. Therefore, it suffices to prove \Cref{lemma:LocalMincutCorrectness} under this assumption.

\begin{claim}
\label{claim:LocalMincutRecursiveSteps}
For any input instance $(G,A,s,t,\nu,\lambda)$ to $\LocalMincut$, the number of calls to $\LocalCut$ is at most $\nu + 2\lambda + 1$. 
\end{claim}
\begin{proof}
The key observation is that, if the algorithm does not return at Step \ref{Step:LocalMincut3}, i.e. the recursive call to $\LocalMincut$ returns the symbol $\bot$ (called a \emph{failed trial}), this recursive call $\LocalMincut$ only calls $\LocalCut$ once. In other words, the recursion tree of $\LocalMincut$ is roughly a path: if the recursive call to $\LocalMincut$ at Step \ref{Step:LocalMincut3} is a failed trial, the algorithm essentially branches into Step \ref{Step:LocalMincut4} with one additional call to $\LocalCut$ in the failed trial; otherwise, the algorithm just branches into Step \ref{Step:LocalMincut3}. The number of Step-\ref{Step:LocalMincut3} branchings (excluding the failed trials) is $\leq\nu$ (each decreases $\nu$ by at least $1$), and the number of Step-\ref{Step:LocalMincut4} branchings is $\lambda$, so the number of calls to $\LocalCut$ is at most $\nu+2\lambda + 1$ (the additive $1$ is just to handle the case when $\lambda = 0$).
\end{proof}

\begin{lemma}
\label{lemma:LocalMincutCorrectness}
For each input instance $(G,A,s,t,\nu,\lambda)$, if 
\begin{enumerate}[label=(\Alph*)]
\item\label{Condition:LocalMincutCorrectness} there exists an $(s,t)$-separator with $A$-volume $\leq \nu$ and size $\leq \lambda$,
\end{enumerate}
then $\LocalMincut(G,A,s,t,\nu,\lambda)$ returns an $(s,t)$-separator $(L,R)$ satisfying either it has size $<\lambda$ or $A$-volume $>\nu$, or any $(s,t)$-separator $(L^{+},R^{+})\succ(L,R)$ with $A$-volume $\leq \nu$ has size $\geq \lambda+1$.
\end{lemma}
\begin{proof}
Consider such an instance $(G,A,s,t,\nu,\lambda)$. We inductively assume that the statement holds for the two recursive instances in Steps \ref{Step:LocalMincut3} and \ref{Step:LocalMincut4}, and now we prove the statement for this instance $(G,A,s,t,\nu,\lambda)$.

First, the algorithm does not return at Step \ref{Step:LocalMincut1}, because there exist $(s,t)$-separators with size $\leq \lambda$ and $A$-volume $\leq \nu$ by Condition \ref{Condition:LocalMincutCorrectness} and because of the guarantee of \Cref{lemma:LocalCuts}. Second, if the algorithm returns at Step \ref{Step:LocalMincut2}, then either the returned $(L,R)$ has size $<\lambda$ or $A$-volume $>\nu$, as desired, or $\lambda = 0$. In the case $\lambda = 0$, $|L \cap R| = 0$, and in this case $(L,R)$ is the unique furthest $(s,t)$ min-cut.  %

From now, the $(L,R)$ returned by $\LocalCut$ has size exactly $\lambda$ and $A$-volume $\leq \nu$, and we will prove that the algorithm returns correctly at either Step \ref{Step:LocalMincut3} or Step \ref{Step:LocalMincut4}. We consider two cases.

\medskip

\noindent{\underline{Case 1.}} Suppose there exists an $(s,t)$-separator $(L^{+},R^{+})\succeq (L,R)$ in $G$ with size $\leq \lambda$ and $A$-volume $\leq \nu$, satisfying that $x\in L^{+}\setminus R^{+}$. By \Cref{ob:project1}, its projected $(s_{1},t)$-separator $(L_{1},R_{1})$ in $G_{1}$ has size $\leq \lambda$ and $A_{1}$-volume $\leq \nu - A(L\setminus R\cup\{x\})$. Hence Condition \ref{Condition:LocalMincutCorrectness} holds for the recursive call $\LocalMincut$ at Step \ref{Step:LocalMincut3}. By the induction hypothesis, it returns an $(s_{1},t)$-separator $(L_{1}^{\star},R_{1}^{\star})$ satisfying either it has size $< \lambda$ or $A_{1}$-volume $> \nu_{1}$, or any $(s_{1},t)$-separator $(L^{+}_{1},R^{+}_{1})\succ (L^{\star}_{1},R^{\star}_{1})$ with $A_{1}$-volume $\leq \nu_{1}$ has size $\geq \lambda +1$. By \Cref{ob:project1}, the back-projected $(s,t)$-separator in $G$ of $(L^{\star}_{1},R^{\star}_{1})$ is what we desire.

\medskip 

\noindent{\underline{Case 2.}} Suppose Case 1 does not hold. Equivalently, all $(s,t)$-separators $(L^{+},R^{+})\succeq (L,R)$ in $G$ with size $\leq \lambda$ and $A$-volume $\leq \nu$ satisfy that $x\in L^{+}\cap R^{+}$. We first get rid of two corner subcases.
\begin{itemize}
\item If the algorithm still returns an $(s,t)$-separator $(L^{\star},R^{\star})$ in $G$ at Step \ref{Step:LocalMincut3}, we know it has size $\leq \lambda$, $(L^{\star},R^{\star})\succeq (L,R)$ and $x\in L^{\star}\setminus R^{\star}$. Since Case 1 does not hold, $(L^{\star},R^{\star})$ has $A$-volume $>\nu$, so it is a correct output.
\item Suppose the algorithm returns $(L,R)$ at Step \ref{Step:LocalMincut4} because $\lambda = 1$. Since $\lambda = 1$, any $(s,t)$-separator $(L^{+},R^{+})\succ (L,R)$ with size $1$ satisfies $x\in L^{+}\setminus R^{+}$. This means there is no $(L^{+},R^{+})\succ (L,R)$ with size $1$ and $A$-volume $\leq\nu$, because otherwise we will go to Case 1. Hence, $(L,R)$ is a feasible output.
\end{itemize}

From now, suppose the algorithm returns at the end of Step \ref{Step:LocalMincut4}.
First, note that $(L,R)$ is an $(s,t)$-separator in $G$ with size $\leq\lambda$ and $A$-volume $\leq \nu$, satisfying that $(L,R)\succeq (L,R)$ and $x\in L\cap R$. So its projected $(s_{2},t)$-separator in $G_{2}$ implies that Condition \ref{Condition:LocalMincutCorrectness} holds for the recursive call $\LocalMincut$ at Step \ref{Step:LocalMincut4}. By the induction hypothesis, the recursive call returns an $(s_{2},t)$-separator $(L^{\star}_{2},R^{\star}_{2})$ satisfying either it has size $<\lambda-1$ or $A_{2}$-volume $> \nu_{2}$, or any $(s_{2},t)$-separator $(L^{+}_{2},R^{+}_{2})\succ (L^{\star}_{2},R^{\star}_{2})$ with $A_{2}$-volume $\leq \nu_{2}$ has size $\geq \lambda$. If the former holds, then its back-projected $(s,t)$-separator $(L^{\star},R^{\star})$ in $G$ has size $<\lambda$ or $A$-volume $>\nu$ by \Cref{ob:project2}, exactly what we want. Now assume the latter holds, and we will prove that any $(s,t)$-separator $(L^{+},R^{+})\succ(L^{\star},R^{\star})$ with $A$-volume $\leq \nu$ has size $\geq \lambda+1$, as desired.

Assume for contradiction that there exists an $(s,t)$-separator $(L^{+},R^{+})\succ (L^{\star},R^{\star})$ with $A$-volume $\leq \nu$ and size $\leq \lambda$. Because Case 1 does not hold, we know $(L^{+},R^{+})\succeq (L,R)$ and $x\in L^{+}\cap R^{+}$, it can be projected onto an $(s_{2},t)$-separator $(L^{+}_{2},R^{+}_{2})\succeq (L^{\star}_{2},R^{\star}_{2})$ in $G_{2}$, which has size $\leq \lambda-1$ and $A_{2}$-volume $\leq \nu_{2}$. However, this contradicts the property of $(L^{\star}_{2},R^{\star}_{2})$.

\end{proof}

\paragraph{The Running Time.} We follow the terms and arguments in the proof of \Cref{claim:LocalMincutRecursiveSteps}. Consider a (non-failed-trial) recursive node in the recursion tree. If it branches into the recursive call at Step \ref{Step:LocalMincut3}, the running time spent on this recursive step is 
\[
\underbrace{O((A(L\setminus R) + \lambda)\lambda^{2}\cdot \log n)}_{\text{$\LocalCut$ at Step \ref{Step:LocalMincut1}}} + \underbrace{O(\lambda + A(L\setminus R\cup \{x\}))}_{\text{Adjacency list of $G_{1}$}}\leq O((A(L\setminus R\cup\{x\})\cdot\lambda^{2} + \lambda^{3})\cdot \log n).
\]
If it branches into the recursive call at Step \ref{Step:LocalMincut4}, the running time of this step is
\[
\underbrace{O((A(L\setminus R)+\lambda)\lambda^{2}\cdot \log n)}_{\text{$\LocalCut$ at Step \ref{Step:LocalMincut1}}} + \underbrace{O(\nu\lambda^{2}\cdot\log n)}_{\text{The failed trial}} + \underbrace{O(\lambda + A(L\setminus R))}_{\text{Adjacency list of $G_{2}$}}\leq O((\nu\lambda^{2}+\lambda^{3})\cdot \log n).
\]
Because every Step-\ref{Step:LocalMincut3} branching reduces the volume parameter by $A((L\setminus R)\cup\{x\})$, the summation of the terms $A(L\setminus R\cup\{x\})$ over all Step-\ref{Step:LocalMincut3} branchings is at most $\nu$. Furthermore, the number of Step-\ref{Step:LocalMincut3} and Step-\ref{Step:LocalMincut4} branchings is at most $\nu$ and $\lambda$ respectively. The total running time is 
\[
O((\nu\lambda^{2} + \nu\lambda^{3})\cdot\log n) + O((\nu\lambda^{2}+\lambda^{3})\cdot \lambda\cdot\log n) = O((\nu\lambda^{3} + \lambda^{4})\cdot \log n).
\]
We replace $\lambda$ by $\lambda + 1$ in the final running time bound to account for the corner case $\lambda = 0$: clearly this case involves only one call to $\LocalCut$.

\subsection{Local Important Separators.}

We prove our local important separator result now.

\LocalImportantSeparator*

\paragraph{The Algorithm.} We describe an algorithm $\LocalImpSep(G,A,s,t,\nu,k,\lambda)$ that receives an additional cut-size lower-bound parameter $0\leq \lambda\leq k+1$.
We will see that $\LocalImpSep(G,A,s,t,\nu,k,1)$ is the algorithm needed by \Cref{lemma:LocalImpSep}. 

The algorithm $\LocalImpSep(G,A,s,t,\nu,k,\lambda)$ is as follows.
\begin{enumerate}
\item\label{Step:LocalImpSep0} 
If $\lambda>k$, return an empty ${\cal S}$. 

Invoke $\LocalMincut(G,A,s,t,\nu,\lambda)$. 
\item\label{Step:LocalImpSep1} If the output of $\LocalMincut(G,A,s,t,\nu,\lambda)$ is the symbol $\bot$, invoke 
\[
\LocalImpSep(G,A,s,t,\nu,k,\lambda+1)
\]
and return its output. 
\item\label{Step:LocalImpSep2} Otherwise, let $(L,R)$ be the $(s,t)$-separator returned by $\LocalMincut(G,A,s,t,\nu,\lambda)$. If $(L,R)$ has size $<\lambda$ or $A$-volume $>\nu$, 
return an empty set. Therefore, from now on $(L,R)$ has size exactly $\lambda$ and $A$-volume $\leq \nu$. If $\lambda = 0$, simply return $\{(L,R)\}$.
\item\label{Step:LocalImpSep4} Let $x$ be the vertex in $L\cap R$ with the smallest id (when defining the cut-advanced graph). 

If $\nu - A(L\setminus R\cup\{x\})<0$, skip this step. 

Otherwise, let $G_{1},A_{1},s_{1},t$ be the cut-advanced graph and its corresponding node-weighting, source vertex and sink vertex. Let $\nu_{1} =\nu - A(L\setminus R\cup\{x\})$ and invoke
\[
\LocalImpSep(G_{1},A_{1},s_{1},t,\nu_{1},k,\lambda+1).
\]
Let ${\cal S}_{1}$ denote the return of this recursive call.

\item\label{Step:LocalImpSep5} Let $G_{2},A_{2},s_{2},t$ be the cut-vertex-fixed graph and its corresponding node-weighting, source vertex and sink vertex. Let $\nu_{2} = \nu - A(L\setminus R)$ and invoke 
\[
\LocalImpSep(G_{2},A_{2},s_{2},t,\nu_{2},k-1,\lambda-1).
\]
Let ${\cal S}_{2}$ denote the return of this recursive call.

\item \label{Step:LocalImpSep6} Let ${\cal S}$ collect $(s,t)$-separators in $G$ back-projected from separators in both ${\cal S}_{1}$ and ${\cal S}_{2}$. Return ${\cal S}$.

\end{enumerate}

\paragraph{Correctness.} It is relatively easy to see that all $(s,t)$-separators in the returned ${\cal S}$ have size $\leq k$ and $A$-volume $\leq k\nu$, which follows from the guarantee of \Cref{lemma:ImpSepCorrectness} together with \Cref{ob:project1} and \Cref{ob:project2} on the back-projection of separators. 

It remains to show that ${\cal S}$ includes all important $(s,t)$-separators with size $\leq k$ and $A$-volume $\leq\nu$ with high probability. For simplicity, we assume that, when running $\LocalImpSep(G,A,s,t,\nu,k,\lambda)$, the high-probability event in \Cref{lemma:LocalMincuts} happens for all calls to the $\LocalMincut$ subroutine. Since the number of calls to $\LocalMincut$ is at most $O(4^{k})$ by \Cref{claim:LocalImpSepStepNumber}, this assumption holds with high probability (after we boost the correct probability of \Cref{lemma:LocalMincuts} by paying an extra $k$ factor in its running time). Therefore, it suffices to prove \Cref{lemma:ImpSepCorrectness} under this assumption.

\begin{claim}
\label{claim:LocalImpSepStepNumber}
When running $\LocalImpSep$ for any input instance $(G,A,s,t,\nu,k,\lambda)$, the number of calls to $\LocalMincut$ is $2\cdot 4^{k}$.
\end{claim}
\begin{proof}
Note that in the recursion tree, each recursive node will either have a single branch at Step \ref{Step:LocalImpSep1} or two branches at Steps \ref{Step:LocalImpSep4} and \ref{Step:LocalImpSep5}. Every recursive call will reduce the quantity of $2k-\lambda$ by at least $1$. Since $\lambda \geq 0$, the recursion depth is at most $2k$ and the total number of recursive nodes is $2 \cdot 4^{k}$. 
\end{proof}

\begin{lemma}
\label{lemma:ImpSepCorrectness}
For each input instance $(G,A,s,t,\nu,k,\lambda)$, if 
\begin{enumerate}[label=(\Alph*)]
\item\label{Condition:ImpSepCorrectness} the graph $G$ has no $(s,t)$-separator with size $<\lambda$ and $A$-volume $\leq\nu$,
\end{enumerate}
then the returned ${\cal S}$ includes all important $(s,t)$-separators with size $\leq k$ and $A$-volume $\leq \nu$.
\end{lemma}
\begin{proof}
Consider such an instance $(G,A,s,t,\nu,k,\lambda)$ satisfying Condition \ref{Condition:ImpSepCorrectness}. We inductively assume that the statement holds for the recursive instances when running $\LocalImpSep(G,A,s,t,\nu,k,\lambda)$. Note that the correctness of the base cases which return at Step \ref{Step:LocalImpSep0} is straightforward. 

Before proceeding, we first introduce the auxiliary \Cref{claim:UseSubmodularity}, which essentially states that under Condition \ref{Condition:ImpSepCorrectness}, every important $(s,t)$-separator (with small $A$-volume) lies beyond any $(s,t)$-separator with small size.

\begin{claim}
\label{claim:UseSubmodularity}
Let $(L,R)$ be an $(s,t)$-separator with size $\leq \lambda$. Then every important $(s,t)$-separator $(L^{\star},R^{\star})$ with $A$-volume $\leq \nu$ satisfies that $(L^{\star},R^{\star})\succeq (L,R)$.
\end{claim}
\begin{proof}
Let $(L_{\cap},R_{\cap})$ be the intersection of $(L,R)$ and $(L^{\star},R^{\star})$, and let $(L_{\cup},R_{\cup})$ be their union. Assume for contradiction that $(L^{\star},R^{\star})\succeq (L,R)$ does not hold. Then, we know $(L_{\cup},R_{\cup})\succ (L^{\star},R^{\star})$. Furthermore, we know $(L_{\cap}, R_{\cap})$ has size $\geq \lambda$, because it has $A$-volume $\leq \nu$ and by Condition \ref{Condition:ImpSepCorrectness}. By submodularity, the size of $(L_{\cup},R_{\cup})$ is at most that of $(L^{\star},R^{\star})$, contradicting that $(L^{\star},R^{\star})$ is important. 
\end{proof}

The following \Cref{claim:Step2Correctness} shows that if the algorithm returns at Step \ref{Step:LocalImpSep2}, then this return is always correct.

\begin{claim}
\label{claim:Step2Correctness}
If the algorithm returns $\emptyset$ at Step \ref{Step:LocalImpSep2}, there are no important $(s,t)$-separators with $A$-volume $\leq \nu$.
\end{claim}
\begin{proof}
Recall that if the algorithm returns $\emptyset$ at Step \ref{Step:LocalImpSep2}, the $(s,t)$-separator $(L,R)$ returned by $\LocalMincut$ in Step \ref{Step:LocalImpSep0} satisfies
(1) it has size $\leq \lambda$ (guaranteed by \Cref{lemma:LocalMincuts}), and
(2) it has either size $<\lambda$ or $A$-volume $>\nu$ (guaranteed by Step \ref{Step:LocalImpSep2}).
Condition \ref{Condition:ImpSepCorrectness} says that there is no $(s,t)$-separator with size $<\lambda$ and $A$-volume $\leq \nu$. Hence $(L,R)$ always has size $\leq \lambda$ and $A$-volume $>\nu$. %
Now, by \Cref{claim:UseSubmodularity}, any important $(s,t)$-separator $(L^{\star},R^{\star})$ with $A$-volume $\leq \nu$ satisfies that $(L^{\star},R^{\star})\succeq (L,R)$. However, $(L^{\star},R^{\star})\succeq (L,R)$ implies that $(L^{\star},R^{\star})$ has $A$-volume $>\nu$, meaning that such $(L^{\star},R^{\star})$ does not exist.

\end{proof}
Next, if the algorithm returned $(L,R)$ at Step~\ref{Step:LocalImpSep2}, it is clearly the unique $s$-$t$ important separator, since $|L \cap R| = 0$.

We are now ready to begin our main argument. We consider the following two cases. The first case only branches into the recursive call in Step \ref{Step:LocalImpSep1}, while the second and more substantial case branches into two recursive calls in both Steps \ref{Step:LocalImpSep4} and \ref{Step:LocalImpSep5}.

\medskip

\noindent\underline{Case 1.} Suppose there is no $(s,t)$-separator with size $\leq \lambda$ and $A$-volume $\leq \nu$. If the algorithm returns at Step \ref{Step:LocalImpSep1} with the output of $\LocalImpSep(G,A,s,t,\nu,k,\lambda+1)$, then this output is exactly what we want because of the induction hypothesis and because Condition \ref{Condition:ImpSepCorrectness} holds for this recursive instance (by the assumption of Case 1).
Otherwise, the algorithm must return nothing at Step \ref{Step:LocalImpSep2} (by Condition \ref{Condition:ImpSepCorrectness} of the current instance), which is feasible by \Cref{claim:Step2Correctness}.

\medskip

\noindent\underline{Case 2.} Suppose there is an $(s,t)$-separator with size $\leq \lambda$ and $A$-volume $\leq \nu$. Then $\LocalMincut(G,A,s,t,\nu,\lambda)$ in Step \ref{Step:LocalImpSep0} returns an $(s,t)$-separator $(L,R)$ satisfying either
\begin{enumerate}
\item $(L,R)$ has size $<\lambda$ or $A$-volume $>\nu$, or 
\item\label{Item:CorrectProperty}
any $(s,t)$-separator $(L^{+},R^{+})\succ (L,R)$ with $A$-volume $\leq \nu$ has size at least $\lambda+1$. 
\end{enumerate}
We further assume that the algorithm does not return $\emptyset$ at Step \ref{Step:LocalImpSep2} (otherwise, the correctness is shown by \Cref{claim:Step2Correctness}), so $(L,R)$ has size exactly $\lambda$ and volume $\leq \nu$ and Property \ref{Item:CorrectProperty} must hold. If $\lambda = 0$, since $|L \cap R| = 0$, it is the unique $s$-$t$ important separator and hence we return exactly this one separator. Thus we now assume $\lambda > 0$. Now by \Cref{claim:UseSubmodularity}, every important $(s,t)$-separator $(L^{\star},R^{\star})$ with $A$-volume $\leq \nu$ satisfies that $(L^{\star},R^{\star})\succeq (L,R)$.

We classify the important $(s,t)$-separators $(L^{\star},R^{\star})$ with size $\leq k$ and $A$-volume $\leq \nu$ into two classes ${\cal S}^{\star}_{1}$ and ${\cal S}^{\star}_{2}$. The former includes all such $(L^{\star},R^{\star})$ with $x\in L^{\star}\setminus R^{\star}$, while the latter includes the rest (satisfying that $L^{\star}\cap R^{\star}\ni x$). The following \Cref{ob:projection} is an immediate consequence of our discussion on the projection of separators.

\begin{observation}
\label{ob:projection}
For each $i\in\{1,2\}$, the $(s_{i},t)$-separators projected from separators in ${\cal S}^{\star}_{i}$ are important $(s_{i},t)$-separators in $G_{i}$ with size $\leq k-i+1$ and $A_{i}$-volume $\leq \nu_{i}$.
\end{observation}

With \Cref{ob:projection}, we hope that the recursive calls in Steps \ref{Step:LocalImpSep4} and \ref{Step:LocalImpSep5} will return the small-volume important $(s_{1},t)$- and $(s_{2},t)$-separators in $G_{1}$ and $G_{2}$ respectively, and then the back-projection in Step \ref{Step:LocalImpSep6} will give us all small-volume important $(s,t)$-separators in $G$ as desired. The correctness of these two recursive calls is guaranteed by the induction hypothesis, as long as the recursive instances satisfy Condition \ref{Condition:ImpSepCorrectness}, i.e., the following \Cref{claim:InductionCondition}.

\begin{claim}
\label{claim:InductionCondition}
The graph $G_{1}$ has no $(s_{1},t)$-separator $(L_{1},R_{1})$ with size $<\lambda+1$ and $A_{1}$-volume $\leq \nu_{1}$. The graph $G_{2}$ has no $(s_{2},t)$-separator $(L_{2},R_{2})$ with size $<\lambda-1$ and $A_{2}$-volume $\leq\nu_{2}$.
\end{claim}
\begin{proof}
We first prove the statement for graph $G_{1}$. Assume for contradiction that there exists such an $(L_{1},R_{1})$. Let $(L^{+},R^{+})$ be the $(s,t)$-separator in $G$ from back-projecting $(L_{1},R_{1})$. By \Cref{ob:project1}, $(L^{+},R^{+})$ has size $<\lambda+1$ and $A$-volume $\leq\nu$. Because $(L^{+},R^{+})\succ (L,R)$, the existence of $(L^{+},R^{+})$ contradicts Property \ref{Item:CorrectProperty}.

For the statement for graph $G_{2}$, we again assume for contradiction that such an $(L_{2},R_{2})$ exists and consider its back-projection $(L^{+},R^{+})$ onto $G$. Similarly by \Cref{ob:project2}, we have $(L^{+},R^{+})$ has size $<\lambda$ and $A$-volume $\leq \nu$, and also $(L^{+},R^{+})\succeq (L,R)$. By Property \ref{Item:CorrectProperty} and the fact that $(L,R)$ has size exactly $\lambda$, any $(s,t)$-separator $(L^{+},R^{+})\succeq (L,R)$ with $A$-volume $\leq \nu$ has size at least $\lambda$, contradicting the existence of $(L^{+},R^{+})$.
\end{proof}
This concludes the proof of \Cref{lemma:ImpSepCorrectness}.
\end{proof}

\paragraph{The Running Time.} Consider a node in the recursion tree. Its running time, excluding the time spent in recursive calls and Step \ref{Step:LocalImpSep6}, can be bounded by
\[
\underbrace{O(\nu\cdot\poly(k)\cdot \log n)}_{\text{$\LocalMincut$ at Step \ref{Step:LocalImpSep0}}} + \underbrace{O(k + A(L\setminus R\cup\{x\}))}_{\text{Adjacency List of $G_{1}$}} + \underbrace{O(k + A(L\setminus R))}_{\text{Adjacency List of $G_{2}$}}\leq O(\nu\cdot\poly(k)\cdot \log n).
\]
By \Cref{claim:LocalImpSepStepNumber}, there are at most $O(4^{k})$ nodes in the recursion tree. Therefore, the total running time excluding Step \ref{Step:LocalImpSep6} is $O(\nu\cdot\poly(k)\cdot 4^{k}\cdot\log n)$. For the total running time of Step \ref{Step:LocalImpSep6} summed over all nodes, observe that each leaf node will generate at most one separator, contributing to at most $2k$ nodes, so the total time of Step \ref{Step:LocalImpSep6} is $O(k\cdot 4^{k}\cdot\nu)$ (note that back-projecting a separator simply takes $O(\nu+k)$ time, since we represent a separator $(L,R)$ by storing $L\setminus R$ and $L\cap R$).

%% file: connected_witness_elimination.tex
\section{Connected Witness Elimination.}\label{sec:connectedwitness}

The broad goal of this section, given a graph $G$, set $X$, and parameters $v,q,k$, is to eliminate all $(X,v,q,k)$ connected witnesses, as discussed in the overview.  We will use the local algorithm from~\Cref{sec:local} for this purpose, and the main theorem of this section is \Cref{thm:connected}. We start with a few definitions and useful lemmas before moving to the main theorem.

\begin{definition}[Witnesses and Connected Witnesses]%
\label{def:witness}%
Let $G$ be an undirected graph. Consider a vertex $v\in V(G)$, a vertex set $X\subseteq V(G)\setminus\{v\}$, and two parameters $q,k$.
A separator $(L,R)$ in $G$ is called a \emph{$(X,v,q,k)$-witness}, if 
\[
X\subseteq R,\qquad v\in L\setminus R,\qquad |L|\geq q+1,\text{and}\qquad |L\cap R|\leq k.
\]
Further consider a terminal set $T\subseteq V(G)$. 
A separator $(L,R)$ in $G$ is called an $(X,v,q,k)$-witness for $T$ in $G$, if 
\[
X\subseteq R,\qquad v\in L\setminus R,\qquad |L\cap T|\geq q+1,\text{and}\qquad |L\cap R|\leq k.
\]
A witness $(L,R)$ defined above is \emph{connected} if $G[L\setminus R]$ is connected and $N_G(L\setminus R)=L\cap R$.
\end{definition}

\begin{definition}[Volume]%
Let $G$ be an undirected graph. The \emph{volume} of a set $S \subseteq V(G)$ is the total degree of all the vertices in $S$, denoted by $\deg_{G}(S)$. 
The volume of a witness $(L, R)$ defined in \Cref{def:witness} is $\deg_G(L \setminus R)$. The subscript $G$ is omitted when $G$ is clear from the context. 

\end{definition}

Using~\Cref{thm:breakable}, we can find an $(X,v,q,k)$ connected witness of volume at most $\nu$, if one exists. 
This is formalized by the following \Cref{lemma:witness}.

\begin{lemma}\label{lemma:witness}
Let $G$ be an undirected graph. Given a vertex subset $X\subseteq V(G)$, a vertex $v\in V(G)\setminus X$, parameters $q,k$ and a volume parameter $\nu'$ satisfying that $|X|\geq k+1$, there is an algorithm that outputs either
\begin{itemize}
\item a connected $(X,v,q,k)$-witness of volume at most $\nu'$, or
\item a symbol $\bot$,
\end{itemize}
with the following guarantee: if (A) every connected $(X,v,q,k)$-witness has volume at most $\nu'$, then with high probability, the output is such a witness (if one exists) with maximum possible volume.
The algorithm runs in $2^{O(k)}\cdot \nu'$ time.

\end{lemma}

\begin{proof}
First, we obtain a graph $G'$ by creating an artificial vertex $t$ and adding edges between each $X$-vertex and $t$. 

Our goal is to show the correctness when Condition (A) holds. Fix a connected $(X,v,q,k)$-witness $(L,R)$ in $G$ (if one exists). Because $(L,R\cup\{t\})$ is a connected $(v,t)$-separator in $G'$ and by the definition of important separators, there must be a connected $(v,t)$-separator $(L'_{\star},R'_{\star})$ that is important in $G'$, such that $L\subseteq L'_{\star}$ and $|L'_{\star}\cap R'_{\star}|\leq |L\cap R| \leq k$.

Furthermore, we claim that $(L'_{\star},R'_{\star})$ has volume $\deg_{G'}(L'_{\star}\setminus R'_{\star})\leq \nu'$. To see this, observe that $(L'_{\star},R'_{\star}\setminus\{t\})$ is a connected $(X,v,q,k)$-witness in $G$ (we exploit $|X|\geq k+1$ here to avoid the corner case $L'_{\star} = V(G)$), which has volume $\deg_{G}(L'_{\star}\setminus (R'_{\star}\setminus \{t\}))\leq \nu'$ by Condition (A). Therefore, $\deg_{G'}(L'_{\star}\setminus R'_{\star}) = \deg_{G}(L'_{\star}\setminus (R'_{\star}\setminus \{t\}))\leq \nu'$ and we complete proving the claim.

Now, we run~\Cref{thm:breakable} on $G'$ with source vertex $v$, sink vertex $t$, cut-size parameter $k$, and volume parameter $\nu'$. The output is guaranteed to include $(L'_{\star},R'_{\star})$, and thus we can recover $(L_{\star},R_{\star}):=(L'_{\star},R'_{\star}\setminus\{t\})$, which is a connected $(X,v,q,k)$-witness in $G$, better than $(L,R)$ in the sense that $|L_{\star}\cap R_{\star}|\leq |L\cap R|$ and $L_{\star}\supseteq L$. Finally, since the output of \Cref{thm:breakable} will include all possible $(L'_{\star},R'_{\star})$, we just return the feasible $(L_{\star},R_{\star})$ with the maximum volume.

The running time is dominated by \Cref{thm:breakable}, which is $2^{O(k)}\cdot\nu'$.

\end{proof}

\begin{definition}[Downward well-linked]%
Given a graph $G$, a set $S \subseteq V(G)$ is downward well-linked, if every separator $(L, R)$ in $G[S \cup N_{G}(S)]$ has either $|L \cap R| \geq |L \cap N_G(S)|$ or $|L \cap R| \geq |R \cap N_G(S)|$.

\end{definition}

\begin{definition}
Given a graph $G$, a $(k,\nu)$-net is a set of vertices $Q \subseteq V(G)$ such that for every set of vertices $X$ of size at most $k$, $Q$ intersects every connected component of $G \setminus X$ with volume at least $\nu$.

\end{definition}

\begin{lemma}\label{lemma:sampleset}
Any graph $G$ with $m$ edges and $n$ vertices has a $(k, \nu)$-net of size $O(mk\log n/\nu)$, which can be obtained with high probability by simple random sampling of vertices proportional to their degree.
\end{lemma}

\begin{proof}
We sample $100mk \log n/\nu$ vertices independently with probability proportional to the degree, and let $Q$ be the sampled vertices. Fix a set $S\subseteq V(G)$ of volume at least $\nu$, the probability that $Q$ hits a vertex in $S$ is at least $1- n^{-5k}$, by the Chernoff bound. Now, observe that the number of sets needed to be hit is at most $n^{k+1}$, since there are at most $n^{k}$ choices of $X$ (recall that $X\subseteq V(G)$ has size $\leq k$), and for each $X$, $G\setminus X$ has at most $n$ components. By the union bound, $Q$ is a $(k,\nu)$-net with probability at least $1-n^{-k}$.

\end{proof}

The following is the main theorem of this section.

\begin{theorem}
\label{thm:connected}
Let $G$ be an undirected graph. Given a vertex set $X\subseteq V(G)$ and two parameters $q,k\geq 1$ such that $|X|\geq k+1$, there is an algorithm that outputs another graph $G_{\final}$ satisfying the following.
\begin{enumerate}
\item\label{item:Connected1} $X\subseteq V(G_{\final})\subseteq V(G)$.
\item\label{item:Connected2} There exists a sequence of graphs $(G^{(0)}:=G, G^{(1)},..., G^{(t)}:=G_{\final})$ satisfying that, for each $0\leq i\leq t-1$, if we let $S_{i}:=V(G^{(i)})\setminus V(G^{(i + 1)})$, then %
\begin{enumerate}
\item $G^{(i+1)}$ is exactly $(G^{(i)}\setminus S_{i})\cup \Clique(N_{G^{(i)}}(S_{i}))$, and
\item $S_{i}$ is downward well-linked in $G^{(i)}$.
\item $|N_{G^{(i)}}(S_i)| \leq k$.
\end{enumerate}

\item\label{item:Connected3} The adhesion of $V(G_{\final})$ in $G$ is at most $k$.
\item\label{item:Connected4} With high probability, there is no connected $(X,v,q,k)$-witness in $G_{\final}$ with volume larger than $200\cdot k^{2}\cdot 8^{k}$, for any vertex $v\in V(G_{\final})$.
\end{enumerate}
The algorithm runs in $\tilde{O}(m\cdot 2^{O(k)})$ time.
\end{theorem}

In the remainder of this section, we prove \Cref{thm:connected}. %

\paragraph{A Subroutine to Ensure Well-Linkedness.} We need to introduce a subroutine, \Cref{lemma:DownwardWellLink}, that makes a given set downward well-linked.

\begin{lemma}
\label{lemma:DownwardWellLink}
Given an undirected graph $G$ and a vertex set $S\subseteq V(G)$ with $\deg_{G}(S)\geq 100\cdot k^{2}\cdot 8^{k}$ where $k:=|N(S)|$, there is an algorithm that outputs another set $S'\subseteq S$ such that
\begin{itemize}
\item $\deg_{G}(S')\geq \deg_{G}(S)/8^{k}$, $|N(S')|\leq |N(S)|$ and 
\item $S'$ is downward well-linked in $G$. 
\end{itemize}
The algorithm runs in $2^{O(k)}\cdot \nu$ time, where $\nu := |G[S\cup N(S)]|$.
\end{lemma}
\begin{proof}
The algorithm is given in \Cref{algo:WellLink}. We make a remark on how to find the desired separator $(A,B)$ in $G[S\cup N(S)]$ for each iteration. To do this, we can first enumerate all possible pairs of $Y_{A}:=N(S)\cap A$ and $Y_{B}:=N(S)\cap B$, and for each pair, run the Ford-Fulkerson algorithm to check if the minimum-size separator $(A^{\star},B^{\star})$ satisfying $N(S)\cap A^{\star} = Y_{A}$ and $N(S)\cap B^{\star} = Y_{B}$ has size less than both $|Y_{A}|$ and $|Y_{B}|$. Note that there are $2^{O(k)}$ pairs of $Y_{A}$ and $Y_{B}$, so it takes $2^{O(k)}\cdot\nu$ time to detect a desired $(A,B)$.

\begin{algorithm}[H]\caption{$\WellLink(S,G)$}
\label{algo:WellLink}
\begin{algorithmic}
\While{{there is a separator $(A,B)$ in $G[S \cup N(S)]$ such that\\ ~~~~~~~~~$|A \cap B| < |A \cap N(S)|$ and $|A \cap B| <|B \cap N(S)|$}}
\If{$\deg_{G}((A \setminus B) \cap S) > \deg_{G}((B \setminus A) \cap S)$}
\State $S \gets (A \setminus B) \cap S$
\Else
\State $S \gets (B \setminus A) \cap S$.
\EndIf
\EndWhile
\State \Return $S':=S$

\end{algorithmic}
\end{algorithm}

\noindent\underline{Running time.} Observe that each iteration will decrease $|N_{G}(S)|$ by at least 1, so the number of iterations is at most $k$ and the total running time is $2^{O(k)}\cdot \nu$.

\medskip

\noindent\underline{Correctness.} First, the stopping condition guarantees that the final $S$ must be downward well-linked in $G$. Also, we have argued that $|N(S')|\leq |N(S)|$. It remains to show that $\deg_G(S')\geq \deg_G(S)/8^{k}$.

It suffices to show that, in each iteration, either $\deg_{G}((A\setminus B)\cap S)\geq \deg_{G}(S)/8$ or $\deg_{G}((B\setminus A)\cap S)\geq \deg_{G}(S)/8$. 
First, observe that every edge incident to $S$ is counted at most twice in $\deg_{G}(S)$. Also, every edge incident to $S$, except those whose endpoints both lie in $A\cap B$ or whose one endpoint lies in $A \cap B$ and the other in $N(S)$, is counted in either $\deg_{G}((A\setminus B)\cap S)$ or $\deg_{G}((B\setminus A)\cap S)$. Note that there are at most $|A\cap B|^{2}\leq k^{2}$ edges with both endpoints lying in $A\cap B$, and at most $|A \cap B||N(S)| \leq k^2$ with one endpoint in $A \cap B$ and the other in $N(S)$ (since $G$ is a simple graph). Therefore,
\[
\frac{\deg_{G}(S)}{2}\leq \deg_{G}((A\setminus B)\cap S) + \deg_{G}((B\setminus A)\cap S) + 2k^{2}
\]
Thus as long as $k^{2}\leq \deg_{G}(S)/100$, we can conclude that either $\deg_{G}((A\setminus B)\cap S)$ or $\deg_{G}((B\setminus A)\cap S)$ is at least $(1/2-1/50)/2\cdot \deg_{G}(S)$ as desired. Repeating the same argument for each of the $k$ iterations finishes the proof.

\end{proof}

\paragraph{The Main Algorithm and Correctness for \Cref{item:Connected1,item:Connected2}.}
The algorithm for \Cref{thm:connected} is described in \Cref{algo:connectedwitness}.

\begin{algorithm}[ht]\caption{Connected witness elimination}\label{algo:connectedwitness}
\begin{algorithmic}[1]
\Require A graph $G$, a set $X \subseteq V(G)$, and two parameters $q,k$ such that $|X| \geq k + 1$
\Ensure A graph $G_{\final}$. %

\For{$\nu = m, m/2, m/4, \ldots, \nu\geq 100\cdot k^{2}\cdot 8^k$}
\State $Q\gets$ a $(k,\nu)$-net of size $O(mk\log n/\nu)$ obtained from \Cref{lemma:sampleset}.
\While{$Q$ is not empty}
\State $v \gets \Pop(Q)$
\State If $v\notin V(G)$ \textbf{ or } $v\in X$, \textbf{continue.}($v$ is not in any relevant witness in this case)
\State Run \Cref{lemma:witness} with $X,v,q,k$ and volume parameter $\nu' = 4\nu k$.%
\If{\Cref{lemma:witness} returns a connected $(X,v,q,k)$-witness $(L,R)$ with volume at least $\nu$}

\State $S \gets \WellLink(L\setminus R, G)$
\State Push vertices
$N_{G}(S)$ to $Q$. Push back $v$ into $Q$.
\State $G \gets (G\setminus S)\cup \Clique(N_{G}(S))$.
\EndIf

\EndWhile
\EndFor
\State Return $G_{\final}:=G$
\end{algorithmic}
\end{algorithm} 

\Cref{item:Connected1,item:Connected2} are straightforward, because in each iteration of the inner loop, $S$ is disjoint from $X$, and $S$ is downward well-linked in the current $G$ (from \Cref{lemma:DownwardWellLink}). Let $G^{(i)}$ denote the graph $G$ after $i$ updates, and let $t$ be the number of graph updates. Then $(G^{(0)},G^{(1)},...,G^{(t)})$ is the graph sequence satisfying \Cref{item:Connected2}. Note that, for each $0\leq i\leq t-1$, $S_{i}:=V(G^{(i)})\setminus V(G^{(i+1)})$ is exactly the $(i+1)$-th $S$ computed by the algorithm. 

\paragraph{Correctness of \Cref{item:Connected3}.} Observe that ${\cal S}^{(i)} = \{S_{i'}\mid 0\leq i'\leq i-1\}$ forms a partition of $V(G)\setminus V(G^{(i)})$. We need the following structural \Cref{lemma:MaximalS}. Intuitively, ${\cal S}^{(i)}$ corresponds to a family of laminar separators in the initial graph $G$, and \Cref{lemma:MaximalS} simply takes the maximal ones.

\begin{lemma}
\label{lemma:MaximalS}
For each $1\leq i\leq t$, there exists a partition ${\cal S}^{(i)}_{\star}$ of $V(G)\setminus V(G^{(i)})$ satisfying the following.
\begin{itemize}
\item ${\cal S}^{(i)}_{\star}$ is a coarsening of ${\cal S}^{(i)}$.
\item For each $S_{\star}\in {\cal S}^{(i)}_{\star}$, $N_{G}(S_{\star})\subseteq V(G^{(i)})$, and $N_{G}(S_{\star}) = N_{G^{(i')}}(S_{i'})$ for some $i'\leq i-1$.
\item $G^{(i)}$ is the same as $G[V(G^{(i)})]\cup\bigcup_{S_{\star}\in {\cal S}^{(i)}_{\star}}\Clique(N_{G}(S_{\star}))$.
\end{itemize}
\end{lemma}
\begin{proof}
We use induction: assume ${\cal S}^{(i)}_{\star}$ exists; we want to show the existence of ${\cal S}^{(i+1)}_{\star}$. Recall that $S_{i} = V(G^{(i)})\setminus V(G^{(i+1)})$. Note that each $S_{\star}\in {\cal S}^{(i)}_{\star}$ has either (i) $N_{G}(S_{\star})$ intersecting $S_{i}$ (implying that $N_{G}(S_{\star})\subseteq N_{G^{(i)}}[S_{i}]$) or (ii) $N_{G}(S_{\star})\subseteq V(G^{(i)})\setminus S_{i} = V(G^{(i+1)})$, because $\Clique(N_{G}(S_{\star}))$ lies inside $G^{(i)}$. Then, we define
\begin{align*}
S^{\nnew}_{\star}&:=S_{i}\cup\text{~the union of all $S_{\star}\in {\cal S}^{(i)}_{\star}$ satisfying (i)}\\
{\cal S}^{(i+1)}_{\star} &:= \{S^{\nnew}_{\star}\}\cup \{S_{\star}\in {\cal S}^{(i)}_{\star}\mid \text{satisfying (ii)}\}.
\end{align*}
We can easily check that ${\cal S}^{(i+1)}_{\star}$ satisfies all the requirements. In particular, the new $S^{\nnew}_{\star}$ has its boundary $N_{G}(S^{\nnew}_{\star}) = N_{G^{(i)}}(S_{i})$.

\end{proof}

Now, \Cref{item:Connected3} is a simple corollary of \Cref{lemma:MaximalS}. Throughout the algorithm, each $S_{i}$ has $|N_{G^{(i)}}(S_{i})|\leq |L\cap R|\leq k$ (here $(L,R)$ is the one used to compute $S_{i}$) by \Cref{lemma:DownwardWellLink}. Hence, in the final ${\cal S}^{(t)}_{\star}$, each $S_{\star}$ has $|N_{G}(S_{\star})|\leq k$ since $N_{G}(S_{\star})$ equals some $N_{G^{(i)}}(S_{i})$ by \Cref{lemma:MaximalS}. Finally, consider a connected component $C$ of $G\setminus G_{\final}$. Some $S_{\star}\in {\cal S}^{(t)}_{\star}$ must entirely contain $C$, and thus $N_{G}(C)\subseteq N_{G}(S_{\star})$ which has size at most $k$. This finishes the proof of \Cref{item:Connected3}.

\paragraph{Correctness of \Cref{item:Connected4}.} For an iteration of the outer loop for a particular $\nu$, we call it the \emph{$\nu$-phase}. Let $G_{\nu}$ denote the graph $G$ at the end of the $\nu$-phase, so $G_{2\nu}$ denotes the graph $G$ at the beginning of the $\nu$-phase (to avoid clutter, we use the notation $G_{2\nu}$ even for the first $\nu=m$). We use induction: 
\begin{enumerate}[label=(\Roman*)]
\item\label{Item:InductionItem4} assume there is no connected $(X,v,q,k)$-witness in $G_{2\nu}$ with volume larger than $2\nu$;
\end{enumerate} we want to show that there is no connected $(X,v,q,k)$-witness in $G_{\nu}$ with volume larger than $\nu$. Note that the induction hypothesis trivially holds for the first $\nu$.

The key idea is to maintain an invariant throughout the $\nu$-phase: $Q$ is a hitting set for the family $\{L \setminus R: (L,R)$ \text{is a connected $(X,v,q,k)$-witness of volume at least} $\nu$\} in the current $G$. This invariant is proven in \Cref{lemma:HittingSet}, but before that, we need helper \Cref{lemma:BackwardMapping,lemma:volume}, which state structural properties about rolling back witnesses on the graph sequence.

\begin{lemma}
\label{lemma:BackwardMapping}
Let $(L,R)$ be a connected separator in $G^{(i+1)}$. Then $L\setminus R$ lies in a connected component of $G^{(i)}\setminus (L\cap R)$.
\end{lemma}
\begin{proof}
Recall that $S_{i} = V(G^{(i)})\setminus V(G^{(i+1)})$, and $G^{(i+1)} = G^{(i)}\setminus S_{i}\cup \Clique(N_{G^{(i)}}(S_{i}))$. Because we add $\Clique(N_{G^{(i)}}(S_{i}))$ to $G^{(i+1)}$, there must be either $N_{G^{(i)}}(S_{i})\subseteq L$ or $N_{G^{(i)}}(S_{i})\subseteq R$. If the latter holds, then \(N_{G^{(i)}}(S_i)\cap (L\setminus R)=\emptyset\), so no clique edge
added on \(N_{G^{(i)}}(S_i)\) is used inside \(G^{(i+1)}[L\setminus R]\). Thus
\(L\setminus R\) is connected in \(G^{(i)}\setminus (L\cap R)\), as desired. Hence, from now on, assume $N_{G^{(i)}}(S_{i})\subseteq L$.

It suffices to show that $(L\setminus R)$ is connected in $G^{(i)}[(L\setminus R)\cup S_{i}]$. Note that $L\setminus R$ is connected in $G^{(i+1)}[L\setminus R]$ since $(L,R)$ is a connected witness. Also, from $G^{(i)}[(L\setminus R)\cup S_{i}]$ to $G^{(i+1)}[L\setminus R]$, the difference is that we remove the vertex set $S_{i}$ and gain a clique on $N_{G^{(i)}}(S_{i})\cap (L\setminus R)$. 

Formally, writing
\[
D := N_{G^{(i)}}(S_i) \cap (L \setminus R),
\]
we want to show that \(D\) is connected in
\[
G^{(i)}[S_i \cup D].
\]

To see this, we will exploit the downward well-linkedness of $S_{i}$. Assume the opposite for contradiction. Then we can obtain a partition $(\tilde{A}, \tilde{B})$ of $S_{i}\cup D$ such that (i) both have non-empty intersection with $D$, denoted by $D_{\tilde{A}}$ and $D_{\tilde{B}}$, and (ii) there is no edge connecting $\tilde{A}$ and $\tilde{B}$ in $G^{(i)}[S_i \cup D]$. Now, consider the separator $(A,B)$ in $G^{(i)}[S_{i}\cup N_{G^{(i)}}(S_{i})]$ defined by
\begin{align*}
A &:= \tilde{A}\cup ((S_{i}\cup N_{G^{(i)}}(S_{i}))\cap R) = \tilde{A}\cup (N_{G^{(i)}}(S_{i})\cap R),\\
B &:= \tilde{B}\cup ((S_{i}\cup N_{G^{(i)}}(S_{i}))\cap R) = \tilde{B}\cup (N_{G^{(i)}}(S_{i})\cap R).
\end{align*}
It violates the downward well-linkedness of $S_{i}$, because $A\cap N_{G^{(i)}}(S_{i}) = (N_{G^{(i)}}(S_{i})\cap R)\cup D_{\tilde{A}}$ clearly has size larger than $A\cap B = N_{G^{(i)}}(S_{i})\cap R$ (similarly, $|B\cap N_{G^{(i)}}(S_{i})|>|A\cap B|$). This is a contradiction, and completes the proof. 
\end{proof}

\begin{lemma}
\label{lemma:volume}
Let $(L,R)$ be a connected $(X,v,q,k)$-witness in $G^{(i)}$. For $i'\leq i-1$, there is a connected $(X,v,q,k)$-witness $(L_{\old},R_{\old})$ in $G^{(i')}$ satisfying that $\deg_{G^{(i)}}(L\setminus R)\leq (k+1)\cdot \deg_{G^{(i')}}(L_{\old}\setminus R_{\old})$.
\end{lemma}
\begin{proof}
Let $(L^{(i)},R^{(i)}):=(L,R)$. We will define $(L^{(j)},R^{(j)})$ in $G^{(j)}$ for $j$ from $i-1$ to $i'$: ensure that
\begin{align*}
L^{(j)}\cap R^{(j)} &= L^{(j+1)}\cap R^{(j+1)}\\
L^{(j)}\setminus R^{(j)} &= \text{the connected component of $G^{(j)}\setminus (L^{(j+1)}\cap R^{(j+1)})$ containing $L^{(j+1)}\setminus R^{(j+1)}$} 
\end{align*}
Note that $(L^{(j)},R^{(j)})$ is well-defined by \Cref{lemma:BackwardMapping}. Additionally, we have 
\[
R^{(j)}\setminus L^{(j)}\supseteq R^{(j+1)}\setminus L^{(j+1)},
\]
because from $G^{(j+1)}$ to $G^{(j)}$, turning $\Clique(N_{G^{(j)}}(S_{j}))$ back to the subgraph $G^{(j)}[S_{j}\cup N_{G^{(j)}}(S_{j})]$ will not bring extra connectivity between $L^{(j+1)}\setminus R^{(j+1)}$ and $R^{(j+1)}\setminus L^{(j+1)}$, meaning $L^{(j)}\setminus R^{(j)}$ is disjoint from $R^{(j+1)}\setminus L^{(j+1)}$. Moreover, $N_{G^{(j)}}(L^{(j)}\setminus R^{(j)})=L^{(j)}\cap R^{(j)}$: if a vertex $c\in L^{(j)}\cap R^{(j)}$ is adjacent to $L^{(j+1)}\setminus R^{(j+1)}$ only through an edge of $\Clique(N_{G^{(j)}}(S_{j}))$, then by the downward well-linkedness of $S_{j}$ (as in the proof of \Cref{lemma:BackwardMapping}), $c$ is connected to $L^{(j+1)}\setminus R^{(j+1)}$ through $S_{j}$, and hence has a neighbor in $L^{(j)}\setminus R^{(j)}$. Therefore, we can easily observe that each $(L^{(j)},R^{(j)})$ is a connected $(X,v,q,k)$-witness in $G^{(j)}$. 

It remains to prove the volume bound. We apply the coarsening argument of
\Cref{lemma:MaximalS} to the subsequence of updates from
\(G^{(i')}\) to \(G^{(i)}\). This gives a partition \(\mathcal S_\star\) of
\(V(G^{(i')})\setminus V(G^{(i)})\) into coarsened removed pieces such that
\[
    G^{(i)}
    =
    G^{(i')}[V(G^{(i)})]
    \cup
    \bigcup_{S_\star\in\mathcal S_\star}
        \Clique\bigl(N_{G^{(i')}}(S_\star)\bigr).
\]
Moreover, for every \(S_\star\in\mathcal S_\star\), the boundary
\(N_{G^{(i')}}(S_\star)\) is equal to the boundary \(N_{G^{(r)}}(S_r)\) of
some actual removed set \(S_r\), where \(i'\le r<i\). Hence
\[
    |N_{G^{(i')}}(S_\star)|\le k.
\]

Fix a vertex \(a\in V(G^{(i)})\). Every neighbor of \(a\) in \(G^{(i)}\) is
either already a neighbor of \(a\) in \(G^{(i')}\), or is created by the
clique added on \(N_{G^{(i')}}(S_\star)\) for some coarsened piece
\(S_\star\) with \(a\in N_{G^{(i')}}(S_\star)\). For each such coarsened
piece, choose one edge of \(G^{(i')}\) from \(a\) to \(S_\star\). These
chosen edges are distinct, because the sets \(S_\star\) are pairwise
disjoint. Therefore the number of coarsened pieces whose boundary contains
\(a\) is at most \(\deg_{G^{(i')}}(a)\).

Each such coarsened piece contributes at most \(k\) neighbors of \(a\) in
\(G^{(i)}\), since its boundary has size at most \(k\). Thus
\[
    \deg_{G^{(i)}}(a)
    \le
    \deg_{G^{(i')}}(a)
    +
    k\cdot \deg_{G^{(i')}}(a)
    =
    (k+1)\deg_{G^{(i')}}(a).
\]

Summing over \(a\in L\setminus R\), and using
\(L\setminus R\subseteq L_{\mathrm{old}}\setminus R_{\mathrm{old}}\), gives
\[
\begin{aligned}
    \deg_{G^{(i)}}(L\setminus R)
    &=
    \sum_{a\in L\setminus R}\deg_{G^{(i)}}(a) \\
    &\le
    (k+1)\sum_{a\in L\setminus R}\deg_{G^{(i')}}(a) \\
    &\le
    (k+1)\deg_{G^{(i')}}(L_{\mathrm{old}}\setminus R_{\mathrm{old}}).
\end{aligned}
\]
This proves the desired volume bound.

\end{proof}

Now, we are ready to show the invariant of the hitting set $Q$. Recall from the beginning of the proof of \Cref{item:Connected4} that, we assume that \ref{Item:InductionItem4} there is no connected $(X,v,q,k)$-witness in $G_{2\nu}$ with volume larger than $2\nu$, where $G_{2\nu}$ is the graph at the beginning of $\nu$-phase. 

\begin{lemma}
\label{lemma:HittingSet}
With the assumption \ref{Item:InductionItem4}, throughout the $\nu$-phase, the set $Q$ remains a hitting set of the family $L\setminus R$ over all connected $(X,v,q,k)$-witnesses $(L,R)$ in $G$ with volume at least $\nu$, with high probability.

\end{lemma}

\begin{proof}

At the very beginning of the $\nu$-phase, we initialize $Q$ to be a $(k,\nu)$-net obtained from~\Cref{lemma:sampleset}. Hence, with high probability, the property holds initially. Inductively, assume that up to some inner iteration of the $\nu$-phase, the property continues to hold. 

There are two cases depending on the output of \Cref{lemma:witness}. Before we start the case study, we need to show that every time we call \Cref{lemma:witness}, its Condition (A) is satisfied: every connected $(X,v,q,k)$-witness in the current graph $G$ has volume at most $\nu' = 4\nu k$. 

\medskip

\noindent\underline{Condition (A) of \Cref{lemma:witness}.} Assume for contradiction that there exists a connected $(X,v,q,k)$-witness $(L,R)$ in the current graph $G$ with volume $>\nu'$. Then by \Cref{lemma:volume}, there is a connected $(X,v,q,k)$-witness $(L_{\old},R_{\old})$ in the graph $G_{2\nu}$ with volume $>\nu'/(k+1)\geq 2\nu$, contradicting assumption (I).

\medskip

Now we start the case study. To avoid clutter, we assume the high-probability event of \Cref{lemma:witness} holds.

\medskip

\noindent\underline{Case 1.} In the first case, \Cref{lemma:witness} returns a connected  $(X,v,q,k)$-witness of volume $\geq \nu$. We update the graph $G$ by carving out a set $S$ and turning $N_G(S)$ into a clique. We also update $Q$ by removing $v$ and adding $N_{G}(S)$. Let $G,G_{\nnew}$ be the graphs before and after the update respectively. Fix a connected $(X,v,q,k)$-witness $(L,R)$ with volume at least $\nu$ in $G_{\nnew}$.
\begin{itemize}
\item
If $L\setminus R$ is disjoint with $N_{G}(S)$, then $\deg_{G_{\nnew}}(L\setminus R) = \deg_{G}(L\setminus R)$. Note that $(L,V(G) \setminus (L \setminus R))$ is a connected $(X,v,q,k)$-witness in $G$: there can be no edges between $ L \setminus R$ and $S$ since $L \setminus R$ does not intersect $N_G(S)$. Recall that in this case, we push the popped vertex $v$ back into $Q$, and hence $Q\cap (L\setminus R)$ does not change during the update, so $Q$ still hits $(L,R)$ by the induction hypothesis.

\item
Otherwise, $L\setminus R$ intersects $N_{G}(S)$. Recall that we add $N_{G}(S)$ to $Q$, and thus $Q$ hits $L\setminus R$.
\end{itemize}

\noindent\underline{Case 2.} In the second case, \Cref{lemma:witness} guarantees that either connected $(X,v,q,k)$-witnesses have maximum volume $<\nu$, or that there is no such witness. Hence, we can safely remove $v$ from $Q$, exactly what the algorithm did.

\end{proof}

At the end, we complete the proof of \Cref{item:Connected4} by finishing the induction. Assuming the induction hypothesis \ref{Item:InductionItem4}, \Cref{lemma:HittingSet} guarantees that, at the end of the $\nu$-phase, $Q$ hits all connected $(X,v,q,k)$-witnesses in $G_{\nu}$ with volume at least $\nu$ (with high probability). However, the algorithm guarantees that $Q$ is empty at the end of the $\nu$-phase, so there is no $(X,v,q,k)$-witness in $G_{\nu}$ with volume at least $\nu$, as desired. 

\paragraph{Running Time.} Finally, we bound the running time of \Cref{algo:connectedwitness}.

\begin{lemma}
\Cref{algo:connectedwitness} runs in time $\Otil(m\cdot 2^{O(k)})$.
\end{lemma}

\begin{proof}

We will bound the time taken for each iteration of the outer for loop. Since there are only $O(\log n)$ choices of $\nu$, this suffices to show the runtime bound. Fix a phase $\nu$. First, we bound the number of pushes into $Q$. Then, we bound the time to process a vertex $v \in Q$ once it is popped.

\medskip

\noindent\underline{Number of Pushes into $Q$.} 
Every time we make a push into $Q$ it is due to an update to $G$, which removes $S$, adds a clique on $N_G(S)$, and pushes at most $k + 1$ vertices $v \cup N_G(S)$ into $Q$. Hence, we bound the number of updates instead. To do this, we keep track of the total degree of $G$. When updating $G$, we remove a vertex set $S$. By \Cref{lemma:DownwardWellLink} and the algorithm statement, we have
\[
\deg_{G}(S)\geq \deg_{G}(L\setminus R)/8^{k}\geq \nu/8^{k},
\]
which causes a volume drop of at least $\nu/8^{k}$.
Also, we add $\Clique(N_{G}(S))$, which causes a volume increase of at most $2k^{2}\leq \nu/(50\cdot 8^{k})$. Therefore, the net effect is a reduction in volume of at least $\rho = \nu/(2\cdot 8^k)$.
The initial volume is $2m$, and thus the number of pushes to $Q$ is at most $2m/\rho = 4\cdot 8^{k}\cdot m/\nu$.

\medskip

\noindent{\underline{Number of Pops from $Q$.}} Since $Q$ has $O(mk\log n/\nu)$ vertices initially, and each push will add at most $k$ vertices, the total number of pops is at most $O((\log n + 8^{k})\cdot mk/\nu)$.

\medskip

\noindent{\underline{Time for One Pop.}} Next, we bound the running time when we pop $v \in Q$ by 
\[
\underbrace{\tilde{O}(2^{O(k)}\cdot \nu k)}_{\text{\Cref{lemma:witness}}} + \underbrace{O(2^{O(k)}\cdot \nu k)}_{\text{\Cref{lemma:DownwardWellLink}}} + \underbrace{O(\nu k+k^{2})}_{\text{Update (the adjacency list of) $G$}} = \tilde{O}(2^{O(k)}\cdot \nu).
\]

\medskip

\noindent{\underline{Total Running Time.}} The total running time is simply $\tilde{O}(2^{O(k)}\cdot\nu) \cdot 2^{O(k)}\cdot m/\nu = \tilde{O}(2^{O(k)}\cdot m)$.

\end{proof}

%% file: phase2.tex
\section{Witness Elimination.}\label{sec:disconnectedwitness}

In the previous section, we obtained a graph $G$ such that there is no connected $(X,v,q,k)$-witness for any $v \in V(G)$. In this section, we will obtain a set $T \subseteq V(G)$ such that there is no $(X,v,k)$ witness for $T$ in $G$. An $(X,v,k)$ witness is a generalization of $(X,v,k,k)$ witness, and is defined below.

\begin{definition}
Given a graph $G$, a set $T$, and $v \in V(G)$, an $(X,v,k)$-witness for $T$ is a separator $(L,R)$ such that $|L \cap R| \leq k$, $|L \cap T| > |L \cap R|$, $v \in L \setminus R$, $X \subseteq R$. A connected $(X,v,k)$-witness is an $(X,v,k)$-witness $(L,R)$ where $G[L \setminus R]$ is connected and $N_G(L \setminus R)=L\cap R$.
    
\end{definition}

\begin{observation}
In the graph $G$, an $(X,v,k)$-witness $(L,R)$ for $T$ and some $v\in V(G)$ is also an $(X,v',k)$-witness for $T$ and some $v'\in T$.
\end{observation}
\begin{proof}
This is simply because $|(L\setminus R)\cap T| = |L\cap T| - |L\cap R\cap T| \geq |L\cap T| - |L\cap R|\geq 1$.
\end{proof}

By the above observation, $\{\text{$(X,v,k)$-witnesses for $T$}\mid v\in V(G)\}$ equals $\{\text{$(X,v,k)$-witnesses for $T$}\mid v\in T\}$. Hence, to avoid clutter, we use $(X,v,k)$-witness (without specifying $v$) to denote an arbitrary one in the above collection. Analogous simplifications also apply to $(X,v,k,k)$-witnesses for $T$.

Observe that an $(X,v,k)$-witness for $T$ in $G$ is simply an $(X,v,k',k')$-witness for $T$ and some $k'\leq k$.
Therefore, if there is no $(X,v,k)$-witness for $T$ in $G$, then there is no $(X,v,k,k)$-witness for $T$ in $G$. Our algorithm will eventually construct $T$ such that there is no $(X,v,k)$-witness for $T$ in $G$. 

The following is the main theorem of this section.

\begin{restatable}{theorem}{WitnessEliminate}
\label{thm:WitnessEliminate}
Given a graph $G$, a set $X \subseteq V(G)$, and parameters $q\geq k\geq 1$ such that $|X|\geq k+1$ and there is no connected $(X,v,q,k)$-witness in $G$ for any $v \in V(G)$, there is a randomized algorithm that runs in time $q^{O(1)}k^{O(k)}m\cdot \polylog(n)$ that with high probability computes a set $T \subseteq V(G)$, such that 
\begin{enumerate}
\item $X \subseteq T$,
\item there is no $(X,v,k,k)$-witness for $T$ and any $v\in T$ in $G$. (In fact, we will ensure there is no $(X,v,k)$ witness for $T$ in $G$), and
\item the adhesion of $T$, $\sigma(T) \leq k$.
\end{enumerate}

\end{restatable}

Throughout this section, we will work with a graph $G$ and a terminal set $T \subseteq V(G)$. The graph $G$ will remain unchanged throughout, and should be thought of simply as the final output graph from the previous section. The terminal set $T$ itself will be updated. We start with $T = V(G)$. The key sub-routine is to find an $(X,v,k)$ witness for $T$ in $G$ if it exists. If it exists, we update $T$, while \emph{keeping $G$ unchanged}. 

\subsection{\texorpdfstring{Detecting $(X,v,k)$-Witnesses for $T$.}{Detecting (X,v,k)-Witnesses for T.}}

Naturally, we need a subroutine to detect $(X,v,k)$-witnesses for $T$. If we only want to detect the connected ones (formalized in \Cref{lemma:connectedwitness}), then the local important separator algorithm suffices. However, the main challenge is that we need to detect even the non-connected ones. To this end, we need to introduce the concepts of \emph{minimal} $(X,v,k)$-witness for $T$ (\Cref{def:Minimal}) and the \emph{torso} of $T$ (\Cref{def:torso}). The main lemma of this subsection is \Cref{lemma:unconnectedwitness}.

\begin{lemma}\label{lemma:connectedwitness}
Let $G$ be a graph and $T \subseteq V(G), X \subseteq T, v \in T$, with $|X|\geq k+1$. Suppose we are promised that there is no connected $(X,v,q,k)$-witness in $G$ for some $q \geq k$. Then there is an algorithm that determines if there is a connected $(X,v,k)$-witness for $T$ in $G$, and finds one if it exists. The algorithm runs in time $q^2\cdot 2^{O(k)}$ and it is correct with high probability.
\end{lemma}

\begin{proof}
Observe that, since there is no connected $(X,v,q,k)$-witness in $G$, there are at most $q$ vertices in $L \setminus R$ for any connected $(X,v,k)$-witness $(L,R)$ for $T$ in $G$. The volume of any such witness is therefore at most $q(q + k)\leq 10 q^{2}$.

We now use the algorithm from~\Cref{lemma:witness} with $\nu = 10q^2$ to find a connected $(X,v,k)$-witness if it exists. The original \Cref{lemma:witness} does not directly provide this, but the algorithm and proof are almost the same - the only difference from \Cref{lemma:witness} is that, after enumerating the relevant
important \(v\)-\(X\) separators, we filter for the condition
\(|L\cap T|>|L\cap R|\). This condition is preserved when a separator is
pushed to an important separator, since the left side only grows and the
separator size does not increase.
\end{proof}

\begin{definition}[The Torso of a Terminal Set]%
\label{def:torso}%
Given a connected graph $G$ and a terminal set $T \subseteq V(G)$, the torso of $T$, denoted $\Torso(T)$ is the graph with vertex set $T$ obtained by keeping every existing edge in $G$ and  turning every adhesion of $T$ into a clique. Equivalently, there is an edge $(u,v) \in E(\Torso(T))$ for distinct $u,v \in T$ whenever there is a path between $u$ and $v$ in $G$ none of whose internal vertices are in $T$.\footnote{We do not add parallel edges between any pairs of vertices, so there is at most one edge between every pair of vertices.}  
\end{definition}

\begin{definition}
\label{def:Minimal}
We say that an $(X,v,k)$-witness for $T$, $(L,R)$, is \emph{(inclusion-wise) minimal} if there is no $(X,u,k)$-witness $(L',R')$, with $L' \setminus R' \subset L \setminus R$ for any $u \in T$.
\end{definition}

The motivation for defining minimal $(X,v,k)$-witnesses is that, if we want to eliminate all $(X,v,k)$-witnesses, it suffices to eliminate all the minimal ones. Furthermore, minimal witnesses give more structure and connectivity that we will exploit. We remark that this notion of minimal witness has appeared in previous papers (e.g., \cite{cygan2020randomized,anand2025unbreakable}) on unbreakable decomposition in different forms. For completeness and clarity, we will fully prove all the structural properties \Cref{lemma:TorsoConnected,lemma:TorsoVolume,lemma:TorsoDegree} that we will exploit for minimal witnesses.

\begin{lemma}
\label{lemma:TorsoConnected}
Let $G$ be a graph with terminal set $T\subseteq V(G)$. Let $H:=\Torso(T)$ be the torso of $T$ in $G$. Then, for any minimal $(X,v,k)$-witness $(L,R)$ for $T$ in $G$, $H[(L \setminus R) \cap T]$ is connected.
\end{lemma}

\begin{proof}
Consider a minimal $(X,v,k)$-witness $(L,R)$ for $T$ in $G$, and suppose to the contrary that $H[(L \setminus R) \cap T]$ is not connected. We assume w.l.o.g. that $N(L \setminus R) = L \cap R$, because otherwise, we can consider another minimal $(X,v,k)$-witness $(L',R'):=(L \setminus R \cup (N(L \setminus R)), V(G) \setminus (L \setminus R))$ (note that $(L'\setminus R')\cap T = (L\setminus R)\cap T$, implying the minimality) and continue the rest of the proof.

Let $C_1, C_2, \ldots, C_t$ denote the connected components of $G[L \setminus R]$. Note that each $C_{i}$ must intersect $T$, since otherwise $(L,R)$ is not minimal. Also, each connected component of $H[(L \setminus R)\cap T]$ must be the union of $C_i \cap T$ for some collection of components $C_i$. Now, since $H[(L \setminus R)\cap T]$ is disconnected, we can partition $L\setminus R$ into two parts $D_{1}$ and $D_{2}$ such that
\begin{itemize}
\item there is no edge in $H$ connecting $D_{1}\cap T$ and $D_{2}\cap T$, and
\item $D_{1} = C_{1}\cup \cdots\cup C_{h}$ and $D_{2} = C_{h+1}\cup\cdots\cup C_{t}$ (after re-numbering $\{C_{i}\mid 1\leq i\leq t\}$).
\end{itemize}

The key observation is that $N_{G}(D_1) \cap N_{G}(D_2) \subseteq T$. Otherwise, there is a vertex $w\notin T$ lying in $N_{G}(D_{1})\cap N_{G}(D_{2})$, i.e., $w$ is adjacent to some vertex $w_{1}\in D_{1}$ and $w_{2}\in D_{2}$. Note that, $w_{1}$ can reach a $T$-vertex $w'_{1}\in C_{w_{1}}\cap T\subseteq D_{1}\cap T$, where $C_{w_{1}}\ni w_{1}$ is the component containing $w_{1}$ of $G[L\setminus R]$, via a path in $G[C_{w_{1}}]$ internally disjoint from $T$. We can find $w'_{2}\in D_{2}\cap T$ similarly. This creates a path in $G$ internally disjoint from $T$ connecting $T$-vertices $w'_{1}$ and $w'_{2}$. By the definition of $\Torso(T)$, there is an edge between $w'_{1}\in D_{1}\cap T$ and $w'_{2}\in D_{2}\cap T$, a contradiction.

Now we argue that one of $(D_1 \cup N(D_1), V(G) \setminus D_1)$ or $(D_2 \cup N(D_2), V(G) \setminus D_2)$ must also be an $(X,u,k)$ witness for some $u \in T$.

The proof follows by simple counting. We know that $|L \cap T| > |L \cap R|$, or equivalently, $|(L \setminus R) \cap T| > |(L \cap R) \setminus T| $. But $(L \setminus R) \cap T$ is the union of disjoint sets $D_1 \cap T$ and $D_2 \cap T$, and $(L \cap R) \setminus T$ is the union of disjoint sets $N(D_1) \setminus T$ and $N(D_2) \setminus T$. Thus, it follows that $|D_i \cap T|> |N(D_i) \setminus T|$ for at least one of $i \in \{1,2\}$. Pick this $i$, and let $u$ be any vertex of $D_i\cap T$. Then it follows that $(D_i \cup N(D_i), V(G) \setminus D_i)$ is an $(X,u,k)$ witness, contradicting the minimality of the witness $(L,R)$.

\end{proof}

\begin{lemma}
\label{lemma:TorsoVolume}
For any minimal $(X,v,k)$-witness $(L,R)$ for $T$ in $G$, every connected component of $G[L \setminus R]$ intersects $T$.

Further, either (a) $L \setminus R$ is connected, or (b) there are at most $k + 1$ connected components of $G[L \setminus R]$ and each of these components has at most $k + 1$ vertices in $T$.
\end{lemma}

\begin{proof}
Suppose there is a connected component of $L \setminus R$, say $C_1$, that does not intersect $T$. Then it follows that $(L \setminus C_1, R \cup C_1)$ is also an $(X,v,k)$ witness, which contradicts the minimality of $(L,R)$. Similarly, if there is a vertex $x$ of $L \cap R$ that does not have any edge to $L \setminus R$, then $(L \setminus x, R)$ is also an $(X,v,k)$ witness, a contradiction.

If $L \setminus R$ is connected, we satisfy condition (a). 

Otherwise, suppose that $L \setminus R$ has more than $k + 1$ connected components. By the above paragraph, we can assume that each such connected component intersects $T$. Let $C^*$ be any component of $ L \setminus R$. It is easy to see that $(L \setminus C^*, R \cup C^*)$ is still an $(X,v,k)$-witness, contradicting minimality.

Next, suppose that there is a connected component $C_1$ of $L \setminus R$ that has more than $k + 1$ vertices of $T$. Since $L \setminus R$ is not connected, there is another connected component $C_2$ of $L \setminus R$. But then $(L \setminus C_2, R \cup C_2)$ is still an $(X,v,k)$-witness, again contradicting minimality.

\end{proof}

\begin{lemma}
\label{lemma:TorsoDegree}
For any minimal $(X,v,k)$-witness $(L,R)$ for $T$ in $G$, if $\sigma(T) \leq k$ and $G[L \setminus R]$ is not connected, then for any $v \in (L \setminus R) \cap T$, the degree of $v$ in $\Torso(T)$ is at most $6k^2$.

\end{lemma}

\begin{proof}
Fix a $v \in (L \setminus R)$. First, there are at most $k  + 1 $ connected components having $k + 1$ vertices of $T$ in $(L \setminus R)$, so the contribution to the degree of $v$ from these vertices is at most $(k + 1)^2$. Second, there are at most $k$ vertices in $L \cap R$. 

Third, we show that if $(v,u)$ is an edge in $\Torso(T)$ with $u \in(R \setminus L)$, then there can be at most $k$ connected components in $G \setminus T$ which have $u$ and $v$ both as part of their adhesion. The reason for this is that for every such connected component $C$, there is a path between $u$ and $v$ whose every internal vertex is in $C$. But since $L \cap R$ separates $L \setminus R$ from $R \setminus L$, one such internal vertex must be in $L \cap R$, and $|L \cap R| \leq k$. Thus there are at most $k$ connected components $C$. Each adhesion contributes a degree of at most $k$ to $v$, and there are at most $k$ of them, leading to a total contribution of at most $k^2$ to the degree of $v$.

Together, the total degree of $v$ is bounded by $(k + 1)^2 + k + k^2 \leq 2k^2 + 3k + 1 \leq 6k^2$.

\end{proof}

We are ready to show the subroutine \Cref{lemma:unconnectedwitness} for detecting (minimal) $(X,v,k)$-witnesses for $T$.

\begin{lemma}\label{lemma:unconnectedwitness}
Let $G$ be a graph and $T \subseteq V(G), X \subseteq T, v \in T$, with $|X|\geq k+1$, and suppose that $\sigma(T) \leq k$. Suppose we are promised that there is no connected $(X,u,q,k)$-witness in $G$ for any $u \in V(G)$.

Then given access to the adjacency list of $\Torso(T)$, there is an algorithm running in time $q^2\cdot k^{O(k)}$ that either returns an $(X,v,k)$-witness for $T$ in $G$, or declares that there is no minimal $(X,v,k)$-witness for $T$ in $G$. 
\end{lemma}

\begin{proof}

First, we run~\Cref{lemma:connectedwitness} to check if there is a connected $(X,v,k)$-witness for $T$ in $G$, if so we return it. 

Otherwise, any minimal $(X,v,k)$-witness $(L,R)$ for $T$ must be disconnected. First, we guess $w = \min\{|(L \setminus R) \cap T|, k+1\}$. Next, we try to obtain a subset of $(L \setminus R) \cap T$ of size $w$. We will exploit \Cref{lemma:TorsoConnected}, i.e., $H[(L\setminus R) \cap T]$ is connected, where $H:=\Torso(T)$. %

In order to do so, we employ the following branching algorithm. Guess a neighbor $u_1$ of $v$ in $H$ such that $u_{1}\in L \setminus R$. Since $\deg_{H}(v) \leq 6k^2$ by \Cref{lemma:TorsoDegree}, there are only $6k^2$ candidates for $u_1$. If the guess is correct, again $\deg_{H}(u_1) \leq 6k^2$, so we first make sure that this degree bound holds, and then repeat and guess another neighbor $u_{2}$ of $\{v,u_{1}\}$ in $H$ such that $u_2 \in L \setminus R$. Repeat this process till we obtain $u_{w-1}$, so that $S_v = \{v, u_1, \ldots u_{w-1}\}$ forms a set of size $w$.

There are $k^{O(k)}$ branches, and one of them must find a subset  $S_v^* \subseteq (L \setminus R) \cap T$ with $|S_v^*| = w$. For each branch, we make $S_v$ connected by adding a clique on the vertex set $S_v$, and let the resulting graph be $G_c$. The following lemma shows that we can now focus on finding a connected witness in this modified graph $G_c$.

\begin{lemma}
We have the following equivalence between witnesses in $G_c$ and $G$.
\begin{itemize}
\item Any $(X,v,k)$ connected witness for $T$ in $G_c$ is an $(X,v,k)$ witness for $T$ in $G$.

\item Conversely, suppose that there is a minimal $(X,v,k)$ witness $(L,R)$ for $T$ in $G$, and we correctly guessed $w$ and $S_v$ as defined above, then there must be an $(X,v,k)$ connected witness for $T$ in $G_c$.  

\end{itemize}

\end{lemma}

\begin{proof}
 First, any $(X,v,k)$ witness for $T$ in this graph $G_c$ corresponds to an $(X,v,k)$-witness in $G$, as we only added edges to $G$ to obtain $G_c$. 

Conversely, suppose there is a minimal $(X,v,k)$-witness $(L,R)$ for $T$ in $G$. We may assume $N_G(L\setminus R)=L\cap R$, since deleting redundant separator vertices preserves the witness condition and minimality. Consider the branch that correctly guesses a set
$S_v \subseteq (L\setminus R)\cap T$ of size
$w=\min\{|(L\setminus R)\cap T|,k+1\}$ containing $v$.

First suppose that $w=|(L\setminus R)\cap T|$. Then $S_v=(L\setminus R)\cap T$. By minimality of $(L,R)$, every connected component of $G[L\setminus R]$ contains a vertex of $T$. Therefore, after adding a clique on $S_v$, all connected components of $G[L\setminus R]$ become connected to each other. Hence $(L,R)$ itself is a connected $(X,v,k)$-witness for $T$ in $G_c$.

Otherwise, $w=k+1$. Let $C$ be the union of all connected components of $G[L\setminus R]$ that contain a vertex of $S_v$. Since we add a clique on $S_v$, the set $C$ is connected in $G_c$. Moreover, $v\in C$, $X\cap C=\emptyset$, and
\[
    |C\cap T|\ge |S_v|=k+1>|L\cap R|.
\]
Also, since $C\subseteq L\setminus R$ and $(L,R)$ is a separator in $G$, we have
\[
    N_{G_c}(C)\subseteq L\cap R,
\]
where the added clique edges do not create neighbors outside $C$ because all vertices of $S_v$ lie in $C$. Thus
\[
    (C\cup N_{G_c}(C),\, V(G_c)\setminus C)
\]
is a connected $(X,v,k)$-witness for $T$ in $G_c$.

\end{proof}

To find a connected $(X,v,k)$-witness for $T$ in $G_{c}$, we call~\Cref{lemma:connectedwitness}. If the promise of \Cref{lemma:connectedwitness} is satisfied, i.e., there is no connected $(X,v,Q,k)$-witness in $G_{c}$ for some $Q$, then ~\Cref{lemma:connectedwitness} will return either
\begin{itemize}
\item a connected $(X,v,k)$-witness for $T$ in $G_{c}$, which, as pointed out before, corresponds to an $(X,v,k)$-witness for $T$ in $G$, or
\item report that no connected $(X,v,k)$-witness for $T$ in $G_{c}$ exists, meaning that there is no minimal $(X,v,k)$-witness for $T$ in $G$.
\end{itemize}

The running time per branch would then be $O(Q^2 2^{O(k)})$. Thus, it suffices to show that the promise of \Cref{lemma:connectedwitness} is satisfied with a small value of $Q$. We will show this for $Q = q(k + 1) + k$. 

Suppose to the contrary that there is a connected $(X,v,Q,k)$-witness $(L',R')$ in $G_c$. Then, $(L',R')$ is also a (possibly disconnected) $(X,v,Q,k)$ witness in $G$. Now since $L' \setminus R'$ is connected in $G_c$, $G[L' \setminus R']$ has at most $k + 1$ connected components, since $|S_v| \leq k + 1$. Choose the connected component $C$ of $L' \setminus R'$ which has the maximum number of vertices. First we have $|L'| - |L' \cap R'| > Q - k = q(k + 1)$. By the choice of $C$, we get $|C| \geq (|L'| - |L' \cap R'|)/(k + 1) > q$. But now $(C \cup N_G(C), V(G) \setminus C)$ is a connected $(X,u,q,k)$-witness in $G$ for any $u \in C$.

\end{proof}

\subsection{\texorpdfstring{Proof of \Cref{thm:WitnessEliminate}.}{Proof of Witness Elimination.}}

For convenience, we restate \Cref{thm:WitnessEliminate} below.

\WitnessEliminate*

The algorithm is given in \Cref{algo:WitnessEliminate}.

\begin{algorithm}[!htbp]\caption{Witness elimination}
\label{algo:WitnessEliminate}
\begin{algorithmic}[1]
\Require A graph $G$ with $X\subseteq V(G)$, and two parameters $q,k$.
\Ensure A terminal set $T$.

\State $Q \gets V(G)$.
\State $T \gets V(G)$
\While{$Q$ is not empty}
\State Keep doing $\Pop(Q)$ until getting a $v\in T$
\State Run \Cref{lemma:unconnectedwitness} with $T,X,v,q$.
\If{~\Cref{lemma:unconnectedwitness} returns an $(X,v,k)$-witness $(L,R)$ for $T$}
\State $(L^{\star},R^{\star})\gets \MakeLean((L,R),G,T)$
\State $U \gets L^* \cap R^* \setminus T$.
\State $T\gets (T\setminus L^{\star})\cup (L^{\star}\cap R^{\star})$

\State Update $\Torso(T)$.
\State Push vertices in $L^{\star}$ to $Q$. Push back $v$ into $Q$. Finally, for every $u \in U$, push the neighbors of $u$ in $\Torso(T)$ into $Q$.
\EndIf

\EndWhile

\end{algorithmic}
\end{algorithm}

\paragraph{Lean Witnesses and the $\MakeLean$ Subroutine.} We introduce the concepts of \emph{lean witnesses}, and then describe the subroutine $\MakeLean$, which makes a witness lean. We note that, as shown in previous works \cite{cygan2020randomized,anand2025unbreakable}, the motivation for introducing lean witnesses is to ensure that $T$ has small adhesion throughout the algorithm. For this part, we only state the algorithm and lemmas, and we refer the readers to Section 5.2 of \cite{anand2025unbreakable} for detailed proofs.

\begin{definition}[Lean Witnesses]%
An $(X,v,k)$-witness $(L^{\star},R^{\star})$ is lean if there exists a set of $|L^{\star}\cap R^{\star}|$ vertex disjoint paths  in $G[L^{\star}]$ starting from every vertex in $L^{\star}\cap R^{\star}$ and ending at $L^{\star}\cap T$.
\end{definition}

\begin{lemma}[Modified Lemma 5.18 in \cite{anand2025unbreakable}]%
Let $G$ be an undirected graph with a terminal set $T$. Given an $(X,v,k)$-witness $(L,R)$ for $T$ and $v\in T$, there is an algorithm $\MakeLean((L,R),G,T)$ that computes a lean $(X,v,k)$-witness $(L^{\star},R^{\star})$ for $T$. The algorithm runs in $O(|G[L]|\cdot k)$ time.
\end{lemma}

$\MakeLean((L,R),G,T)$ simply computes a minimum-size separator $(L',R')$ in $G[L]$ such that $L\cap T\subseteq L'$ and $L\cap R\subseteq R'$ (we slightly abuse the notation by allowing $(L',R') = (L,L\cap R)$). Then let $(L^{\star},R^{\star}):=(L',R'\cup (R\setminus L))$ be the output. Clearly, the running time is $O(|G[L]|\cdot k)$ using the Ford-Fulkerson algorithm.

\begin{lemma}[Lemma 5.9 in \cite{anand2025unbreakable}]%
\label{lemma:UpdateAdhesion}%
Suppose $(L^{\star},R^{\star})$ is a lean $(X,v,k)$-witness for $T$. Then the adhesion of the new terminal set $T' = (T \setminus L^*) \cup (L^* \cap R^*)$ is at most $\max(|L^* \cap R^*|, \sigma(T)) = \max(k,\sigma(T))$.
\end{lemma}

\paragraph{Correctness.} First, by \Cref{lemma:UpdateAdhesion}, the final $T$ has adhesion at most $k$ (note that initially $T=V(G)$ has adhesion $0$). Also, $X\subseteq T$ is easy to see. Hence, we focus on showing that there is no $(X,v,k)$-witness for $T$ and any $v\in T$ at the end (implying that $(X,v,k,k)$-witnesses for $T$ and some $v\in T$ do not exist). 

\begin{lemma}
Throughout the algorithm, $Q\cap T$ is a hitting set for all $(X,v,k)$-witnesses for $T$ across all $v\in T$ in $G$.
\end{lemma}

\begin{proof}

In the beginning, we initialize $Q = V(G)$, so this holds. Consider an iteration. Suppose inductively that this statement is true at the beginning of this iteration, and we want to reprove it at the end of this iteration. There are two cases depending on the output of \Cref{lemma:unconnectedwitness}.

First, consider the case when \Cref{lemma:unconnectedwitness} successfully finds an $(X,v,k)$-witness $(L,R)$ for $T$ in $G$. {\textsc{MakeLean}} then returns another $(X,v,k)$ witness $(L^*, R^*)$, and we update $T$ to $T'$ where $T' = T \setminus L^* \cup (L^* \cap R^*)$. Now consider a minimal $(X,v,k)$ witness $(L',R')$ for $T'$ (it is enough to prove the lemma for minimal witnesses). Without loss of generality, we can assume that $N_G(L' \setminus R') = L' \cap R'$ (for if not, we can simply move vertices from $L' \cap R'$ to $R' \setminus L'$ till this is satisfied). If $L' \setminus R'$ intersects $L^* \cap R^*$, since $L^* \cap R^* \subseteq T'$ was added to $Q$, $(Q \cap T') \cap (L' \setminus R') \neq \emptyset$.

Therefore we can now assume that $L' \setminus R'$ is disjoint with $L^* \cap R^*$. But by the minimality of $(L',R')$, we know that every component of $L' \setminus R'$ intersects $T'$ (\Cref{lemma:TorsoVolume}), and by the execution of the algorithm, $L^* \cap R^* \subseteq T'$ and $(L^* \setminus R^*) \cap T' = \emptyset$. Thus it must be the case that $L' \setminus R'$ is disjoint with $L^* \setminus R^*$ as well. Now if $(L',R')$ was an $(X,v,k)$ witness for $T$, then since we only added more vertices to $Q$ and $Q \cap T$ was already a hitting set for $(X,v,k)$ witnesses for $T$, we must have $(Q \cap T) \cap (L' \setminus R') \neq \emptyset$, and hence $(Q \cap T') \cap (L' \setminus R') \neq \emptyset$. 

Finally, suppose that $(L',R')$ was not previously an $(X,v,k)$ witness for $T$, but is now an $(X,v,k)$ witness for $T'$. In this case, $L' \cap R'$ must intersect $L^* \cap R^*$ in some vertex $x$, such that $x \in T' \setminus T$. Informally, $x$ helps make $(L',R')$ a witness: it was previously a non-terminal, but now has been turned into a terminal. Since $N_G(L' \setminus R') = L' \cap R'$, $x$ has at least one edge incident to $L' \setminus R'$. By minimality of $(L',R')$ and~\Cref{lemma:TorsoVolume}, since every component of $L' \setminus R'$ intersects $T'$, it follows that there is some $y \in (L' \setminus R') \cap T$, such that there is a path from $y$ to $x$ in $G[L']$, whose each internal vertex is not in $T'$. Then $(x,y)$ is an edge in $\Torso(T')$. But this means $y$ was pushed to $Q$ by the execution of the algorithm, and thus $Q$ hits $(L',R')$.  

Now consider the case when we simply pop $v$ from $Q$, with no witness found, and no update to $T$. By the execution of the algorithm, this means that there is no minimal $(X,v,k)$ witness for $T$ in $G$. Since every hitting set for minimal $(X,v,k)$ witnesses is also a hitting set for $(X,v,k)$ witnesses, this means that removing $v$ from $Q$ is safe and $Q$ continues to be a hitting set for all $(X,v,k)$ witnesses for $T$.

\end{proof}

\paragraph{Running Time.}

\begin{lemma}\label{lemma:torsotime}
Suppose we are given an $(X,v,k)$-witness $(L,R)$ for $T$ in $G$, such that $|L|\leq \poly(q,k)$. Further suppose that every connected component of $G \setminus T$ has at most $q$ vertices. Consider the update operation which updates $T$ to $T':=(T \setminus L)\cup (L \cap R)$. Then we can update (the adjacency list of) $\Torso(T)$ to (that of) $\Torso(T')$ in time $\tilde{O}(\poly(q,k))$.
\end{lemma}

\begin{proof}
Because $\Torso(T)$ changes due to the change in the connected components of $G \setminus T$, let us additionally maintain a list of the connected components $S$ of $G\setminus T$, along with their adhesions $N_{G}(S)\subseteq T$. Recall that we define
\[
\Torso(T)= G[T]\cup\bigcup_{\text{comps $S$ of $G\setminus T$}}\Clique(N_{G}(S)). 
\]

We say that a connected component of $G \setminus T$ is \emph{affected} if it has a vertex in $L$. If a component $S$ is not affected, then $N_G(S)$ does not change, so we do not need to update anything for these components. There are at most $\poly(q,k)$ affected components since $|L|\leq \poly(q,k)$. Consider an affected component $S$ in the old list. 
\begin{enumerate}
\item First, we remove $S$ with its adhesion from the list. This may lead to the removal of some edges in $\Clique(N_{G}(S))$.
\item Then, we should add back new components (intuitively, generated by breaking $S$) in $G[S\setminus L]$. Finding these components and their adhesions simply needs searching starting from every vertex in $S\setminus L$. Then add cliques on the new adhesions to the torso graph.
\end{enumerate}
Lastly, there will be new components in $G[L\setminus R]$, which can be handled similar to the Step 2 above. The running time is clearly bounded by $\tilde{O}(\poly(q,k))$.
\end{proof}

\begin{lemma}
\Cref{algo:WitnessEliminate} runs in time $\Otil(nk^{O(k)}\poly(q,k))$.
\end{lemma}

\begin{proof}
First, we will bound the number of pushes into $Q$. Then we will bound the running time for each pop. 

\medskip

\noindent{\underline{Number of pushes to $Q$.} }Every time we push a vertex into $Q$, it is due to an update, which updates $T \gets T \setminus L^* \cup (L^* \cap R^*)$ for some lean $(X,v,k)$ witness $(L^*,R^*)$ satisfying $|L^* \cap T| > |L^* \cap R^*|$. This means that the size of $T$ decreases by at least $1$ in every update, and hence there are at most $n$ updates. For each update, first we push $|L^*| \leq \poly(q,k)$ vertices into $Q$. Then we push for each $u \in U$, the neighbors of $u$ in $\Torso(T)$ into $Q$. The number of such neighbors is at most the time needed to update $\Torso(T)$, and hence by \Cref{lemma:torsotime}, at most $\Otil(\poly(q,k))$. We note here that the precondition of the lemma is fulfilled: every connected component of $G \setminus T$ has size at most $q$. This is because we are guaranteed that there is no $(X,v,q,k)$ witness in $G$ for any $v \in V(G)$, and since the adhesion of $T$ at all times is at most $k$ (\Cref{lemma:UpdateAdhesion}), any component with size more than $q$ would contradict this. Thus there are at most $\Otil(n\poly(q,k))$ pushes into $Q$.

\medskip

\noindent{\underline{Time for each pop.}} The algorithm checks if there is an $(X,v,k)$ witness using~\Cref{lemma:unconnectedwitness}, which takes time $\Otil(k^{O(k)}\poly(q,k))$. \textsc{MakeLean} and the update to $T$ take only time $\Otil(\poly(q,k))$. By~\Cref{lemma:torsotime}, $\Torso(T)$ can be updated in time $\Otil(\poly(q,k))$.

\end{proof}

\section{Key Subroutine for Unbreakable Decomposition.}
\label{sec:keysubroutine}

In this section, we describe the main subroutine that we will use for unbreakable decomposition. The next two lemmas show how to uncross a cut with a downward well-linked set. 

\begin{lemma}
\label{lemma:UncrossOneSet}
Given graphs $H, G$ satisfying 

\begin{enumerate}
\item $G$ is exactly $H \setminus S\, \cup \mathrm{Clique}(N_{H}(S))$
\item $S$ is downward well-linked in $H$.

\end{enumerate}
 Let $X \subseteq T \subseteq V(G)$. If for every $v \in T$ there is no $(X,v,k,k)$ witness for $T$ in $G$, then for every $v \in T$, there is no $(X,v,k,k)$ witness for $T$ in $H$.

\end{lemma}

\begin{proof}
Suppose to the contrary that there is an $(X,v,k,k)$ witness $(L,R)$ for $T$ in $H$. We will ``uncross'' this witness to get an $(X,u,k,k)$ witness for $T$ in $G$ for some $u \in T$. For brevity, we will write $N(S)$ for $N_{H}(S)$.

First, suppose that $N(S) \subseteq L$ or $N(S) \subseteq R$. Consider the cut $(L',R')$ obtained by restricting the cut $(L,R)$ to $V(G)$. Then $|L' \cap R'| \leq |L \cap R|$, $|L \cap T|$ does not change, $X \subseteq R'$ and $v \in L' \setminus R'$. Hence $(L',R')$ is an $(X,v,k,k)$ witness for $T$ in $G$.

Therefore, we now assume that $N(S) \cap (L \setminus R) \neq \emptyset, N(S) \cap (R \setminus L) \neq \emptyset$. Now we consider the restriction of the cut $(L,R)$ to $H[S \cup N(S)]$. Formally, let $A = L \cap (S \cup N(S))$ and $B = R \cap (S \cup N(S))$. $(A,B)$ is a cut in $H[S \cup N(S)]$, since we have $N(S) \cap (L \setminus R) \neq \emptyset, N(S) \cap (R \setminus L) \neq \emptyset$.

By the definition of downward-well-linkedness, the downward-well-linkedness of $S$ implies that either $|A \cap B| \geq |A \cap N(S)|$ or $|A \cap B| \geq |B \cap N(S)|$. By symmetry, assume the former. Consider the separation $(L',R')$ obtained as follows. Let $L' = L \setminus S$ and $R' = R \cup S \cup N(S)$. 

First, we argue that $(L',R')$ is a cut in $H$. Notice that $(L' \setminus R') \setminus S = (L \setminus R) \setminus S$ and $(R' \setminus L') \setminus S = (R \setminus L) \setminus S$. Also $(L' \setminus R') \cap S = \emptyset$. Thus any edge between $L' \setminus R'$ and $R' \setminus L'$ must be from $(R' \setminus L') \cap S$ and some vertex of $L'\setminus R'$. But $L' \setminus R'$ is disjoint from $S \cup N(S)$, hence there cannot be such an edge. Also $L' \cup R' = V(H)$, thus $(L',R')$ is indeed a cut in $H$.

Next, $L' \cap R' = (L \cap R) \setminus S \cup (N(S) \cap (L \setminus R))$. Thus $|L' \cap R'| \leq |L \cap R| - |(L \cap R) \cap S| + |(L \setminus R) \cap N(S)|$. Note that $|A \cap B| \geq |A \cap N(S)|$. But $|A \cap B| = |L \cap R \cap S| + |L \cap R \cap N(S)|$, and $|A \cap N(S)| = |L \cap N(S)|$. Combining, we get $|L \cap R \cap S| \geq |(L \setminus R) \cap N(S)|$, and thus $|L' \cap R'|\leq |L \cap R|$.

Observe that $T \cap S = \emptyset$, hence the set $L' \cap T$ can only grow compared to $L \cap T$, and thus $|L' \cap T| > |L' \cap R'|$ still holds. Now consider the witness $(L'',R'')$ which is the restriction of $(L',R')$ to the graph $G$. Thus $L'' = L'$, and $R'' = R' \setminus S$. Then $|L'' \cap R''| = |L' \cap R'|$, and $|L'' \cap T| = |L' \cap T| > |L' \cap R'| = |L'' \cap R''|$. We also have $X \subseteq R''$. Pick any $u \in (L'' \setminus R'') \cap T$ (such a $u$ exists since $|(L'' \setminus R'') \cap T| \geq |L \cap T| - |L \cap R| > 0$).
Then $(L'',R'')$ is an $(X,u,k,k)$ witness for $T$ in $G$.
\end{proof}

One can generalize the previous lemma to account for multiple steps of uncrossing.

\begin{lemma}\label{lemma:uncrossingwitness}
Given a sequence of graphs $G^{(0)}, G^{(1)} \ldots G^{(t)} = G$ and vertex sets $S_0, S_1 \ldots S_{t-1}$ such that
\begin{enumerate}
\item $G^{(i+1)}$ is exactly $(G^{(i)}\setminus S_{i})\cup \Clique(N_{G^{(i)}}(S_{i}))$, and
\item $S_{i}$ is downward well-linked in $G^{(i)}$.
\end{enumerate}
Let $X \subseteq T \subseteq V(G)$. If for every $v \in T$ there is no $(X,v,k,k)$ witness for $T$ in $G$, then for every $v \in T$ there is no $(X,v,k,k)$ witness for $T$ in $G^{(0)}$.
\end{lemma}

\begin{proof}
The proof follows directly by ``unrolling'' the previous \Cref{lemma:UncrossOneSet}. Assume for contradiction that there is an $(X,v,k,k)$ witness $(L,R)$ for $T$ in $G^{(0)}$. Using \Cref{lemma:UncrossOneSet}, with $S = S_0$, $H = G^{(0)}$ and $G = G^{(1)}$, we obtain an $(X,v_1,k,k)$ witness $(L_1,R_1)$ for $T$ in $G^{(1)}$, and some $v_1 \in T$.

Inductively, using \Cref{lemma:UncrossOneSet} with $S = S_{i-1}$, $H = G^{(i-1)}$ and $G = G^{(i)}$ we obtain an $(X,v_i,k,k)$ witness for $T$ in $G^{(i)}$ for each $i \geq 1$, and some $v_i \in T$. By induction, it follows that there is an $(X,v_t,k,k)$ witness for $T$ in $G^{(t)} = G$ for some $v_t \in T$, a contradiction. This completes the proof.
\end{proof}

\begin{lemma}\label{lemma:unrolladhesion}
Given graphs $H, G$ satisfying 

\begin{enumerate}
\item $G$ is exactly $H \setminus S\, \cup \mathrm{Clique}(N_{H}(S))$
\item $S$ is downward well-linked in $H$.
\item $|N_H(S)| \leq k'$.

\end{enumerate}
 Let $X \subseteq T \subseteq V(G)$. If the adhesion of $T$ in $G$, $\sigma_G(T) \leq k'$, then $\sigma_H(T) \leq k'$.
\end{lemma}

\begin{proof}

First, note that since $N_H(S)$ is a clique in $G$, $N_H(S)$ can intersect at most one connected component of $G \setminus T$. 

Next, note that every connected component of $H \setminus T$ not containing a vertex of $S$ is also a connected component of $G \setminus T$. We will now argue that every connected component $C'$ of $H \setminus T$ intersecting $S$ also satisfies $|N_H(C')| \leq k'$. Fix such a connected component $C'$.

If $N_H(S) \subseteq T$, then $N_H(C') \subseteq N_H(S)$, and $|N_H(S)| \leq k'$.

Otherwise $N_H(S)$ intersects at most one connected component $C$ of $G \setminus T$. Choose some $v \in N_H(S) \cap C$. Since $N_H(S)$ is a clique in $G$, there is an edge between $v$ and every vertex of $N_H(S) \cap T$ in $G$. In particular, since $v \in C$, this means $N_H(S) \cap T \subseteq N_G(C)$. Then $N_H(C') \cap T \subseteq (N_H(S) \cup N_G(C)) \cap T \subseteq N_G(C)$, since $N_H(S) \cap T$ is contained in $N_G(C)$. But $|N_G(C)| \leq \sigma_G(T) \leq k'$.

\end{proof}

\begin{lemma}\label{lemma:unrolladhesionsequence}
Given a sequence of graphs $G^{(0)}, G^{(1)} \ldots G^{(t)} = G$ and sets $S_0, S_1 \ldots S_{t-1}$ such that for each $i$,
\begin{enumerate}
\item $G^{(i+1)}$ is exactly $(G^{(i)}\setminus S_{i})\cup \Clique(N_{G^{(i)}}(S_{i}))$, and
\item $S_{i}$ is downward well-linked in $G^{(i)}$.
\item $|N_{G^{(i)}}(S_i)| \leq k'$.
\end{enumerate}
Let $X \subseteq T \subseteq V(G)$. If the adhesion of $T$ in $G$, $\sigma_G(T) \leq k'$, then $\sigma_{G^{(0)}}(T) \leq k'$.
\end{lemma}

\begin{proof}
Follows directly by repeated application of~\Cref{lemma:unrolladhesion}.

\end{proof}

\begin{lemma}\label{lemma:unbreakabletowitness}
Let $G$ be a graph, $c \geq 1$ be a constant, and let  $X$ be a $(ck,k)$ unbreakable set, and suppose that for $T \subseteq V(G)$, $X \subseteq T$, there is no $(X,v,(c + 1)k,(c + 1)k)$ witness for any $v \in T$ in $G$. Then $T$ is $((c + 1)k,k)$ unbreakable.  
\end{lemma}

\begin{proof}
Suppose that $T$ is not $((c + 1)k,k)$ unbreakable. Then there is a separator $(L,R)$ that satisfies $|L \cap T|, |R \cap T| \geq (c + 1)k+1$ and $|L \cap R| \leq k$. Since $X$ is $(ck,k)$ unbreakable, either $|X \cap L| \leq ck$ or $|X \cap R| \leq ck$. Assume without loss of generality that $|X \cap L| \leq ck$, the other case is symmetric. Now consider the cut $(L', R')$ obtained by setting $L' = L, R' = R \cup (X \cap L)$, and let $v$ be any vertex of $(L' \setminus R') \cap T$. We have $|L' \cap T| \geq |L \cap T| \geq (c + 1)k + 1$, and $|L' \cap R'| = |L \cap R| + |X \cap L| \leq (c + 1)k$. We also have $X \subseteq R'$.

It follows that $(L',R')$ is an $(X,v,(c+1)k,(c+1)k)$ witness in $G$, a contradiction.
\end{proof}

\begin{lemma}\label{thm:reducing_adhesion}
Given a graph $G$, and a $(ck,k)$ unbreakable set $X$ with $|X| \geq (c+1)k + 1$, there is an algorithm that runs in time $\Otil(mk^{O(k)})$ and finds a set $X \subseteq T \subseteq V(G)$ that is $((c+1)k,k)$ unbreakable, and the adhesion of $T$ in $G$ is at most $(c+1)k$.  
\end{lemma}

\begin{proof}
Fix $k' = (c+1)k$ and $q' =200\cdot k'^2 \cdot 8^{k'}$. We call \Cref{thm:connected} with $G,X$ and parameters $ k', q'$. \Cref{thm:connected} runs in time $\Otil(m2^{O(k)})$ and returns a graph $G_{\mathrm{final}}$ satisfying the following properties. 

\begin{enumerate}
\item $X\subseteq V(G_{\final})\subseteq V(G)$.
\item There exists a sequence of graphs $(G^{(0)}:=G, G^{(1)},..., G^{(t)}:=G_{\final})$ satisfying that, for each $0\leq i\leq t-1$, if we let $S_{i}:=V(G^{(i)})\setminus V(G^{(i + 1)})$, then %
\begin{enumerate}
\item $G^{(i+1)}$ is exactly $(G^{(i)}\setminus S_{i})\cup \Clique(N_{G^{(i)}}(S_{i}))$, and
\item $S_{i}$ is downward well-linked in $G^{(i)}$.
\item $|N_{G^{(i)}}(S_i)| \leq k'$.
\end{enumerate}

\item The adhesion of $V(G_{\final})$ in $G$ is at most $k' = (c+1)k$.
\item With high probability, there is no connected $(X,v,q',k'
)$-witness in $G_{\final}$ with volume larger than $q'$, for any vertex $v\in V(G_{\final})$. Let $q = q' + k'$. Then since a connected set with $s$ vertices has volume at least $s$, this means that for any  $(X,v,q,k')$ connected witness $(L,R)$ in $G_{\mathrm{final}}$, we have $\vol(L \setminus R) \geq |L \setminus R| > q - k' = q'$, a contradiction. Hence there is no $(X,v,q,k')$ connected witness in $G_{\mathrm{final}}$. 
\end{enumerate}

We now invoke~\Cref{thm:WitnessEliminate} with graph $G_{\mathrm{final}}$, set $X$ and the parameter $k'$, which then runs in time $\Otil(mk^{O(k)})$ and returns a terminal set $T \subseteq V(G_{\mathrm{final}})$ satisfying
\begin{enumerate}
\item $X \subseteq T$,
\item there is no $(X,v,k',k')$-witness for $T$ and any $v\in T$ in $G_{\mathrm{final}}$, and
\item the adhesion of $T$ in $G_{\mathrm{final}}$, $\sigma(T) \leq k'$.
\end{enumerate}

Applying~\Cref{lemma:uncrossingwitness} to the sequence of graphs $G^{(i)}$, we obtain that there is no $(X,v,k',k')$ witness for $T$ for any $v \in T$ in $G$.

Next,~\Cref{lemma:unbreakabletowitness} shows that if there is no $(X,v,(c+1)k,(c+1)k)$ witness for $T$ for any $v \in T$ in $G$, then $T$ is $((c+1)k,k)$ unbreakable in $G$. By~\Cref{lemma:unrolladhesionsequence}, since the adhesion of $T$ in $G_{\mathrm{final}}$ is at most $(c+1)k$, the adhesion of $T$ in $G^{(0)} = G$ is at most $(c+1)k$ as well. 

This completes the proof.
\end{proof}

%% file: decomposition.tex
\section{The Unbreakable Decomposition Algorithm.}\label{sec:Decomp}

In this section, we show our second main result, a fast algorithm to construct unbreakable decomposition. We will use the following lemmas from~\cite{anand2025unbreakable}.

\begin{lemma}[Lemma 4.5 of~\cite{anand2025unbreakable}]%
Given an $m$-edge graph $G$ with a set $W\subseteq V(G)$ and parameters $q,k$, there is a deterministic algorithm that either
\begin{itemize}
\item certifies that $W$ is $(q,k)$-unbreakable in $G$, or
\item outputs a $(q, k)$-breakable witness $(L,R)$ of $W$ in $G$.
\end{itemize}
The running time is $\OO(2^{|W|}k m)$.
\label{lemma:CheckUnbreakable}
\end{lemma}

\begin{definition}[Balanced Origin]%
Given a graph $G$, an $\alpha$-balanced $\sigma$-origin is a vertex set $X \subseteq V(G)$ such that, for any vertex set $X'\supseteq X$ with adhesion $\sigma$, $X'$ must be $\alpha$-balanced.

\end{definition}

\begin{lemma}[Lemma 4.2 of~\cite{anand2025unbreakable}]%
\label{lemma:SepBySampleSet}%
Given an $m$-edge graph $G$ with parameters $k,\sigma$ where $k\leq \sigma$, there is a randomized algorithm that computes a set $X$ which is always $(k,k)$-unbreakable such that with constant probability, $X$ is a  $\frac{1}{2}$-balanced $\sigma$-origin. The running time is $2^{\OO(\sigma)}m$.

\end{lemma}

The following lemma combines our result in~\Cref{thm:reducing_adhesion} with the previous lemma, and shows how to obtain Unbreakable Balanced Low-Adhesion Sets, an analogue of Theorem 5.2 of~\cite{anand2025unbreakable}.

\begin{lemma}[Unbreakable Balanced Low-Adhesion Sets]%
\label{coro:UnbreakLowAdhesion}%
Given an $m$-edge graph $G$ with parameter $k\geq 1$, there is a randomized algorithm that with high probability computes a set $X\subseteq V(G)$ such that $X$ is $(2k, k)$-unbreakable, $1/2$-balanced and has adhesion $2k$.
The running time is $k^{O(k)}m\,\polylog(n)$. 
\end{lemma}
\begin{proof}
We do the following $O(\log n)$ times. First, let $X_{0}$ be an initial set by applying \Cref{lemma:SepBySampleSet} on $G$ with parameters $k$ and $\sigma = 2k$. If $|X_0| \leq 2k$, check if $X_0$ is $\frac{1}{2}$ balanced, if so return it, else continue to the next iteration. Now we can assume that $|X_0| \geq 2k + 1$, and we compute $X$ by applying \Cref{thm:reducing_adhesion} on graph $G$ with $X_{0}$ as the initial set. If $X$ is $1/2$-balanced, we return $X$ and terminate the whole algorithm, otherwise proceed to the next iteration.

To see the correctness, consider each of the $O(\log n)$ iterations. \Cref{lemma:SepBySampleSet} guarantees that $X_{0}$ is always $(k,k)$-unbreakable and with constant probability, $X_{0}$ is a $1/2$-balanced $\sigma$-origin. Therefore, \Cref{thm:reducing_adhesion} guarantees that, with high probability, $X\supseteq X_{0}$ is $(2k,k)$-unbreakable and has adhesion at most $\sigma = 2k$. Furthermore, because $X_{0}$ is a $1/2$-balanced $\sigma$-origin with constant probability, we have $X$ is $1/2$-balanced with constant probability. It follows that with high probability, at least one of the $O(\log n)$ iterations will give a $1/2$-balanced $X$.
\end{proof}

The following theorem is the main theorem of this section which constructs an unbreakable decomposition. 

\begin{theorem}
\label{thm:DetailedMain}
Given an $n$-vertex $m$-edge undirected graph $G$ with a parameter $k$, there is a randomized algorithm that, with high probability,\footnote{In recursive calls on subgraphs, all random sampling and repetition counts use $n=|V(G)|$ of the input graph, so a union bound over all calls applies.} computes a $(5k,k)$-unbreakable decomposition with adhesion $4k$, which further has the following properties.  
\begin{itemize}
\item The unbreakable decomposition is compact and admits subtree unbreakability.
\item The depth of the decomposition is $O(k\log n)$.
\item The number of tree nodes is $O(n)$ and the total bag size is $O(kn)$.
\end{itemize}
The running time is $k^{O(k)}\Otil(m)$. Alternatively, the algorithm can compute an $(O(k),k)$-unbreakable decomposition with adhesion $O(k)$, depth $O(\log n)$, and all the other properties as above in asymptotically the same running time.
\end{theorem}

We now show the algorithm for computing an unbreakable decomposition with $O(k\log(n))$ depth. Later in \Cref{sect:DepthReduction}, we will discuss how to reduce the depth to $O(\log n)$ by modifying the algorithm slightly.

\begin{proof}[Proof of \Cref{thm:DetailedMain}]
We fix a parameter $\sigma = 2k$ (which is an upper bound of the adhesion of the set returned by \Cref{thm:reducing_adhesion}). We will design a subroutine $\Decomp(H,B)$, which will receive as input an undirected graph $H$ with a boundary vertex set $B\subseteq V(H)$ such that
\begin{enumerate}
\item\label{cond:Decomp1} $|B| \leq 2\sigma$,
\item\label{cond:Decomp2} $H\setminus B$ is connected, and
\item\label{cond:Decomp3} $N_{H}(V(H)\setminus B) = B$.
\end{enumerate}
We remark that condition \ref{cond:Decomp1} is relatively important, while conditions 2 and 3 are only for avoiding some corner cases (so the reader may ignore them for now and come back when they are used). 
The output of $\Decomp(H,B)$ is an unbreakable decomposition $(T,\beta)$ of $H$ in which the root bag $\beta(r)$ contains $B$. Then the desired unbreakable decomposition of $G$ can be obtained by invoking $\Decomp(G,\emptyset)$. Without loss of generality, we assume $G$ is connected (otherwise, just work on each connected component separately). Thus it satisfies conditions \ref{cond:Decomp2} and \ref{cond:Decomp3}. 

In what follows, we first describe $\Decomp(H,B)$ and then show the output guarantees of this subroutine.

\begin{algorithm}[h]
\caption{$\Decomp(H,B)$}
\label{alg:Decomp}
\begin{algorithmic}[1]
\Require Parameters $k\geq 1$ and $\sigma = 2k$. A graph $H$ with a boundary vertex set $B\subset V(H)$ such that $|B|\leq 2\sigma$.
\Ensure A rooted tree decomposition $(T,\beta)$ of $H$ such that $B$ is contained in the root bag.
\If{$\sigma + 1\leq |B|\leq 2\sigma$}
\State\label{line:CheckUnbreakable}Apply \Cref{lemma:CheckUnbreakable} on $B$ and $H$ to check if $B$ is $(k,k)$-unbreakable in $H$.
\If{\Cref{lemma:CheckUnbreakable} returns a $(k,k)$-breakable witness $(L,R)$ of $B$ in $H$}
\State $X_{1}\gets L\cap R$ and $X\gets B\cup X_{1}$.
\Else \Comment{$B$ is $(k,k)$-unbreakable in $H$}
\State $X\gets $ apply \Cref{thm:reducing_adhesion} on graph $H$ with $B$ as the initial set.
\EndIf
\Else \Comment{$|B|\leq \sigma$}
\State $X_{1}\gets$ apply \Cref{coro:UnbreakLowAdhesion} on graph $H$ with parameter $k$.
\State $X\gets B\cup X_{1}$.
\EndIf
\If{$X=B$}
\State\label{line:corner}Return a trivial decomposition $(T,\beta)$ with a single bag $V(H)$.
\Else \Comment{$X\supset B$.}
\State Initialize $(T,\beta)$ with only the root node $r$ whose bag is $\beta(r) = X$.
\For{each connected component $D$ of $H\setminus X$}
\State $(T_{D},\beta_{D})\gets \Decomp(H[D\cup N_{H}(D)]\setminus E(H[N_{H}(D)]), N_{H}(D))$
\State Update $(T,\beta)$ by appending $(T_{D},\beta_{D})$ to the root node $r$.
\EndFor
\State Return $(T,\beta)$.
\EndIf
\end{algorithmic}
\end{algorithm}

\paragraph{Description of $\Decomp(H,B)$.} See \Cref{alg:Decomp} for a pseudocode. We now explain the algorithm below. 

\medskip

\noindent{\underline{Compute a ``separator'' $X$.}} We first compute a set $X\subseteq V(H)$ according to the following two cases.

\medskip

\noindent{\textbf{Case 1.}} If $\sigma+1\leq |B|\leq 2\sigma$, we will compute an unbreakable set $X\supseteq B$ with adhesion smaller than $|B|$ to reduce the boundary size of the subproblems.
\begin{itemize}
\item[(a)] If $B$ has a $(k,k)$-breakable witness in $H$, find such a witness $(L,R)$ using \Cref{lemma:CheckUnbreakable} (let $X_{1} = L\cap R$), and we let $X = X_{1}\cup B$.
\item[(b)] Otherwise, $B$ is $(k,k)$-unbreakable in $H$. We apply \Cref{thm:reducing_adhesion} on graph $H$ with $B$ as the initial set, and let $X$ be the output set.
\end{itemize}

\begin{claim}
\label{claim:DecompCase1}
In Case 1, the set $X$ is $(5k,k)$-unbreakable and has adhesion at most $|B|-1$.
\end{claim}
\begin{proof}
We first focus on case (a). $X$ is trivially $(5k,k)$-unbreakable because $|X|\leq |X_{1}| + |B|\leq k + 2\sigma  \leq 5k$. To see that $X$ has adhesion $|B|-1$, consider an arbitrary connected component $D$ of $H\setminus X$. Recall that $X = X_{1}\cup B$, so
\[
|N_{H}(D)| = |N_{H}(D)\cap X_{1}| + |N_{H}(D)\cap (B\setminus X_{1})|,
\]  
and we will bound $|N_{H}(D)\cap X_{1}|$ and $|N_{H}(D)\cap (B\setminus X_{1})|$ separately. Observe that $|N_{H}(D)\cap X_{1}|\leq |X_{1}|\leq k$, where $|X_{1}|\leq k$ is trivially because $X_{1}$ is a $(k,k)$-breakable witness. To bound $|N_{H}(D)\cap (B\setminus X_{1})|$, we look at the connected component $D'\supseteq D$ of $H\setminus X_{1}$ (note that $D'$ is unique since $X_{1}\subseteq X$). We have $N_{H}(D)\cap (B\setminus X_{1})\subseteq D'\cap B$ because $N_{H}(D')\subseteq X_{1}$. Furthermore, $|D'\cap B|\leq |B|-(k+1)$ because $X_{1}$ is a $(k,k)$-breakable witness of $B$ in $H$. Therefore, $|N_{H}(D)\cap (B\setminus X_{1})|\leq |B|-(k+1)$. Putting it all together, $|N_{H}(D)| \leq k + |B|-(k+1)\leq |B|-1$ as desired.

For (b), the claim holds simply by the guarantees of \Cref{thm:reducing_adhesion}, which guarantees that $X$ is $(2k,k)$-unbreakable and has $\sigma(X) \leq 2k$ (which implies $\sigma(X)\leq \sigma\leq |B|-1$). 
\end{proof}

\medskip

\noindent{\textbf{Case 2.}} If $|B|\leq \sigma$, we will compute a balanced unbreakable set $X\supseteq B$ (with adhesion at most $2\sigma$) to significantly reduce the graph size of the subproblems. We first compute $X_{1}$ by applying \Cref{coro:UnbreakLowAdhesion} on graph $H$ with parameter $k$, and then let $X$ be $B\cup X_{1}$.

\begin{claim}
In Case 2, the set $X$ of $H$ is $(4k,k)$-unbreakable and $1/2$-balanced, and has adhesion at most $4k$.
\label{claim:DecompCase2}
\end{claim}
\begin{proof}

\Cref{coro:UnbreakLowAdhesion} guarantees that, with high probability, $X_{1}$ is $(2k,k )$-unbreakable with adhesion at most $2k$.
Because $X = X_{1}\cup B$, $X$ is a $(2k + |B|,k)$-unbreakable and $1/2$-balanced set in $H$ with adhesion at most $2k + |B|$. Recall that $|B|\leq \sigma = 2k$, so the claim holds.
\end{proof}

\noindent{\underline{Construct the decomposition.}} After computing the set $X$, we construct the decomposition $(T,\beta)$ as follows. We first consider the corner case that $X = B$, in which we return a trivial decomposition with $V(H)$ as the only bag. 

\begin{claim}
In the corner case $X = B$, we have $|V(H)|\leq 2\sigma$, which implies $V(H)$ is $(4k,k)$-unbreakable in $H$.
\label{claim:CornerCase}
\end{claim}
\begin{proof}
Note that only case 2 can lead to this corner case. Because from the input conditions \ref{cond:Decomp2} and \ref{cond:Decomp3}, the adhesion of $X$ is exactly $|B|$, contradicting \Cref{claim:DecompCase1} for case 1. Therefore, $|X| = |B|\leq \sigma$ and \Cref{claim:DecompCase2} tells that $X$ is $1/2$-balanced. Combining that there is only one connected component $D$ in $H\setminus X$ (by input condition \ref{cond:Decomp2}), we have $|V(H)|\leq |D| + |X|\leq |V(H)|/2 + \sigma$, which means $|V(H)| \leq 2\sigma$. Hence $V(H)$ is trivially $(2\sigma,k)$-unbreakable in $H$.
\end{proof}

From now we assume $X\supset B$. We consider each connected component $D$ of $H\setminus X$, and solve a subproblem $\Decomp(H[D\cup N_{H}(D)]\setminus E(H[N_{H}(D)]),N_{H}(D))$\footnote{In other words, $H[D\cup N_{H}(D)]\setminus E(H[N_{H}(D)])$ is the subgraph of $H$ induced by vertices $D\cup N_{H}(D)$ excluding edges with both endpoints in $N_{H}(D)$.} recursively which gives output $(T_{D},\beta_{D})$ (note that this will ensure the graphs of the subproblems are edge-disjoint). After solving all the subproblems, we create a node $r$ as the root of $T$ corresponding to bag $\beta(r) = X$, and then append all $(T_{D},\beta_{D})$ to $r$ as \emph{child-subtrees}.

\paragraph{Correctness.} Let $(T,\beta)$ be the output of $\Decomp(H,B)$.

\medskip

\noindent{\underline{$(T,\beta)$ is an unbreakable decomposition.}} We now prove by induction that $(T,\beta)$ from $\Decomp(H,B)$ is a $(
5k,k)$-unbreakable decomposition of $H$ with adhesion $4k$, further satisfying that the root bag contains $B$.

We first consider the base case that $\Decomp(H,B)$ causes no further recursion, which means the set $X$ is trivially the whole $V(H)$. Then trivially the output $(T,\beta)$ is a tree decomposition with only a single bag, and this bag is $(5k,k)$-unbreakable in $H$ by \Cref{claim:DecompCase1} and \Cref{claim:DecompCase2}.

Next, we consider the recursive subproblems that $\Decomp(H,B)$ invokes. Recall that each recursion $\Decomp(H[D\cup N_{H}(D)]\setminus E(H[N_{H}(D)]), N_{H}(D))$ relates to a connected component $D$ of $H\setminus X$, and we assume all these recursions output child-subtrees $(T_{D},\beta_{D})$ correctly. We verify the properties in \Cref{def:TreeDecomp,def:UnbreakableDecomp} one by one to show the correctness of $\Decomp(H,B)$.
\begin{itemize}
\item Consider a vertex $v\in V(H)$. If $v\notin X$, then $v$ is inside a unique connected component $D_{v}$ of $H\setminus X$, so only the child-subtree $(T_{D_{v}}, \beta_{D_{v}})$ may have bags containing $v$, and by the correctness of recursions, bags containing $v$ induce a connected subtree. If $v\in X$, any child-subtree $(T_{D},\beta_{D})$ with bags containing $v$ must relate to a component $D$ s.t. $v\in N_{H}(D)\subseteq X$. Hence, $v$ is inside the root bag of such child-subtrees $(T_{D},\beta_{D})$ (since $N_{H}(D)$ are the boundary vertices of the next recursion). By the construction of $(T,\beta)$, bags containing $v$ induce a connected subtree.

\item Each edge in $E(H)$ will be covered by at least a bag in $(T,\beta)$, because edges in $H[X]$ are inside the root bag of $(T,\beta)$, and each of the rest of edges is inside some $H[D\cup N_{H}(D)]\setminus E(H[N_{H}(D)])$.

\item All bags of $(T,\beta)$ are $(5k,k)$-unbreakable in $H$ by the following reasons. The new bag $X$ is $(5k,k)$-unbreakable in $H$ by \Cref{claim:DecompCase1} and \Cref{claim:DecompCase2}. Each old bag is even $(5k,k)$-unbreakable in some subgraph of $H$.

\item To see that $(T,\beta)$ has adhesion $2\sigma = 4k$, it suffices to show that the root bag $X$ of $(T,\beta)$ has an intersection of size $2\sigma$ with the root bag of each child-subtree $(T_{D},\beta_{D})$. Indeed, observe that the intersection is exactly $N_{H}(D)$, so it has size at most $2\sigma$ (since $X$ has adhesion $2\sigma$ in $H$ by \Cref{claim:DecompCase1} and \Cref{claim:DecompCase2}).
\end{itemize}

\medskip

\noindent\underline{The tree size and the total bag size.} Next, we prove that $(T,\beta)$ from $\Decomp(H,B)$ satisfies (1) $|V(T)|\leq |V(H)|$, (2) the total bag size is at most $O(|V(H)|\sigma) = O(|V(H)|k)$. Recall that the definition of tree decomposition guarantees the $\beta(t)\setminus \sigma(t)$ for all $t\in V(T)$ are disjoint. Hence (1) holds because each $t\in V(T)$ has non-empty $\beta(t)\setminus \sigma(t)$ by \Cref{ob:TreeSize}. Next, (2) is because $\sum_{t\in V(T)}|\beta(t)| = \sum_{t\in V(T)}|\beta(t)\setminus \sigma(t)| + \sum_{t\in V(T)}|\sigma(t)| \leq |V(H)| + |V(T)|\cdot (2\sigma) = O(|V(H)|\sigma)$.

\begin{observation}
For each tree node $t\in V(T)$, its bag $\beta(t)$ is a strict superset of its adhesion $\sigma(t)$, i.e. $\beta(t)\supset \sigma(t)$.
\label{ob:TreeSize}
\end{observation}
\begin{proof}
To see this, in the recursion step creating $t$, we have $\sigma(t)= B$. Furthermore, $\beta(t) = V(H)$ if $X = B$ (the corner case) and $\beta(t) = X$ if $X \supset B$, so always $\beta(t)\supset \sigma(t)$. 
\end{proof}

\medskip

\noindent{\underline{The tree depth.}} We are now ready to show the depth of $T$ is $O(k\log|V(H)|)$. Consider a path from the root (upside) to a leaf (downside) in the recursion tree. We divide this path into two parts, the upper part and the lower part, at the uppermost step $\Decomp(H',B')$ such that $|V(H')|\leq 10k$. 

For the lower part, i.e. from the step $\Decomp(H',B')$ to the leaf, trivially there are at most $|V(H')| = O(k)$ steps (note that the tree decomposition of $\Decomp(H',B')$ has at most $|V(H')|$ nodes).

Consider the upper part. Each step $\Decomp(H',B')$ on this part has $|V(H')|>10k$. The number of case-2 steps on the upper part is at most $O(\log n)$ by the following reasons. Let $\Decomp(H',B')$ be a case-2 step and $\Decomp(H'',B'')$ be its sub-step on this path. Because $X$ is $1/2$-balanced and has adhesion $2\sigma$ by \Cref{claim:DecompCase2}, we have $|V(H'')|\leq |V(H')|/2 + 2\sigma$. Moreover, $2\sigma = 4k < 2|V(H')|/5$, so $|V(H'')| < 9|V(H')|/10$. Next, the number of \emph{consecutive} case-1 recursive steps is at most $\sigma$, because the earliest case-1 step has boundary size at most $2\sigma$ and the subsequent case-1 steps will keep reducing the boundary size by \Cref{claim:DecompCase1}, until the boundary size drops below $\sigma + 1$. Therefore, the total length of the upper part is at most $O(\log n)\cdot \sigma = O(k\log n)$.

In summary the recursion depth is at most $O(k\log n) + O(k) = O(k\log n)$.

\medskip

\noindent{\underline{Compactness and subtree unbreakability}.} Finally, the compactness and subtree unbreakability of $(T,\beta)$ simply follow from the algorithm.

\paragraph{Running Time of $\Decomp(G,\emptyset)$.} 

Observe that, at each level $i$, the total graph size of level-$i$ steps is $O(m)$, because the recursive steps $\Decomp(H_{i}, B_{i})$ at a level $i$ are on edge-disjoint graphs $H_{i}$. Because there are $O(k\log n)$ levels, the total graph size across all levels is $O(km\log n)$. In each recursive step $\Decomp(H,B)$, invoking \Cref{lemma:CheckUnbreakable,lemma:SepBySampleSet,thm:reducing_adhesion} takes time $k^{O(k)} \Otil(m)$. In summary, the total running time of $\Decomp(G,\emptyset)$ is $k^{O(k)} \Otil(m)$.
\end{proof}

\subsection{Depth Reduction.}
\label{sect:DepthReduction}

Finally, we discuss how to reduce the depth of the unbreakable decomposition from $O(k\log n)$ to $O(\log n)$, at a cost of bringing constant factors to the unbreakability and adhesion. Precisely, we will show how to compute an $(O(k),k)$-unbreakable decomposition with adhesion $O(k)$ and depth $O(\log n)$.

We will change \Cref{alg:Decomp} slightly as follows.
\begin{enumerate}
\item We will set $\sigma = 5k$ (instead of $\sigma = 2k$).
\item Before Line 1, we add a termination condition: if $|V(H)|\leq 40k$, we return a trivial decomposition $(T,\beta)$ with a single bag $V(H)$. We point out that after adding this new termination condition, the algorithm will never go into \Cref{line:corner}, because by \Cref{claim:CornerCase}, there should be $|V(H)|\leq 2\sigma = 10k$ when we reach \Cref{line:corner}, which is impossible because of the new termination condition.
\item In \Cref{line:CheckUnbreakable}, we will check if $B$ is $(2k, k)$-unbreakable (instead of $(k,k)$-unbreakable). This will lead to a new \Cref{claim:DecompCase1}:
\begin{itemize}
\item In Case 1, the set $X$ is $(O(k),k)$-unbreakable and has adhesion at most $|B| - k$.
\end{itemize}
The new claim can be shown in the same way as \Cref{claim:DecompCase1}. Roughly speaking, in case (a), $X$ is $(k+2\sigma,k)$-unbreakable and has adhesion at most $|B| - k$. In case (b), $X$ is $(3k, k)$-unbreakable and has adhesion $3k < |B| - k$ (as $|B|\geq \sigma+1 = 5k+1$).
\end{enumerate}

Following the same argument in the proof of \Cref{thm:DetailedMain}, it is not hard to see that the resulting decomposition is $(O(k),k)$-unbreakable and has $O(k)$ adhesion. It also has all the other properties stated in \Cref{thm:DetailedMain}. In particular, the depth reduces to $O(\log n)$ by the following reasons. Consider a path from the root to an arbitrary leaf. The number of case-2 steps is at most $O(\log n)$, since each case-2 step gives $|V(H'')|\leq |V(H')|/2+2\sigma<3|V(H')|/4$ as $2\sigma=10k<|V(H')|/4$ (note that there is no lower part because of the new termination condition). The number of consecutive case-1 steps is at most $O(\sigma/k) = O(1)$, because every case-1 step will reduce the adhesion by $k$ (from modification 3).

\section{Single Sink Important Separators in Single Exponential Near-Linear Time.}\label{sec:singlesink}

In this section, we prove~\Cref{thm:singlesink}. We use the following algorithm.

\begin{algorithm}
\begin{algorithmic}
\caption{Single Sink Important Separators}\label{algo:singlesink}
\State Initialize $\mathcal{C} = \emptyset$
\For{$\nu = m, m/2 \ldots 1$}
\State Sample $R_{\nu} = O(\frac{mk \log n}{\nu})$ vertices at random, proportional to degree.
\For{each $v \in R_{\nu}$}
\State Call~\Cref{thm:breakable} with $G, s = v, t = t, \nu = 2\nu$ and $k = k$. Add the output separators to $\mathcal{C}$.

\EndFor
\EndFor

\State For each $(L,R)$ in $\mathcal{C}$ if there is some $(L',R')$ with $|L' \cap R'| \leq |L \cap R|$ such that $L\setminus R \subset (L' \setminus R')$, then remove $(L,R)$ from $\mathcal{C}$.\label{line:clean}
\Return $\mathcal{C}$.
\end{algorithmic}
\end{algorithm}

\begin{lemma}
The preceding algorithm correctly computes the set of all $(v,t)$ important separators of size at most $k$ across all $v \in V(G)$, $v \neq t$. 
\end{lemma}

\begin{proof}
Let us first show that every important separator of size at most $k$ is listed by the algorithm at some point.
 Fix a $(v,t)$ important separator $(L,R)$ with volume $\nu'$, such that $\nu \leq \nu' \leq 2\nu$, for some value $\nu$ considered by the algorithm. First, note that by~\Cref{lemma:sampleset}, $R_{\nu}$ hits $L \setminus R$ for every such $(v,t)$ important separator of size at most $k$ and volume at least $\nu$. Let $u \in R_{\nu} \cap (L \setminus R)$. Then, by a standard argument $(L,R)$ is also a $(u,t)$ important separator.

It follows that $(L,R)$ must be listed by~\Cref{thm:breakable}.

Finally, if $(L,R)$ is not important, then there exists another important separator $(L',R')$ such that $L \setminus R \subset L' \setminus R'$. It follows that $(L,R)$ would be removed by the condition of~\Cref{line:clean}. 

\end{proof}

\begin{lemma}
\Cref{algo:singlesink} runs in time $2^{O(k)}m\,\polylog(n)$.
\end{lemma}

\begin{proof}
Clearly, there are at most $O(mk\log n/\nu)$ invocations of \Cref{thm:breakable} in iteration $\nu$, each of which takes time $\Otil(\nu 2^{O(k)})$ and there are at most $O(\log n)$ iterations. Thus the total time across iterations for invocations of \Cref{thm:breakable} is bounded by $\Otil(m2^{O(k)})$.

It remains only to implement and bound the time of the step that checks if there is some $(L',R')$ with $|L' \cap R'| \leq |L \cap R|$ such that $L \setminus R \subseteq L' \setminus R'$.

This can be implemented easily. We sort the $(L,R)$ in decreasing order of size of $L \setminus R$. Let $(L_i, R_i)$ denote the cut at index $i$ after sorting. Then for each $v$, we keep track of the index set $I_v$, such that $v \in L_i \setminus R_i$ for each $i \in I_v$. At all times, we will ensure the invariant that every member of $I_v$ is a $(v,t)$ important separator and hence $|I_v| \leq 4^k$.

We process $(L_i,R_i)$ in increasing order of $i$. When processing $(L_i, R_i)$, we pick an arbitrary $v \in L_i \setminus R_i$, and check if there is some index $j$ in $I_v$ such that $L_i \setminus R_i \subset L_j \setminus R_j$, and $|L_j \cap R_j| \leq |L_i \cap R_i|$. For a given $j$ this can be accomplished in $|L_i \setminus R_i|$ time, and there are at most $4^k$ indices in $I_v$. If there is such a $j$, we do nothing. If there is no such $j$, we add $i$ to $I_v$ for every $v \in L_i \setminus R_i$.

It follows that the whole process can be accomplished with the running time of $$\sum_{X: X = L \setminus R, (L,R) \in \mathcal{C}} |X|4^k \leq \Otil(m \cdot 16^k)$$, since the total size of enumerated $(L,R)$ is bounded by $\Otil(2^{O(k)}m)$, which in turn follows from the fact that the running time of all the invocations of~\Cref{thm:breakable} is bounded by $\Otil(m \cdot 4^k)$.

Finally, we need to show the invariant that $I_v$ consists of exactly those indices $i$ such that each $(L_i, R_i)$ is an important $v$-$t$ separator indeed holds at all times. For each $v$, at the beginning, the invariant is clearly true since $I_v$ is empty. Inductively, assume that the invariant holds for all $v \in V(G)$ after $(L_{i-1}, R_{i-1})$ is processed.

When $(L_i, R_i)$ is processed, suppose the algorithm adds $i$ to $I_v$. We claim that $(L_i, R_i)$ must be a $v$-$t$ important separator. Otherwise, there exists some $j \in I_v$ such that $(L_j, R_j)$ is an important $v$-$t$ separator and both $L_i \setminus R_i \subset L_j \setminus R_j$ and $|L_j \cap R_j| \leq |L_i \cap R_i|$ . Notice that in this case, $(L_j,R_j)$ is also an $(x,t)$ important separator for every $x \in L_j \setminus R_j$, which in particular means that $(L_j,R_j)$ is also an $(x,t)$ important separator for every $x \in L_i \setminus R_i$. But in this case, the algorithm picked some $u \in L_i \setminus R_i$, and we must have $j \in I_u$, and further $L_i \setminus R_i \subset L_j \setminus R_j$. But this means the algorithm cannot have added $i$ to $I_v$, a contradiction.

 Finally, since $I_v$ only consists of important $(v,t)$ separators of size at most $k$, we must have $|I_v| \leq 4^k$ at all times.

\end{proof}

\section{Almost-linear time FPT algorithm for Edge Multicut.}\label{sec:multicut}

On a high level, our algorithm follows the approach of~\cite{chu2026faster}. For the ease of readability, we provide an overview of our algorithm. The key innovation of our algorithm is faster shadow removal, which we can now accomplish in near-linear FPT time thanks to our result on single-sink important separators. Another key difference is that we avoid iterative compression and give a new approximation algorithm instead. 

\paragraph{Overview of the algorithm.}
In the \textsc{Edge Multicut} problem, the input is a graph
\(G=(V,E)\), a set of terminal pairs
\[
    \mathcal T=\{(s_i,t_i): i\in[q]\},
\]
and an integer \(k\).  The goal is to decide whether there is an edge set
\(X\subseteq E(G)\), \(|X|\le k\), such that, in \(G-X\), no terminal pair \(s_i,t_i\)
remains connected.  Let \(N=|V|+|E|+|\mathcal T|\).  We describe a reduction chain in which
each step creates only \(f(k)\) many instances and preserves near-linear size, so that the total
running time is \(f(k)N^{1+o(1)}\).

\smallskip
\noindent\textbf{Step 1: approximate solution and compression.}
We start by obtaining an approximate solution. Consider the standard LP relaxation for multicut. 
\[
\begin{array}{ll}
\min & \displaystyle \sum_{e\in E} x_e \\[0.5em]
\text{s.t.} & \operatorname{dist}_x(s_i,t_i) \ge 1
\qquad \forall i\in[q], \\[0.3em]
& x_e \ge 0
\qquad \forall e\in E .
\end{array}
\]

Here $\mathrm{dist}_x$ means the distance according to the lengths $x_e$ on the edges. We first solve this LP approximately using the result of~\cite{haeupler2024low}. We include a proof of the following \Cref{thm:FastMultiCutLPSolver} in \Cref{sect:FastMulticut}.

\begin{theorem}[\cite{haeupler2024low}]%
\label{thm:FastMultiCutLPSolver}%
For any $\epsilon > 0$, assuming $\OPT^{\mathrm{MC}}_{\mathrm{LP}}\le k$, one can obtain a  $2^{O(1/\epsilon)}$ approximation to the standard LP in time $k\cdot N^{1 + \mathrm{poly}(\epsilon)}$.
\end{theorem}

 Next, we show how to round this LP by incurring only an $O(\log k)$ approximation factor. To the best of our knowledge such a rounding algorithm is not previously known. Our framework essentially follows directly the region growing method of Lee and Wahlström~\cite{lee2016lp}. This gives us a rounding algorithm with $O(\log k)$ loss, which returns an edge multicut $F$. Since the rest of the problem will be solved using the vertex version, we employ a simple reduction: we make $k + 1$ copies for every vertex, and subdivide edges and add a new vertex. $F$ then becomes a vertex multicut for this instance of the same size. For simplicity, let us abuse the notation and still denote the resulting graph by $G$, and let $W_0$ be the vertex multicut approximate solution corresponding to $F$ in $G$.

 Suppose $X$ is the hidden optimal vertex multicut of size at most $k$ in $G$. We guess the vertices in $X \cap W_0$, and guess the partition of $W_0$ in $G \setminus X$, and contract each partition into a single vertex. This gives a set of contracted vertices $W$. Henceforth, our goal is to find a  vertex multicut solution $X$, with $|X| \leq k$ disjoint from $W$ which is also a multiway cut for $W$.

 Thus the problem reduces to \textsc{Multicut Compression}. An instance is of the form $(G, \mathcal{T},W,k)$, the goal is to find a set of vertices $X$ of size at most $k$ that is a multicut for pairs in $\mathcal{T}$ and a multiway cut for $W$. Note that the size of $W$ is at most the size of our approximate solution, which is of size at most $O(k\log k)2^{O(1/\epsilon)}$ (we will set $\epsilon$ later in the formal analysis. Essentially, $\epsilon$ is chosen so that $\frac{1}{\epsilon}$ is a very small growing function of $n$, this ensures that the running time is $m^{1 + \mathrm{poly}(\epsilon)} = m^{1 + o(1)}$).

\smallskip
\paragraph{Step 2: reduction to a shadowless instance.}For a solution $X\subseteq V(G)\setminus W$, its shadow is
\[
\operatorname{sh}_G(X)
=
\left\{
v\in V(G)\setminus X :
v \text{ is not connected to any vertex of } W \text{ in } G-X
\right\}.
\]
Thus, $X$ is shadowless if every vertex surviving after deleting $X$
belongs to a component containing a terminal of $W$.

Following the shadow-removal framework of~\cite{marx2011fixed,chu2026faster},
we may choose a solution $X$ of size at most $k$ that is
$k$-shadow-removable. In particular, for every
$v\in\operatorname{sh}_G(X)$, there is an important $(v,W)$-separator
of size at most $k$ that is contained in $X$. Moreover, for every
$x\in X$, every important $(x,W)$-separator of size at most $k$
contains a vertex outside $X$.

We use color coding to construct a family of red-blue colorings.
For the correct coloring, every vertex of $X$ is red, while a set of
blue vertices intersects every important $(x,W)$-separator of size at
most $k$, for every $x\in X$. For such a coloring, let
\[
Z
=
\left\{
v\in V(G)\setminus W :
\text{some important $(v,W)$-separator of size at most $k$ is entirely red}
\right\}.
\]
The first property of $X$ implies that
\[
\operatorname{sh}_G(X)\subseteq Z,
\]
whereas the blue hitting sets imply that
\[
Z\cap X=\emptyset.
\]
We then take the torso over $Z$. Since $Z$ contains the shadow of $X$
and is disjoint from $X$, the same solution $X$ becomes shadowless in
the resulting instance, while every solution of the torso instance
lifts to a solution of the original instance.

Computing $Z$ for a fixed coloring amounts to enumerating all
$(v,W)$-important separators of size at most $k$. After contracting
$W$ to a single sink, this is exactly the single-sink important
separator problem, and hence can be done in
$\widetilde O(2^{O(k)}m)$ time using our algorithm.

\noindent\textbf{Step 3: reduction to a bipedal instance.}
We next reduce the shadowless compression instance to a \emph{bipedal} one.  A compression
instance is called bipedal if every connected component \(C\) of \(G-W\) has at most two neighbors
in \(W\), that is,
\[
    |N_G(C)\cap W|\le 2.
\]
Intuitively, once the instance is bipedal, every surviving nonterminal component can attach to
only one or two terminals, which makes the remaining structure essentially Boolean.  To obtain
bipedality, we follow exactly the branching framework of~\cite{chu2026faster}. The only technical change is that we cannot afford to make cliques when we torso a set of vertices. Hence we use a sparse torso instead, that replaces cliques by $k$-connected subgraphs, which are equivalent for any cuts of size $k$.

\smallskip
\noindent\textbf{Step 4: solving bipedal instances by \textsc{Almost 2-SAT}.}
It remains to solve a bipedal compression instance. For this, we follow the algorithm of~\cite{marx2011fixed,chu2026faster} that reduces to \textsc{Almost 2-SAT}. It is easy to show that this formula is of near-linear size, and we can use existing fast algorithms for \textsc{Almost 2-SAT}~\cite{iwata2014linear} to solve the resulting instance.

Since the Steps $1$ and $2$ are the only new steps, whereas $3$ and $4$ mostly follow~\cite{chu2026faster} directly with the only difference being that we replace torsos by a sparse version of torsos, we defer them to the appendix. 

\subsection{Almost-linear time approximate solution.}

We start with the LP for multicut.

\[
\begin{array}{rll}
\min & \displaystyle \sum_{e\in E} x_e \\[1ex]
\mathrm{s.t.}
    & \operatorname{dist}_x(s_i,t_i) \ge 1
    & \forall i\in[q], \\[0.5ex]
    & x_e \ge 0
    & \forall e\in E .
\end{array}
\]
Denote the optimum value of this LP by $\OPT^{\mathrm{MC}}_{\mathrm{LP}}(G,\mathcal T)$. Using \Cref{thm:FastMultiCutLPSolver} we can obtain a $2^{O(1/\epsilon)}$ approximate solution to this LP in time $m^{1 + \mathrm{poly}(\epsilon)}$. Next, we describe how to round this LP.

Let \(x:E\to \mathbb R_{\ge 0}\) be any feasible solution
of the LP, and let
\[
    L=\sum_{e\in E}x_e,
    \qquad
    \Lambda=\max\{2,\lceil L\rceil\}.
\]
For a graph \(H\), a center \(v\), and a radius \(r\), let
\[
    S_H(v,r)=\{u\in V(H): \operatorname{dist}^H_x(v,u)<r\}.
\]
For an edge \(e=ab\in E(H)\), write \(d(a)=\operatorname{dist}^H_x(v,a)\), and define
\[
\theta_e(r)=
\begin{cases}
x_e, & a,b\in S_H(v,r),\\
r-d(a), & a\in S_H(v,r),\ b\notin S_H(v,r),\\
r-d(b), & b\in S_H(v,r),\ a\notin S_H(v,r),\\
0, & a,b\notin S_H(v,r).
\end{cases}
\]
The fractional volume is
\[
    \operatorname{vol}_H(v,r)=\sum_{e\in E(H)}\theta_e(r).
\]
At every regular radius,
\[
    \frac{d}{dr}\operatorname{vol}_H(v,r)
    =
    |\delta_H(S_H(v,r))|.
\]

\paragraph{Rounding algorithm.}
Initialize \(F=\emptyset\) and \(H=G\). While some terminal pair \((s_i,t_i)\) is connected in
\(H\), choose such a pair, set \(v=s_i\), choose a radius \(r\in[1/8,1/4]\) satisfying
\[
    |\delta_H(S_H(v,r))|
    \le
    16\ln\Lambda\left(\operatorname{vol}_H(v,r)+\frac{L}{\Lambda}\right),
\]
add \(\delta_H(S_H(v,r))\) to \(F\), and delete \(S_H(v,r)\) from \(H\). Return \(F\).

\begin{lemma}[Existence of a good radius]\label{lemma:regiongrowing}%
In every iteration of the rounding algorithm, such a radius \(r\in[1/8,1/4]\) exists.
\end{lemma}

\begin{proof}
Let
\[
    y(r)=\operatorname{vol}_H(v,r)+\frac{L}{\Lambda}.
\]
If no radius satisfies the displayed inequality, then at all regular radii in \([1/8,1/4]\),
\[
    y'(r)>16\ln\Lambda\cdot y(r).
\]
Integrating gives
\[
    \ln y(1/4)-\ln y(1/8)
    >
    16\ln\Lambda\cdot (1/8)
    =
    2\ln\Lambda.
\]
Thus
\[
    y(1/4)>\Lambda^2 y(1/8).
\]
Since \(y(1/8)\ge L/\Lambda\), this implies \(y(1/4)>\Lambda L\). On the other hand,
\[
    \operatorname{vol}_H(v,1/4)\le L,
\]
and hence
\[
    y(1/4)\le L+\frac{L}{\Lambda}\le \frac{3L}{2},
\]
contradicting \(\Lambda\ge 2\). Therefore a good radius exists.
\end{proof}

\begin{lemma}[Correctness]%
The rounding algorithm returns an edge multicut.
\end{lemma}

\begin{proof}
For every \(r<1/2\), the ball \(S_H(v,r)\) contains at most one endpoint of every terminal pair.
Indeed, otherwise some terminal pair \((s_j,t_j)\) would have both endpoints in \(S_H(v,r)\), and
therefore
\[
    \operatorname{dist}^G_x(s_j,t_j)
    \le
    \operatorname{dist}^H_x(s_j,t_j)
    <
    2r
    <
    1,
\]
contradicting feasibility of the distance LP.

Hence, when a region \(S_H(v,r)\) is deleted, every terminal pair crossing the boundary of this
region is separated by the added edge cut \(\delta_H(S_H(v,r))\), no terminal pair has both
endpoints inside the deleted region, and all remaining terminal pairs are handled in later
iterations. Therefore the final set \(F\) separates every terminal pair.
\end{proof}

\begin{lemma}[Iteration bound]%
The number of iterations is at most \(8L\).
\end{lemma}

\begin{proof}
Consider an iteration centered at \(v=s_i\). The other endpoint \(t_i\) is connected to \(v\) in
the current graph \(H\), and every \(v\)-\(t_i\) path has \(x\)-length at least \(1\). Therefore the
fractional ball of radius \(r\ge 1/8\) captures fractional mass at least \(1/8\) along a shortest
\(v\)-\(t_i\) path.

The deleted regions are disjoint in the active graph. Moreover, once an edge contributes to the
fractional volume of a deleted region, it is removed from the active graph and never contributes
again. Thus the total captured fractional mass over all iterations is at most
\[
    \sum_{e\in E}x_e=L.
\]
Since each iteration captures mass at least \(1/8\), there are at most \(8L\) iterations.
\end{proof}

\begin{theorem}[Region-growing approximation]%
Given a feasible LP solution of value \(L>0\), the rounding algorithm returns a multicut \(F\)
with
\[
    |F|\le 144L\ln\Lambda.
\]
\end{theorem}

\begin{proof}
Let \(V_j\) be the fractional volume of the region chosen at iteration \(j\), and let
\(\delta_j\) be its boundary in the current graph. By the choice of the radius in \Cref{lemma:regiongrowing},
\[
    |\delta_j|
    \le
    16\ln\Lambda\left(V_j+\frac{L}{\Lambda}\right).
\]
The region volumes sum to at most \(L\). Also, by the iteration bound, the number of iterations
is at most \(8L\le 8\Lambda\). Therefore
\[
\begin{aligned}
    |F|
    &\le
    16\ln\Lambda
    \left(
        \sum_j V_j
        +
        \sum_j \frac{L}{\Lambda}
    \right)  \\
    &\le
    16\ln\Lambda
    \left(
        L+8\Lambda\cdot\frac{L}{\Lambda}
    \right)  \\
    &=
    144L\ln\Lambda.
\end{aligned}
\]
\end{proof}
\subsection{Compression from an Approximate Edge Cut.}
\label{subsec:multicut-compression}

For the compression problem, we will work with Vertex Multicut Compression so that we can follow the approach of~\cite{chu2026faster}. We therefore use a
protected-subdivision gadget to transform edge cuts to vertex cuts. For every original vertex $v\in V(G)$, create
$k+1$ clones
\[
v^1,\ldots,v^{k+1}.
\]
For every original edge $e=uv$, create one edge-vertex $x_e$ adjacent to every
clone of $u$ and every clone of $v$. For every original terminal pair $(s,t)$,
create all $(k+1)^2$ terminal pairs
\[
(s^a,t^b), \qquad a,b\in[k+1].
\]
All non-compression vertices are formally deletable in the vertex instance.
However, deleting original vertices is useless as an edge operation: deleting
all clones of a single original vertex costs $k+1$, which is impossible within
budget $k$.

\begin{lemma}[protected-subdivision equivalence]%
\label{lem:protected-subdivision-equivalence}%
The original Edge Multicut instance has an edge multicut of size at most $k$ if
and only if the protected-subdivision vertex instance has a vertex multicut of
size at most $k$.
\end{lemma}

\begin{proof}
Given an edge multicut $F$, delete the corresponding vertices
$\{x_e:e\in F\}$. Every surviving path between clones of $s$ and $t$ would
project to an $s$-$t$ path in $G-F$.

Conversely, let $X$ be any vertex multicut of size at most $k$ in the protected
instance. For every original vertex $v$, at least one clone $v^a$ survives. Let
\[
F_X=\{e:x_e\in X\}.
\]
If some original terminal pair $s,t$ is connected in $G-F_X$, take an $s$-$t$
path. For every original vertex on the path choose a surviving clone. Every
surviving original edge $e=uv$ has $x_e\notin X$, and $x_e$ is adjacent to every
surviving clone of $u$ and $v$, so the original path lifts to a path between
some surviving $s^a$ and $t^b$ in the protected graph minus $X$. This
contradicts the vertex multicut property. Hence $F_X$ is an edge multicut and
$|F_X|\le |X|\le k$.
\end{proof}

Let $F_0$ be the reference edge multicut returned by the region-growing
rounding, with $|F_0| \leq O(k \log k \cdot 2^{O(1/\epsilon)})$. In the protected-subdivision graph, let
\[
W_0=\{x_e:e\in F_0\}.
\]
Then $W_0$ is a known vertex multicut in the protected-subdivision graph.

Now, the target solution $X$ may reuse some vertices of $W_0$ and may connect other
vertices of $W_0$ to each other. We guess both choices.

\begin{definition}[compression guess]%
\label{def:compression-guess}%
A compression guess consists of a subset $Y\subseteq W_0$ and a partition
$\Pi$ of $W_0\setminus Y$. We delete $Y$, reduce the remaining budget to
$k-|Y|$, and contract every block $B\in\Pi$ to a new terminal vertex $w_B$.
The new terminal set is
\[
W=\{w_B:B\in\Pi\}.
\]
Every terminal pair is updated under these contractions; pairs containing a
vertex of $Y$ are already separated by the forced deletion and are removed.
\end{definition}

\begin{lemma}[compression correctness]%
\label{lem:compression-correctness}%
The original Edge Multicut instance has a solution of size at most $k$ if and
only if at least one compression guess produces a Vertex Multicut Compression
instance with a solution of size at most its remaining budget.
\end{lemma}

\begin{proof}
Suppose $X$ is a solution of size at most $k$ in the protected-subdivision
graph. Let $Y=X\cap W_0$, and partition $W_0\setminus Y$ according to connected
components of the graph after deleting $X$. In the guessed instance,
$X\setminus Y$ is disjoint from the new terminal set, separates every updated
terminal pair, and separates every two distinct terminals $w_B,w_{B'}$:
distinct blocks were chosen from distinct components of the graph after
deleting $X$.

Conversely, let a guessed compression instance have a solution $X'$. Then
$X'\cup Y$ has size at most $k$. Contraction only adds possible connections, so
if all updated terminal pairs are separated after contraction, then the
corresponding clone terminal pairs are separated before contraction. Thus
$X'\cup Y$ is a vertex multicut in the protected-subdivision graph, and
Lemma~\ref{lem:protected-subdivision-equivalence} turns it into an edge
multicut in $G$ of size at most $k$.
\end{proof}

The number of guesses we make is at most $|W_0|^{O(|W_0|)}$. Note that at this point, if our guess is correct, any vertex multicut of the resulting instance is a multiway cut for $W$, and $W$ itself is a vertex multicut. We denote such a compression instance by $(G, \mathcal{T}, W,k)$. Throughout our future transformations, we will maintain that $W$ is a vertex multicut. 

\begin{lemma}[Cleaning equal-endpoint terminal pairs]%
\label{lem:clean-equal-endpoint-pairs}%
Suppose a current compression instance $(G, \mathcal{T}, W, k)$ has a terminal pair $(a,a)$. A solution
separates this pair if and only if it deletes $a$. Therefore the following
normalization is exact.
\begin{itemize}
    \item If $a\in W$, discard the branch as a fixed no-instance.
    \item If $a\notin W$ and the remaining budget is zero, discard the branch
    as a fixed no-instance.
    \item If $a\notin W$ and the remaining budget is positive, delete $a$ from
    the graph, decrease the budget by one, and remove all terminal pairs
    containing $a$.
\end{itemize}
Repeating this rule terminates. Every solution of the original current instance
contains all force-deleted vertices and restricts to a solution of the cleaned
instance; every solution of the cleaned instance together with the
force-deleted vertices solves the original current instance. If a shadowless
solution is preserved into the rule, then the restricted solution remains
shadowless after the forced deletion. If the input to the rule is bipedal, the
cleaned instance is bipedal.
\end{lemma}

\begin{proof}
The length-zero path from $a$ to itself exists in $G-X$ exactly when
$a\notin X$. Hence every feasible solution must delete $a$. If $a\in W$,
deletion is forbidden. If the remaining budget is zero, deletion is impossible.
Otherwise deletion is forced; after deleting $a$, every terminal pair containing
$a$ is already separated and may be removed. The rule strictly reduces
$|V|+|\mathcal{T}|$ or discards the branch, so repetition terminates.

For several forced deletions, apply the one-pair argument in the order in which
the cleaning rule performs them. At step $i$, every feasible solution of the
pre-cleaned instance must contain the vertex deleted at that step, and after
deleting it the remaining constraints are exactly the constraints of the next
instance in the cleaning sequence. This proves both directions by induction
over the sequence. Shadowlessness is preserved because if $a\in X$, then
\[
(G-a)-(X\setminus\{a\})=G-X.
\]
Deleting nonterminal vertices cannot create a component of $G-W$ with more
neighbors in $W$, so bipedality is also preserved.
\end{proof}

\subsection{Reduction to shadowless instances.}\label{sec:shadowremoval}
In this section we reduce a general compression instance to a small family of
instances in which the solution may be assumed to be shadowless.  Recall that a
compression instance is a tuple
\[
    I=(G,\mathcal T,W,k),
\]
where \(W\subseteq V(G)\) is the set of undeletable compression terminals, every
vertex in \(V(G)\setminus W\) is deletable, and a solution is a set
\(X\subseteq V(G)\setminus W\), \(|X|\le k\), that separates all pairs in
\(\mathcal T\) and also separates every two distinct vertices of \(W\).  For such
a set \(X\), its shadow is
\[
    \operatorname{sh}_G(X)
    =
    \{v\in V(G)\setminus X :
      v \text{ is not connected to any vertex of } W \text{ in } G-X\}.
\]
We say that \(X\) is shadowless if \(\operatorname{sh}_G(X)=\emptyset\), that is,
every vertex surviving after the deletion of \(X\) lies in a connected component
that contains a terminal of \(W\).  The goal of this section is to construct, in
time
\[
    2^{O(k\log k)}\widetilde O(|V(G)|+|E(G)|+|\mathcal T|),
\]
a family of at most \(2^{O(k\log k)}\log |V(G)|\) compression instances
\[
    I_1,\ldots,I_r
\]
such that the following two properties hold.  First, if the original instance
\(I\) has a solution of size at most \(k\), then, with high probability, at least
one instance \(I_j\) has a shadowless solution of size at most \(k\).  Second,
the reduction is sound in the reverse direction: if any output instance \(I_j\)
has a solution of size at most \(k\), then the original instance \(I\) also has a
solution of size at most \(k\).  Thus, after this section, it suffices to solve
only instances that admit a shadowless solution in the yes-case.

We now give the reduction promised above. Let us recall the definition of an important separator, but now specialized to the setting where $W$ is undeletable.  

\begin{definition}[Important separator]%
Let \(s\in V(G)\setminus W\).  An \((s,W)\)-separator
\(S\subseteq V(G)\setminus (W\cup\{s\})\) is important if there is no
\((s,W)\)-separator \(S'\) with \(|S'|\le |S|\) and
\[
    R_G(s,S)\subsetneq R_G(s,S'),
\]
where \(R_G(s,S)\) denotes the set of vertices reachable from \(s\) in \(G-S\).
\end{definition}

\begin{definition}[\(k\)-shadow-removable set]%
A set \(Y\subseteq V(G)\setminus W\) is \(k\)-shadow-removable with respect to
\(W\) if \(|Y|\le k\) and the following two conditions hold:
\begin{enumerate}
    \item for every \(v\in \operatorname{sh}_G(Y)\), there is an important
    \((v,W)\)-separator \(S\) of size at most \(k\) with \(S\subseteq Y\);
    \item for every \(y\in Y\), every important \((y,W)\)-separator of size
    at most \(k\) contains at least one vertex outside \(Y\).
\end{enumerate}
\end{definition}

We use the following imported fact from the shadow-removal framework of~\cite{marx2011fixed, chu2026faster}.

\begin{theorem}\cite{marx2011fixed, chu2026faster}
If a multicut compression instance $(G, \mathcal{T}, W, k)$ is a yes-instance, then it has a
solution \(X\) of size at most \(k\) that is \(k\)-shadow-removable with respect
to \(W\).
\end{theorem}

\begin{lemma}[One coloring produces one shadow cover]%
Let \(c:V(G)\setminus W\to\{\mathrm{red},\mathrm{blue}\}\) be a two-coloring.
Define
\[
    Z(c)
    =
    \{v\in V(G)\setminus W :
      \text{some important } (v,W)\text{-separator of size at most }k
      \text{ is entirely red}\}.
\]
Then

\begin{enumerate}

\item For every \(k\)-shadow-removable set \(Y\), there are disjoint sets
\(R,B\subseteq V(G)\setminus W\) with
\[
    |R|\le k,
    \qquad
    |B|\le k^2,
\]
such that if all vertices of \(R\) are red and all vertices of \(B\) are blue,
then
\[
    \operatorname{sh}_G(Y)\subseteq Z(c),
    \qquad
    Z(c)\cap Y=\emptyset .
\]

\item Given $c$, \(Z(c)\) can be computed in
\[
    2^{O(k)}\widetilde O(|V(G)|+|E(G)|)
\]

\end{enumerate}
\end{lemma}

\begin{proof}
Set \(R=Y\).  Since \(|Y|\le k\), this satisfies \(|R|\le k\).  If
\(v\in \operatorname{sh}_G(Y)\), then by \(k\)-shadow-removability there is an
important \((v,W)\)-separator \(S\subseteq Y\) of size at most \(k\).  When
all vertices of \(Y\) are red, the separator \(S\) is entirely red, and hence
\(v\in Z(c)\).  Thus
\[
    \operatorname{sh}_G(Y)\subseteq Z(c).
\]

It remains to ensure that no vertex of \(Y\) enters \(Z(c)\).  Fix \(y\in Y\).
By the second condition in the definition of \(k\)-shadow-removability, every
important \((y,W)\)-separator of size at most \(k\) contains a vertex outside
\(Y\).  Applying the hitting lemma for important separators (Lemma 8 of ~\cite{chu2026faster}), there is
a set
\[
    H_y\subseteq V(G)\setminus Y,
    \qquad
    |H_y|\le k,
\]
that intersects every important \((y,W)\)-separator of size at most \(k\).
Let
\[
    B=\bigcup_{y\in Y} H_y .
\]
Then \(B\cap R=\emptyset\) and \(|B|\le k^2\).  If all vertices of \(B\) are
blue, then every important \((y,W)\)-separator of size at most \(k\) contains
a blue vertex, and therefore no such separator is entirely red.  Hence
\(y\notin Z(c)\) for every \(y\in Y\), proving \(Z(c)\cap Y=\emptyset\).

\end{proof}

\begin{lemma}[Finding shadows]\label{lemma:shadowremoval}%
One can compute a family
\[
    \mathcal Z=\{Z_1,\ldots,Z_r\},
    \qquad
    r\le 2^{O(k \log k)}\log |V(G)|,
\]
in time
\[
    2^{O(k \log k)}\widetilde O(|V(G)|+|E(G)|+|\mathcal T|)
\]
such that, for every \(k\)-shadow-removable set \(Y\), some \(Z_i\in\mathcal Z\)
satisfies
\[
    \operatorname{sh}_G(Y)\subseteq Z_i,
    \qquad
    Z_i\cap Y=\emptyset .
\]
\end{lemma}

\begin{proof}

Fix a \(k\)-shadow-removable set \(Y\).  By the previous lemma, there are
disjoint sets \(R,B\subseteq V(G)\setminus W\) with \(|R|\le k\), \(|B|\le k^2\),
such that coloring \(R\) red and \(B\) blue gives the desired shadow cover. Using, for example, Lemma 2.1 from \cite{cygan2020randomized}, we compute a family of sets $\mathcal{C}$ of size $2^{O(k \log k)}\log n$ in time $\Otil(n2^{O(k \log k)})$, so that there is some $C \in \mathcal{C}$, so that  
\[
    R\subseteq C,
    \qquad
    C\cap B=\emptyset .
\]
For each $C \in \mathcal{C}$, we now color the vertices red and blue according to $C$ (vertices in $C$ are red and the rest are blue). If the choice of $C$ is correct, then $R$ is colored red and $B$ is colored blue.

Using our result on single-sink important separators \Cref{thm:singlesink}, it follows that for a given coloring $c$, $Z(c)$ can be computed in time \(
    2^{O(k)}\widetilde O(|V(G)|+|E(G)|)
\). We contract $W$ into a single vertex $\tau$, compute $(v, \tau)$ important separators for each $v$, and include $v \in Z(c)$ if and only if some important $(v,\tau)$-separator is all red.\footnote{Vertices already disconnected from $W$ belong to $Z(c)$ automatically. We run the single-sink algorithm only on the component containing $\tau$; if $W=\emptyset$, set $Z(c)=V(G)$.}
\end{proof}

\subsubsection{Sparse torso for a guessed shadow.}

For a finite set \(A\), define \(\operatorname{Conn}_{k+1}(A)\) as follows.  If
\(|A|\le k+1\), then \(\operatorname{Conn}_{k+1}(A)\) is the clique on \(A\).
Otherwise, order
\[
    A=\{a_0,a_1,\ldots,a_{r-1}\}
\]
cyclically, and join two vertices whenever their cyclic distance is at most $k + 1$. Clearly, the graph $\operatorname{Conn}_{k+1}(A)$ is $(k + 1)$-connected whenever $|A| > k + 1$.  

Let \(Z\subseteq V(G)\setminus W\).  Let
\[
    C_1,\ldots,C_p
\]
be the connected components of \(G[Z]\), and define
\[
    A_j=N_G(C_j) .
\]
The sparse torso \(G^Z\) is obtained from \(G-Z\) by adding
\(\operatorname{Conn}_{k+1}(A_j)\) for every \(j\in[p]\).

We also update terminal pairs.  For each component \(C_j\), choose a representative
set
\[
    R_j\subseteq A_j,
    \qquad
    |R_j|=\min\{k+1,|A_j|\}.
\]
If an endpoint of a terminal pair lies in \(C_j\), we replace that endpoint by all
vertices of \(R_j\).  Thus, if one endpoint lies in \(C_j\), the pair is replaced
by \(|R_j|\) pairs; if the two endpoints lie in distinct components \(C_j,C_{j'}\),
the pair is replaced by all pairs in \(R_j\times R_{j'}\).  If both endpoints of
an original terminal pair lie in the same component \(C_j\), we discard this
guess \(Z\) as a fixed no-instance.  Duplicate pairs are removed, while
equal-endpoint pairs are retained and handled by the cleaning rule \Cref{lem:clean-equal-endpoint-pairs}.

\begin{lemma}[Sparse shadow-torso equivalence]%
Let \(Y\subseteq V(G)\setminus (W\cup Z)\) with \(|Y|\le k\).  Suppose the torso
\(G^Z\) is retained.  Then:
\begin{enumerate}
    \item any connection between vertices of \(V(G)\setminus Z\) in \(G-Y\)
    gives a connection in \(G^Z-Y\);
    \item any connection in \(G^Z-Y\) expands to a connection in \(G-Y\);
    \item if \(Y\) solves the torso instance, then \(Y\) solves the original
    instance.
\end{enumerate}
Furthermore, if \(X\) is a solution of the original instance and
\[
    \operatorname{sh}_G(X)\subseteq Z,
    \qquad
    Z\cap X=\emptyset,
\]
then \(X\) is a shadowless solution of the torso instance.
\end{lemma}

\begin{proof}
For the first claim, consider a path in \(G-Y\) between two vertices of
\(V(G)\setminus Z\).  Whenever the path has a maximal subpath whose internal
vertices lie in some component \(C_j\subseteq Z\), its entry and exit vertices
belong to \(A_j\setminus Y\).  By the connector lemma, these two surviving
boundary vertices are connected in \(\operatorname{Conn}_{k+1}(A_j)-Y\).  Replacing
every such subpath gives a connection in \(G^Z-Y\).

For the second claim, expand every connector edge on \(A_j\) to a path through
the connected component \(C_j\).  Since \(C_j\subseteq Z\) and
\(Y\cap Z=\emptyset\), these expanded paths avoid \(Y\).  Thus every connection
in \(G^Z-Y\) expands to a connection in \(G-Y\).

We now prove soundness for terminal pairs.  Suppose some original terminal pair
\((s,t)\in\mathcal T\) is connected in \(G-Y\).  If neither endpoint lies in
\(Z\), the first claim gives a connected updated pair in \(G^Z-Y\).  Suppose
\(s\in C_j\subseteq Z\) and \(t\notin Z\).  The path from \(s\) to \(t\) exits
\(C_j\) through some vertex \(a\in A_j\setminus Y\).  If \(|A_j|\le k+1\), then
\(R_j=A_j\), so \(a\) itself is a representative.  If \(|A_j|>k+1\), then
\(|R_j|=k+1\), and since \(|Y|\le k\), some \(r\in R_j\setminus Y\) survives.
The connector lemma gives a path from \(r\) to \(a\) in
\(\operatorname{Conn}_{k+1}(A_j)-Y\), so the updated pair \((r,t)\) is connected
in \(G^Z-Y\).  The case in which both endpoints lie in distinct components of
\(G[Z]\) is identical at both ends.  The case in which both endpoints lie in the
same component of \(G[Z]\) was discarded.  Therefore, if \(Y\) separates all
updated terminal pairs in the torso, then \(Y\) separates all original terminal
pairs in \(G\).  The same expansion argument shows that separating distinct
vertices of \(W\) in the torso implies separating them in \(G\).  Hence \(Y\)
solves the original instance.

Finally, assume \(X\) is a solution of the original instance with
\(\operatorname{sh}_G(X)\subseteq Z\) and \(Z\cap X=\emptyset\).  The compatible
torso is not discarded: if both endpoints of a terminal pair lay in the same
component of \(G[Z]\), then they would be connected in \(G-X\), contradicting
that \(X\) solves the original instance.  If an updated terminal pair were
connected in \(G^Z-X\), then expanding connector edges through the corresponding
components of \(Z\), and attaching the endpoint-to-representative paths inside
\(Z\), would give a connection between the original terminal-pair endpoints in
\(G-X\), again a contradiction.  Distinct vertices of \(W\) are handled by the
same expansion argument.  Thus \(X\) is feasible for the torso instance.

It remains to prove shadowlessness.  Let \(u\in V(G^Z)\setminus X\).  Since
\(u\notin Z\), and \(\operatorname{sh}_G(X)\subseteq Z\), the vertex \(u\) has a
path to some vertex of \(W\) in \(G-X\).  Replacing every maximal subpath through
a component of \(G[Z]\) by the corresponding connector path gives a path from
\(u\) to \(W\) in \(G^Z-X\).  Hence every vertex of \(G^Z-X\) reaches \(W\), and
\(X\) is shadowless in the torso.
\end{proof}

\begin{algorithm}[t]
\caption{\(\textsc{ShadowRemoval}(G,\mathcal T,W,k)\)}
\label{alg:shadow-removal}
\begin{algorithmic}[1]
\State Apply the shadow covering lemma from \Cref{lemma:shadowremoval} to obtain a family of sets $\mathcal{Z}$.
\For{each \(Z\in\mathcal Z\)}
    \State Construct the sparse torso \(G^Z\).
    \State Update terminal pairs using the representative sets \(R_j\).
    \If{some terminal pair has both endpoints in the same component of \(G[Z]\)}
        \State discard this \(Z\)-branch.
    \Else
        \State remove duplicate terminal pairs.
        \State clean equal-endpoint terminal pairs by the cleaning rule \Cref{lem:clean-equal-endpoint-pairs}.
        \State output the resulting compression instance.
    \EndIf
\EndFor
\end{algorithmic}
\end{algorithm}

\begin{theorem}[Near-linear sparse shadow removal]%
Given a compression instance \(I=(G,\mathcal T,W,k)\) in which every vertex
outside \(W\) is deletable, Algorithm~\ref{alg:shadow-removal} outputs, with the
success probability of the single-sink important-separator oracle~\Cref{thm:singlesink}, at most
\[
    2^{O(k \log k)}\log |V(G)|
\]
compression instances in time
\[
    2^{O(k \log k)}\widetilde O(|V(G)|+|E(G)|+|\mathcal T|).
\]
The output family satisfies the following:
\begin{enumerate}
    \item if \(I\) has a solution of size at most \(k\), then at least one output
    instance has a shadowless solution of size at most \(k\);
    \item if any output instance has a solution of size at most \(k\), then \(I\)
    has a solution of size at most \(k\).
\end{enumerate}
\end{theorem}

\begin{proof}
Assume first that \(I\) is a yes-instance.  By the shadow-removal theorem of~\cite{chu2026faster}, \(I\) has a solution \(X\) of size at most \(k\) that is
\(k\)-shadow-removable with respect to \(W\).  By the shadow-cover
lemma (\Cref{lemma:shadowremoval}), there is a set \(Z\in\mathcal Z\) such that
\[
    \operatorname{sh}_G(X)\subseteq Z,
    \qquad
    Z\cap X=\emptyset .
\]
Build the sparse torso for this \(Z\).  By the sparse shadow-torso equivalence
lemma, before cleaning equal-endpoint pairs, the same set \(X\) is a shadowless
solution of the torso instance.  The cleaning rule is exact and preserves
shadowlessness, so the cleaned output instance has a shadowless solution of size
at most \(k\).

Conversely, suppose some output instance has a solution.  Undo the cleaning by
adding back the force-deleted vertices; by exactness of the cleaning rule, this
gives a solution of the corresponding uncleaned torso instance with total size at
most \(k\).  The sparse shadow-torso equivalence lemma then maps this solution
back to a solution of the original instance \(I\).

It remains to bound the number of outputs and the running time.  The number of
sets \(Z\in\mathcal Z\) is at most \(2^{O(k \log k)}\log |V(G)|\).  For a fixed
\(Z\), the components \(C_j\) of \(G[Z]\) and their boundaries \(A_j\) are found
by one graph scan.  The connectors add
\[
    \sum_j O(k|A_j|)
    =
    O(k|E(G)|)
\]
edges, because each boundary incidence can be charged to an original edge
incident with \(Z\).  Each terminal pair creates at most \((k+1)^2\) updated
pairs.  Thus each torso has size and construction time bounded by an
FPT-in-\(k\) factor times \(|V(G)|+|E(G)|+|\mathcal T|\).  Multiplying by the
size of \(\mathcal Z\), and including the time to compute \(\mathcal Z\), gives
\[
    2^{O(k \log k)}\widetilde O(|V(G)|+|E(G)|+|\mathcal T|).
\]
\end{proof}

\subsection{Putting the Pieces Together.}
\label{subsec:multicut-putting-together}

First reject if some terminal pair has equal endpoints. If $k=0$, the algorithm
directly checks whether every terminal pair is already disconnected and accepts
exactly in that case. Assume below that $k\ge 1$.

Solve the standard LP oracle with the tradeoff parameter $\epsilon$.
Reject if the solver exceeds the iteration bound guaranteed when
$\OPT^{\mathrm{MC}}_{\mathrm{LP}}\le k$. Let
\[
L=\sum_{e\in E}x_e
\]
be the returned LP value which is at most $\alpha(\epsilon) = 2^{O(1/\epsilon)}$ (the factor of \Cref{thm:FastMultiCutLPSolver}) times the optimal LP value. If
\[
L>\alpha(\epsilon)k,
\]
reject. This rejection is safe: if $\operatorname{OPT}(G,\mathcal{T})\le k$,
then $\OPT^{\mathrm{MC}}_{\mathrm{LP}}(G,\mathcal{T})\le k$, and the LP oracle would give
$L\le \alpha(\epsilon)k$.

If all terminal pairs are already disconnected, accept. Otherwise $L>0$, and
the region-growing theorem returns a reference edge multicut $F_0$ with
\[
|F_0|
\le 144L\ln\max\{2,\lceil L\rceil\}
\le 144\alpha(\epsilon)k
    \ln\max\{2,\lceil \alpha(\epsilon)k\rceil\}.
\]
Let
\[
K_0:=144\alpha(\epsilon)k
    \ln\max\{2,\lceil \alpha(\epsilon)k\rceil\}.
\]
Build the protected-subdivision instance and let
\[
W_0=\{x_e:e\in F_0\}.
\]
Enumerate all compression guesses $(Y,\Pi)$ from
Section~\ref{subsec:multicut-compression}. For every guessed instance, apply
sparse shadow removal from \Cref{sec:shadowremoval}, sparse bipedalization from \Cref{sec:bipedalization}, and the bipedal solver from
Theorem~\ref{thm:near-linear-bipedal-solver} of \Cref{sec:almost2sat}. Accept if and only if some final
bipedal solver accepts.

Soundness is as follows. If the algorithm accepts, then some bipedal solver
accepted. Theorem~\ref{thm:near-linear-bipedal-solver} gives an actual solution
for that bipedal compression instance. Sparse bipedalization lifts this solution
to the shadow-removal output instance. Sparse shadow removal lifts it to the
guessed compression instance. Lemma~\ref{lem:compression-correctness} turns the
guessed compression solution into a vertex multicut of the protected-subdivision
instance of size at most $k$. Lemma~\ref{lem:protected-subdivision-equivalence}
maps that vertex multicut to an edge multicut of the original graph of size at
most $k$. Therefore the original Edge Multicut instance is a yes-instance.

Completeness is the reverse chain. Suppose the original graph has an edge
multicut $X_E$ of size at most $k$. The LP rejection above does not occur, and
the region-growing step produces $F_0$ with $|W_0|=|F_0|\le K_0$. By
Lemma~\ref{lem:protected-subdivision-equivalence}, $X_E$ corresponds to a
vertex multicut $X$ in the protected-subdivision graph. Among the enumerated
compression guesses, choose $Y=X\cap W_0$ and let $\Pi$ be the partition of
$W_0\setminus Y$ induced by connected components of the protected graph minus
$X$. Lemma~\ref{lem:compression-correctness} states that $X\setminus Y$ is a
solution of that guessed compression instance. Sparse shadow removal gives an
output with a shadowless solution. Sparse bipedalization gives a bipedal output
with a shadowless solution. Theorem~\ref{thm:near-linear-bipedal-solver}
accepts every bipedal instance that has a shadowless solution. Hence at least
one branch accepts.

Recall that $N = |V(G)| +| E(G)| + |\mathcal{T}|$ is the input-size. Once we fix $K_0$, the parameterized part of the running time of the algorithm is bottlenecked by the guessing we perform for the compression instance, which takes time $|K_0|^{O(|K_0|)}$. There are at most $k^{O(k)}$ branches of the branching tree in the shadow-removal and bipedalization steps, and for each branch, the running time is at most $\Otil(k^{O(k)}|W_0|N)$. The first step where we solve the LP (\Cref{thm:FastMultiCutLPSolver}) takes time $kN^{1 + \poly(\epsilon)}$.

Thus the total running time is bounded by
\[
K_0^{O(K_0)}
\cdot
\widetilde{O}\!\left(k^{O(k)}N\right)
+
\widetilde{O}\!\left(kN^{1+\poly(\epsilon)}\right).
\]
Choose $\frac{1}{\epsilon} = \omega(1)$ to be a small growing function of $n$. Then,
\[
K_0^{O(K_0)}=f(k)N^{o(1)}.
\]
Consequently, the total running time is
\[
f(k)N^{1+o(1)}.
\]

Note that one can choose any growing function of $n$, as small as possible. For a concrete choice, observe that for $\epsilon = \Theta(\frac{1}{\log \log \log n})$ chosen so that $\alpha(\epsilon)\le \log\log n$, $K_0 = O(k\log\log N\cdot\log(k\log\log N))$, so $|K_0|^{O(|K_0|)}$ is bounded by $g(k)N^{o(1)}$ via a standard trick: either $\log \log N \leq k$, in which case we bound by $(k^2 \log k)^{O(k^2 \log k)}$, or $k \leq \log \log N$ in which case the whole expression is $N^{o(1)}$.

%% file: appendix_fast_multicut.tex
\section{An Efficient Primal Solver for the Edge Multicut LP.}
\label{sect:FastMulticut}

Let \(G=(V,E)\) be an undirected graph, let
\[
        \mathcal T=\{(s_i,t_i): i\in[q]\}
\]
be a set of terminal pairs, and let \(c\in \mathbb Z_{\ge 1}^E\) be
edge costs.  We write
\[
        N:=\max\{|E|+q,2\}.
\]
The standard fractional edge-multicut LP is
\[
\OPT^{\mathrm{MC}}_{\mathrm{LP}}
=
\min_{x\in\R_{\ge 0}^E}
\left\{
    \sum_{e\in E} c_e x_e
    :
    \dist_x(s_i,t_i)\ge 1
    \text{ for every } i\in[q]
\right\}.
\]

We use the convention that \(\dist_x(s_i,t_i)=+\infty\) if
\(s_i\) and \(t_i\) are disconnected.

\begin{theorem}[Cost-budgeted non-concurrent flow oracle, \cite{haeupler2024low}]%
\label{thm:LowStepFlowEmulator}%
Fix any constant \(\eps\in(0,1)\).  There is a constant
\[
        A_\eps=2^{1/\eps}
\]
and an algorithm with the following guarantee.

Given an undirected graph with capacities \(u_e\ge 1\), nonnegative edge
costs \(b_e\ge 0\), terminal pairs \(\mathcal T\), and a budget \(B\ge 0\),
the algorithm returns a capacity-respecting non-concurrent flow \(H\)
routing only between terminal pairs in \(\mathcal T\), satisfying
\[
        \sum_{e\in E} b_e H(e)\le B,
\]
and
\[
        \val(H)
        \ge
        \frac{1}{A_\eps}\,\OPT_{\mathrm{flow}}(B),
\]
where \(\OPT_{\mathrm{flow}}(B)\) is the maximum value of any
capacity-respecting non-concurrent flow of \(b\)-cost at most \(B\).

One call takes
\[
        \widetilde O(N^{1+\operatorname{poly}(\eps)})
\]
time. The oracle returns the flow $H$ in the sense that it gives the aggregate edge
loads \(H(e)\).
\end{theorem}

\begin{theorem}[Efficient multicut-LP solver]%
\label{thm:SectAmain}%
Fix any constant \(\eps\in(0,1)\).  Under the black-box theorem above,
there is an algorithm which outputs an explicit vector
\(x\in\R_{\ge 0}^E\) satisfying
\[
        \dist_x(s_i,t_i)\ge 1
        \qquad\forall i\in[q],
\]
and
\[
        \sum_{e\in E} c_e x_e
        \le
        C_\eps\,\OPT^{\mathrm{MC}}_{\mathrm{LP}},
\]
where
\[
        C_\eps=2^{O(1/\eps)}.
\]
The running time is
\[
        \widetilde O\!\left(
            \OPT^{\mathrm{MC}}_{\mathrm{LP}}\,N^{1+\operatorname{poly}(\eps)}
        \right),
\]
assuming polynomially bounded numerical precision.
\end{theorem}

\begin{proof}
We first remove terminal pairs \((s_i,t_i)\) for which \(s_i\) and
\(t_i\) are disconnected, since they impose no constraint in the
multicut LP and cannot carry flow in the dual.  If no terminal pair
remains, then \(x=0\) is optimal.  Hence we may assume at least one
terminal pair is connected.  Since \(c_e\ge 1\), one unit of flow can be
sent along any terminal path, so
\[
        \OPT^{\mathrm{MC}}_{\mathrm{LP}}\ge 1.
\]
Set
\[
        \beta:=\OPT^{\mathrm{MC}}_{\mathrm{LP}}.
\]

The LP dual of the multicut LP is the maximum non-concurrent
multicommodity flow LP with edge capacities \(c_e\).  Therefore, by
strong LP duality, \(\beta\) is also the maximum value of a
capacity-respecting non-concurrent flow between the terminal pairs.

We now prove two lemmas and then combine them.

\begin{lemma}[Approximate shortest unit-flow oracle]%
\label{lemma:SectAFirstLemma}%
For every length vector \(\ell\in\R_{\ge 0}^E\), define
\[
        \alpha(\ell)
        :=
        \min_i \dist_\ell(s_i,t_i).
\]
There is an oracle which returns a value-one terminal-pair flow \(F\)
such that
\[
        \ell\cdot F
        \le
        \rho\,\alpha(\ell),
\]
where
\[
        \rho:=2A_\eps.
\]
Moreover,
\[
        F(e)\le 1\le c_e
        \qquad\forall e\in E.
\]
One call to this oracle takes \(\widetilde O(N^{1+\operatorname{poly}(\eps)})\) time.
\end{lemma}

\begin{proof}
Let \(\alpha=\alpha(\ell)\).  A shortest terminal path of
\(\ell\)-length \(\alpha\) supports one unit of capacity-respecting flow,
because \(c_e\ge 1\) for every edge.  Hence for every budget
\(B\ge \alpha\),
\[
        \OPT_{\mathrm{flow}}(B)\ge 1.
\]

Run the budgeted flow oracle for geometrically increasing budgets \(B\),
and take the first budget for which the returned flow has value at least
\(1/A_\eps\).  By the previous paragraph, such a budget is found with
\[
        B\le 2\alpha.
\]
Let \(H\) be the flow returned at this budget and let
\[
        v:=\val(H).
\]
Then \(v\ge 1/A_\eps\) and
\[
        \ell\cdot H\le B.
\]
Scale \(H\) by \(1/v\), and call the resulting value-one flow \(F\).
Then
\[
        \ell\cdot F
        =
        \frac{\ell\cdot H}{v}
        \le
        A_\eps B
        \le
        2A_\eps\,\alpha
        =
        \rho\,\alpha.
\]

After deleting cycles, \(F\) is a convex combination of simple terminal
paths with total weight \(1\).  Therefore no edge is used by more than
one unit of total path weight, so
\[
        F(e)\le 1\le c_e
        \qquad\forall e\in E.
\]
The geometric search contributes only logarithmic overhead, hidden in
\(\widetilde O(\cdot)\).
\end{proof}

\begin{lemma}[Garg--Könemann primal certificate, \cite{DBLP:journals/siamcomp/GargK07}]%
\label{lemma:SectASecondLemma}%
Suppose we have an oracle which, for every \(\ell\ge 0\), returns a
value-one terminal flow \(F\) satisfying
\[
        \ell\cdot F\le \rho\,\alpha(\ell),
        \qquad
        F(e)\le c_e
        \quad\forall e.
\]
Then one can output a feasible multicut-LP vector \(x\) of cost
\[
        \sum_e c_e x_e
        \le
        O(\rho)\,\OPT^{\mathrm{MC}}_{\mathrm{LP}}
\]
using
\[
        O(\OPT^{\mathrm{MC}}_{\mathrm{LP}}\log N)
\]
oracle calls.
\end{lemma}

\begin{proof}
Set
\[
        \Delta:=N^{-4}
\]
and initialize
\[
        \ell^{(0)}_e:=\frac{\Delta}{c_e}.
\]
Let
\[
        D(\ell):=\sum_e c_e\ell_e.
\]
Thus
\[
        D(\ell^{(0)})\le N^{-3}.
\]

At iteration \(r\), call the oracle on \(\ell^{(r-1)}\), obtaining
\(F^{(r)}\).  Define
\[
        C_r:=\ell^{(r-1)}\cdot F^{(r)},
        \qquad
        D_{r-1}:=D(\ell^{(r-1)}),
\]
and record the ratio
\[
        R_r:=\frac{D_{r-1}}{C_r}.
\]
Then update
\[
        \ell^{(r)}_e
        :=
        \ell^{(r-1)}_e
        \left(
            1+\frac{F^{(r)}(e)}{2c_e}
        \right).
\]
Stop at the first iteration \(t\) such that
\[
        D(\ell^{(t)})\ge 1.
\]

Let
\[
        r^*\in\arg\min_{1\le r\le t} R_r.
\]
We output
\[
        x_e
        :=
        \frac{\rho\,\ell^{(r^*-1)}_e}{C_{r^*}}.
\]

We first prove feasibility.  The oracle guarantee gives
\[
        C_{r^*}
        \le
        \rho\,\alpha(\ell^{(r^*-1)}),
\]
so
\[
        \alpha(\ell^{(r^*-1)})
        \ge
        \frac{C_{r^*}}{\rho}.
\]
Therefore, for every terminal pair,
\[
        \dist_x(s_i,t_i)
        =
        \frac{\rho}{C_{r^*}}
        \dist_{\ell^{(r^*-1)}}(s_i,t_i)
        \ge 1.
\]

It remains to bound the cost.  The cost of \(x\) is
\[
        \sum_e c_e x_e
        =
        \rho\,\frac{D_{r^*-1}}{C_{r^*}}
        =
        \rho R_{r^*}.
\]
Thus it is enough to show that
\[
        R_{r^*}=O(\beta).
\]

Let
\[
        L:=\log_{3/2}(1/\Delta)=O(\log N).
\]
Consider any prefix of iterations before the stopping time.  Since
\(D(\ell)<1\) throughout such a prefix, each edge satisfies
\[
        \ell_e < \frac{1}{c_e}.
\]
Since \(\ell_e^{(0)}=\Delta/c_e\), the total multiplicative increase of
any edge before stopping is less than \(1/\Delta\).

Also, because \(F^{(r)}(e)\le c_e\),
\[
        1+\frac{F^{(r)}(e)}{2c_e}
        \ge
        \left(\frac32\right)^{F^{(r)}(e)/c_e}.
\]
Hence, before stopping,
\[
        \sum_r F^{(r)}(e)
        \le
        Lc_e
        \qquad\forall e.
\]
Therefore the average of the first \(T\) oracle flows, scaled by \(1/L\),
is capacity-respecting and has value \(T/L\).  Since \(\beta\) is the
maximum feasible flow value,
\[
        \frac{T}{L}\le \beta.
\]
So the loop stops after
\[
        t=O(\beta L)
        =
        O(\beta\log N)
\]
iterations.

Finally, because
\[
        D_r
        =
        D_{r-1}
        +
        \frac12 C_r,
\]
and \(R_{r^*}\le R_r=D_{r-1}/C_r\) for every \(r\), we have
\[
        D_r
        \le
        D_{r-1}
        \left(
            1+\frac{1}{2R_{r^*}}
        \right).
\]
Multiplying over all \(t\) iterations,
\[
        1
        \le
        D(\ell^{(t)})
        \le
        D(\ell^{(0)})
        \exp\!\left(\frac{t}{2R_{r^*}}\right).
\]
Since \(D(\ell^{(0)})\le N^{-3}\), this implies
\[
        R_{r^*}
        \le
        O\!\left(\frac{t}{\log N}\right)
        =
        O(\beta).
\]
Therefore
\[
        \sum_e c_e x_e
        =
        \rho R_{r^*}
        \le
        O(\rho)\beta
        =
        O(\rho)\OPT^{\mathrm{MC}}_{\mathrm{LP}}.
\]
\end{proof}

We now finish the proof of \Cref{thm:SectAmain}.  By \Cref{lemma:SectAFirstLemma}, the budgeted flow oracle \Cref{thm:LowStepFlowEmulator} gives an approximate shortest
unit-flow oracle with
\[
        \rho=2A_\eps
        =
        2^{1+1/\eps}.
\]
By \Cref{lemma:SectASecondLemma}, this oracle gives an explicit feasible multicut-LP
vector \(x\) of cost
\[
        O(\rho)\OPT^{\mathrm{MC}}_{\mathrm{LP}}
        \le
        C_\eps
        \OPT^{\mathrm{MC}}_{\mathrm{LP}}.
\]

The number of outer iterations is
\[
        O(\OPT^{\mathrm{MC}}_{\mathrm{LP}}\log N).
\]
Each iteration makes one approximate shortest unit-flow oracle call,
which costs
\[
        \widetilde O(N^{1+\operatorname{poly}(\eps)}).
\]
Therefore the total running time is
\[
        \widetilde O\!\left(
            \OPT^{\mathrm{MC}}_{\mathrm{LP}}\,N^{1+\operatorname{poly}(\eps)}
        \right).
\]
The output vector is explicit and has one coordinate per edge.
\end{proof}
\section{Sparse Bipedalization.}\label{sec:bipedalization}

This section gives the sparse torso version of the bipedalization step~\cite{chu2026faster}. We keep the invariant
that every vertex outside \(W\) is deletable. A compression instance $(G, \mathcal{T}, W, k)$ is bipedal if every connected
component \(C\) of \(G-W\) has
\[
    |N_G(C)\cap W|\le 2.
\]
Starting from a shadowless compression instance, we will output a bipedal compression instance.

We use the Node Multiway Cut LP with terminal set \(W\):
\[
\begin{array}{rll}
\min & \displaystyle \sum_{v\in V(G)\setminus W} d_v \\[1ex]
\mathrm{s.t.}
    & \displaystyle \sum_{v\in V(P)\setminus W} d_v \ge 1
    & \text{for every path \(P\) connecting two distinct terminals of \(W\),} \\[1ex]
    & d_v\ge 0
    & \forall v\in V(G)\setminus W .
\end{array}
\]
Its optimum is denoted \(\operatorname{OPT}^{\mathrm{MW}}_{\mathrm{LP}}(G,W)\).

\begin{theorem}[Half-integral Node Multiway Cut LP
{\normalfont\cite{iwata2016half,iwata20180}}]
\label{thm:node-multiway-cut-lp}
The Node Multiway Cut LP admits a half-integral optimal solution.
Moreover, given a threshold $k$, in time
\[
O\!\left(\poly(k)(|V(G)|+|E(G)|)\right),
\]
one can either compute a half-integral optimal solution when
$\operatorname{OPT}^{\mathrm{MW}}_{\mathrm{LP}}(G,W)\leq k$, or correctly certify that
\[
\operatorname{OPT}^{\mathrm{MW}}_{\mathrm{LP}}(G,W)>k.
\]
\end{theorem}

For a finite set \(A\), set \(q=4k+1\) and define \(\operatorname{Conn}_q(A)\) as in
the shadow removal section, with \(q\) in place of \(k+1\).

\(\operatorname{Conn}_q(A)\) has \(O(q|A|)\) edges and can be built in
\(O(q|A|)\) time. If \(|A|>q\), then \(\operatorname{Conn}_q(A)\) is
\(q\)-vertex-connected. In particular, after deleting any set of at most \(4k\)
vertices, all surviving vertices of \(A\) remain connected in
\(\operatorname{Conn}_{4k+1}(A)\).

For \(v\in V(G)\setminus W\), the sparse one-vertex torso is
\[
    \operatorname{storso}_k(G,v)
    =
    G-v+\operatorname{Conn}_{4k+1}(N_G(v)).
\]
Terminal pairs involving \(v\) are updated as follows. Let \(A=N_G(v)\) and choose
\[
    R_v\subseteq A,
    \qquad
    |R_v|=\min\{k+1,|A|\}.
\]
A pair \((v,t)\) is replaced by all \((r,t)\) with \(r\in R_v\), and a pair
\((s,v)\) is replaced by all \((s,r)\) with \(r\in R_v\). Pairs not involving
\(v\) are kept. Equal-endpoint pairs produced by this operation are handled by
Lemma~\ref{lem:clean-equal-endpoint-pairs}.

\begin{lemma}[LP preservation below the cutoff]%
\label{lem:sparse-torso-lp-preservation}%
Let \(F\) be the full one-vertex torso of \(v\), obtained by making \(N_G(v)\) a
clique, and let \(S=\operatorname{storso}_k(G,v)\). For the Node Multiway Cut LP:
\begin{enumerate}
    \item \(\operatorname{OPT}^{\mathrm{MW}}_{\mathrm{LP}}(S,W)
        \le \operatorname{OPT}^{\mathrm{MW}}_{\mathrm{LP}}(F,W)\);
    \item if \(\operatorname{OPT}^{\mathrm{MW}}_{\mathrm{LP}}(S,W)\le 2k\), then
    \[
        \operatorname{OPT}^{\mathrm{MW}}_{\mathrm{LP}}(S,W)
        =
        \operatorname{OPT}^{\mathrm{MW}}_{\mathrm{LP}}(F,W).
    \]
\end{enumerate}
\end{lemma}

\begin{proof}
Since \(S\) is a subgraph of \(F\) on the same vertex set, every \(W\)-\(W\) path
of \(S\) is also a \(W\)-\(W\) path of \(F\). Thus the LP for \(S\) has no more
constraints than the LP for \(F\), proving (1).

For (2), take a half-integral optimum solution \(d\) for \(S\) of value at most
\(2k\) (\Cref{thm:node-multiway-cut-lp}). Let
\[
    Z_d=\{u\in V(S)\setminus W : d_u>0\}.
\]
Every positive entry is at least \(1/2\), so \(|Z_d|\le 4k\).

We prove that \(d\) is feasible for \(F\). Let \(P\) be any path in \(F\) connecting
two distinct terminals of \(W\). If an edge of \(P\) is not a clique edge added by
the full torso, it is already present in \(S\). Consider a full-torso clique edge
\(ab\) with \(a,b\in N_G(v)\). If \(ab\) is an edge of the sparse connector, keep
it. Otherwise \(|N_G(v)|>4k+1\). The connector is \((4k+1)\)-vertex-connected, so
by Menger's theorem it contains \(4k+1\) internally vertex-disjoint \(a\)-\(b\)
paths. Since \(|Z_d|\le 4k\), at least one of these paths has no internal vertex
in \(Z_d\). Replace \(ab\) by such a path. The replacement adds no LP mass,
because every newly inserted nonterminal vertex has \(d\)-value zero, and terminal
vertices have no LP variable.

After all replacements, \(P\) becomes a \(W\)-\(W\) walk in \(S\) of LP mass no
larger than that of \(P\). If \(P\) had LP mass below \(1\), then a subwalk between
two consecutive distinct terminals, after deleting cycles, would contain a simple
\(W\)-\(W\) path in \(S\) of LP mass below \(1\), contradicting feasibility of
\(d\) for \(S\). Hence every \(W\)-\(W\) path in \(F\) has LP mass at least \(1\),
so \(d\) is feasible for \(F\). Therefore
\[
    \operatorname{OPT}^{\mathrm{MW}}_{\mathrm{LP}}(F,W)
    \le
    \operatorname{OPT}^{\mathrm{MW}}_{\mathrm{LP}}(S,W).
\]
Combining this inequality with part (1) gives equality.
\end{proof}

\begin{lemma}[Budget-\(k\) sparse torso equivalence]%
\label{lem:sparse-torso-equivalence}%
Assume the current instance is cleaned, so \(\mathcal T\) has no equal-endpoint pair. Let
\(S=\operatorname{storso}_k(G,v)\), and let \(\mathcal T^\rho\) be the representative
terminal-pair update described above. For every
\[
    X\subseteq V(G)\setminus (W\cup\{v\})
    \qquad\text{with}\qquad
    |X|\le k,
\]
the following are equivalent:
\begin{enumerate}
    \item \(X\) separates all pairs of \(\mathcal T\) and all distinct vertices of \(W\)
    in \(G\), under the branch promise \(v\notin X\);
    \item \(X\) separates all pairs of \(\mathcal T^\rho\) and all distinct vertices of
    \(W\) in \(S\).
\end{enumerate}
Moreover, if \(X\) is shadowless in \(G\) and \(v\notin X\), then \(X\) is
shadowless in \(S\).
\end{lemma}

\begin{proof}
First compare connectivity for vertices not equal to \(v\). Let
\(p,q\in V(G)\setminus\{v\}\). If \(p\) and \(q\) are connected in \(G-X\), take
a \(p\)-\(q\) path. Every maximal subpath using \(v\) has the form
\(a-v-b\), with \(a,b\in N_G(v)\setminus X\). If \(a=b\), remove the detour. If
\(a\neq b\), the connector \(\operatorname{Conn}_{4k+1}(N_G(v))-X\) connects
\(a\) and \(b\), because it remains connected after at most \(k\) deletions.
Replacing every \(a-v-b\) detour gives a \(p\)-\(q\) walk in \(S-X\).

Conversely, if \(p\) and \(q\) are connected in \(S-X\), expand each connector
edge \(ab\) to the length-two walk \(a-v-b\) in \(G-X\). Since \(v\notin X\),
this gives a \(p\)-\(q\) walk in \(G-X\). These two transformations prove
equivalence for all original terminal pairs not involving \(v\) and for all pairs
of distinct terminals in \(W\).

Now consider an original pair \((v,t)\in \mathcal T\); the case \((s,v)\) is symmetric.
If \(v\) is connected to \(t\) in \(G-X\), then some neighbor
\(a\in N_G(v)\setminus X\) is connected to \(t\) in \(G-v-X\). If
\(|N_G(v)|\le k+1\), then \(a\in R_v\). If \(|N_G(v)|>k+1\), then \(R_v\) has
\(k+1\) vertices and at least one \(r\in R_v\setminus X\) survives; the connector
minus \(X\) connects \(r\) to \(a\). Hence some updated pair \((r,t)\) is connected
in \(S-X\).

Conversely, if an updated pair \((r,t)\) is connected in \(S-X\), expanding
connector edges gives a connection from \(r\) to \(t\) in \(G-X\). Since \(r\) is
adjacent to \(v\) and \(v\notin X\), the original pair \((v,t)\) is connected in
\(G-X\). Thus the representative update is exact for all budget-\(k\) sets
respecting the branch promise.

For shadowlessness, let \(u\in V(S)\setminus X\). Since \(X\) is shadowless in
\(G\), \(u\) has a path in \(G-X\) to some terminal in \(W\). The first
connectivity transformation converts that path into a path in \(S-X\). Thus every
component of \(S-X\) contains a vertex of \(W\).
\end{proof}

\begin{lemma}[Full torso and nonzero vertices]%
\label{lem:full-torso-nonzero}%

Let
\[
    \lambda=\operatorname{OPT}^{\mathrm{MW}}_{\mathrm{LP}}(G,W)\le 2k,
\]
and let \(F\) be the full torso of \(v\).
\begin{enumerate}
    \item \(\operatorname{OPT}^{\mathrm{MW}}_{\mathrm{LP}}(F,W)\) equals the optimum of the LP
    on \((G,W)\) with the additional constraint \(d_v=0\).
    \item A vertex \(v\) is nonzero if and only if
    \[
        \operatorname{OPT}^{\mathrm{MW}}_{\mathrm{LP}}(F,W)>\lambda.
    \]
    Here nonzero means that every optimal solution has \(d_v>0\).
    \item If \(v\) is nonzero, then
    \[
        \operatorname{OPT}^{\mathrm{MW}}_{\mathrm{LP}}(F,W)\ge \lambda+\frac12 .
    \]
    \item If \(S=\operatorname{storso}_k(G,v)\) and the sparse-torso child is
    live, meaning \(\operatorname{OPT}^{\mathrm{MW}}_{\mathrm{LP}}(S,W)\le 2k\), then a nonzero
    \(v\) satisfies
    \[
        \operatorname{OPT}^{\mathrm{MW}}_{\mathrm{LP}}(S,W)\ge \lambda+\frac12 .
    \]
\end{enumerate}
\end{lemma}

\begin{proof}
For (1), a feasible LP solution on \(F\) extends to \(G\) by setting \(d_v=0\):
every \(W\)-\(W\) path in \(G\) can be turned into a walk in \(F\) of the same LP
mass by replacing each subpath \(a-v-b\) by the clique edge \(ab\). Conversely, a
feasible LP solution on \(G\) with \(d_v=0\) restricts to a feasible solution on
\(F\): every clique edge \(ab\) of \(F\) expands to \(a-v-b\) with no extra LP
mass. The two feasible regions projected to variables outside \(W\cup\{v\}\) have
the same objective values.

Part (2) is now the definition of nonzero. For (3), the two finite optima are
half-integral by \Cref{thm:node-multiway-cut-lp}, so a strict increase is at least \(1/2\). If
\(\operatorname{OPT}^{\mathrm{MW}}_{\mathrm{LP}}(F,W)=+\infty\), the inequality is also true.
For (4), Lemma~\ref{lem:sparse-torso-lp-preservation} gives
\[
    \operatorname{OPT}^{\mathrm{MW}}_{\mathrm{LP}}(S,W)
    =
    \operatorname{OPT}^{\mathrm{MW}}_{\mathrm{LP}}(F,W)
\]
whenever the sparse child is live.
\end{proof}

\begin{lemma}[Cleaning does not hurt measure progress]%
\label{lem:cleaning-measure-progress}%
Suppose an operation produces a graph \(H\) with budget \(b\), and then cleaning
force-deletes a set \(Q\subseteq V(H)\setminus W\), producing \(H-Q\) with budget
\(b-|Q|\). Then
\[
    2(b-|Q|)-\operatorname{OPT}^{\mathrm{MW}}_{\mathrm{LP}}(H-Q,W)
    \le
    2b-\operatorname{OPT}^{\mathrm{MW}}_{\mathrm{LP}}(H,W)-|Q|.
\]
\end{lemma}

\begin{proof}
Any feasible LP solution on \(H-Q\) extends to a feasible LP solution on \(H\) by
assigning value \(1\) to each vertex of \(Q\). Thus
\[
    \operatorname{OPT}^{\mathrm{MW}}_{\mathrm{LP}}(H,W)
    \le
    \operatorname{OPT}^{\mathrm{MW}}_{\mathrm{LP}}(H-Q,W)+|Q|.
\]
Rearranging and substituting into the displayed measure gives the claim.
\end{proof}

When no nonzero vertex exists, for every $w\in W$, let $FC_G(w)$ be
the source side of the furthest minimum
$\{w\}$--$(W\setminus\{w\})$ separator.

The following is Lemma~13 of~\cite{chu2026faster}.

\begin{lemma}[Isolating-cut identity]%
\label{lem:isolating-cut-identity}%
If there is no nonzero vertex, then
\[
2\operatorname{OPT}^{\mathrm{MW}}_{\mathrm{LP}}(G,W)
=
\sum_{w\in W}
\left|N_G(FC_G(w))\setminus W\right|.
\]
\end{lemma}

At a live node,
$\operatorname{OPT}^{\mathrm{MW}}_{\mathrm{LP}}(G,W)\leq 2k$, and hence
Lemma~\ref{lem:isolating-cut-identity} implies
\[
\sum_{w\in W}
\left|N_G(FC_G(w))\setminus W\right|
\leq 4k.
\]

\begin{definition}[Contractible solution]%
\label{def:contractible-solution}%
A solution \(X\) is contractible if
\begin{enumerate}
    \item \(X\) contains no vertex of
    \[
        \bigcup_{w\in W}\bigl(N(FC_G(w))\setminus W\bigr);
    \]
    \item for every \(w\in W\) and every
    \(v\in N(FC_G(w))\setminus W\), vertices \(v\) and \(w\) lie in different
    components of \(G-X\).
\end{enumerate}
\end{definition}

\begin{lemma}[Branching to contractible candidates]\label{lemma:branchingcontractible}%
There is a branching algorithm that produces a branching tree with at most \(k^{O(k)}\) nodes and can be implemented in
\[
    k^{O(k)}
\widetilde O\!\left(|W|(|V|+|E|+|\mathcal T|)
\right)
\]
time, using sparse torsos only. If the input has a shadowless solution, then some node of the branching tree has a shadowless contractible solution\footnote{As remarked in \cite{chu2026faster} we even preserve an optimal shadowless contractible solution}. If any tree node is a
yes-instance, then the original instance is a yes-instance.
\end{lemma}

\begin{proof}
At a node with remaining budget \(b\le k\), first clean equal-endpoint pairs by
Lemma~\ref{lem:clean-equal-endpoint-pairs}. If the branch survives, compute \(\operatorname{OPT}^{\mathrm{MW}}_{\mathrm{LP}}(G,W)\)
with threshold \(2b\) using \Cref{thm:node-multiway-cut-lp}. If
\(\operatorname{OPT}^{\mathrm{MW}}_{\mathrm{LP}}(G,W)>2b\), discard the node. Otherwise define
the measure
\[
    \mu=2b-\operatorname{OPT}^{\mathrm{MW}}_{\mathrm{LP}}(G,W).
\]

We next find a nonzero vertex if one exists, without constructing a full clique
torso. Keep the half-integral optimum \(d\) returned by the LP routine. Every
nonzero vertex belongs to the support of \(d\), and since
\[
    \sum_v d_v\le 2b\le 2k
\]
and every positive value is at least \(1/2\), this support has size at most
\(4k\). For each support vertex \(v\), build the sparse one-vertex torso
\(\operatorname{storso}_k(G,v)\) and run the thresholded Node Multiway Cut LP on
it with threshold \(2k\). If the call certifies value larger than \(2k\), then
\(v\) is nonzero because the current optimum is at most \(2b\le 2k\). Otherwise
let \(\lambda_v\) be the returned optimum value. By Lemmas~\ref{lem:sparse-torso-lp-preservation}
and~\ref{lem:full-torso-nonzero}, $v$ is nonzero exactly when
\[
    \lambda_v>\operatorname{OPT}^{\mathrm{MW}}_{\mathrm{LP}}(G,W).
\]
If no support vertex passes this test, then no nonzero vertex exists.

If a nonzero vertex \(v\) exists, create two children. The delete child removes
\(v\), removes all terminal pairs containing \(v\), and decreases the budget to
\(b-1\). The sparse-torso child replaces \(G\) by \(\operatorname{storso}_k(G,v)\),
applies the representative update to \(\mathcal T\), keeps budget \(b\), and then cleans.
The delete child decreases \(\mu\) by at least \(1\): every LP solution on
\(G-v\) extends to one on \(G\) by setting \(d_v=1\), so
\[
    \operatorname{OPT}^{\mathrm{MW}}_{\mathrm{LP}}(G-v,W)
    \ge
    \operatorname{OPT}^{\mathrm{MW}}_{\mathrm{LP}}(G,W)-1.
\]
The live sparse-torso child decreases \(\mu\) by at least \(1/2\) by Lemma~\ref{lem:full-torso-nonzero}(4),
and cleaning preserves or improves this progress by Lemma~\ref{lem:cleaning-measure-progress}.

If no nonzero vertex exists, compute the furthest isolating regions \(FC_G(w)\)
for all \(w\in W\) before applying the second branching rule. At a live no-nonzero node,
Lemma~\ref{lem:isolating-cut-identity} gives
\[
    \sum_{w\in W}\lambda_G(w)
    =
    2\operatorname{OPT}^{\mathrm{MW}}_{\mathrm{LP}}(G,W)
    \le 4k.
\]
Thus every isolating separator has size at most \(4k\). For each \(w\in W\), we
compute \(FC_G(w)\) by the standard vertex-splitting reduction to an \(s\)-\(t\)
max-flow instance and run capped augmenting paths, stopping after \(4k+1\)
augmentations. The vertices that cannot reach the sink in the final residual graph give the
furthest minimum isolating region. Computing the isolating regions separately for all $w\in W$ takes
\[
\widetilde O\!\left(
k|W|(|V|+|E|)
\right)
\]
time per no-nonzero node.

If no nonzero vertex exists, store the current node as a candidate. For every
\(w\in W\) and every \(v\in N(FC_G(w))\setminus W\), branch on deleting \(v\) and
on contracting \(v\) into \(w\). There are at most $4k$ such ordered choices at a live node by
Lemma~\ref{lem:isolating-cut-identity}. Delete branches decrease the measure by at least \(1\)
as above. For a contraction branch, the argument in the proof of Lemma~17
of~\cite{chu2026faster} shows that either the LP optimum strictly
increases, or the contracted instance contains a nonzero vertex.
By half-integrality, any strict increase is at least $1/2$.
Consequently, unless the branch is pruned, either the measure decreases
by at least $1/2$ immediately or the next nonzero-vertex branch decreases
it by at least $1/2$.
Lemma~\ref{lem:cleaning-measure-progress} again shows that cleaning cannot undo progress. Therefore every one or
two live branch levels decrease \(\mu\) by at least \(1/2\). Since
\(0\le \mu\le 2k\), every root-to-leaf path has length \(O(k)\), and the tree has
\(k^{O(k)}\) nodes.

Soundness is checked by reversing a root-to-descendant path. A delete child
solution lifts by adding back the deleted vertex. A sparse-torso child solution
lifts by Lemma~\ref{lem:sparse-torso-equivalence}. A contract child solution lifts by Lemma~7
of~\cite{chu2026faster}. Cleaning is
sound by Lemma~\ref{lem:clean-equal-endpoint-pairs}. Thus any yes-descendant gives a yes-root.

For preservation of a shadowless solution, let \(X\) be a shadowless solution at
a live node. At a nonzero branch, if \(v\in X\), the delete child preserves
\(X\setminus\{v\}\) and shadowlessness because
\[
    (G-v)-(X\setminus\{v\})=G-X.
\]
If \(v\notin X\), Lemma~\ref{lem:sparse-torso-equivalence} preserves feasibility and shadowlessness in the
sparse-torso child. At a no-nonzero node, if \(X\) is contractible then the stored
candidate is the desired one. If not, then Definition~\ref{def:contractible-solution} fails. If Item~(1) fails, some
boundary vertex \(v\) belongs to \(X\), and the corresponding delete child
preserves a shadowless solution. If Item~(1) holds but Item~(2) fails, then for some
\(w\in W\) and some \(v\in N(FC_G(w))\setminus W\), the vertex \(v\notin X\) is
in the same component of \(G-X\) as \(w\); the corresponding contract child preserves a shadowless solution by
Lemma~14(3) of~\cite{chu2026faster}. Following the preserving child
until the finite tree stops yields a stored candidate with a shadowless
contractible solution.

Sparse connectors have \(O(k|A|)\) edges, and there are only \(O(k)\) sparse-torso
operations on any live path. Terminal-pair representative updates multiply the
pair list by at most \(k+1\) per sparse torso. Thus every node has size bounded
by a \(k^{O(k)}\) factor times the input size. Multiplying this size bound by the
\(k^{O(k)}\) node bound and by the threshold LP cost from \Cref{thm:node-multiway-cut-lp} gives the
stated running time.
\end{proof}

\subsubsection{Final Region Contraction.}

This section follows the exact same algorithm of \cite{chu2026faster}, so we only state the algorithm and do not prove the full correctness. We will argue only about the running time.

For \(R\subseteq W\), define
\[
    P_W(G,R)
    =
    \bigcap_{w\in R} FC_G(w)
    \setminus
    \bigcup_{w'\in W\setminus R} FC_G(w').
\]
The empty intersection is \(V(G)\), so
\[
    P_W(G,\emptyset)
    =
    V(G)\setminus \bigcup_{w\in W}FC_G(w).
\]
For \(S\subseteq V(G)\), let
\[
    B(S)=N_G(V(G)\setminus S)\cap S.
\]
If \(P_W(G,\emptyset)\neq\emptyset\), output a fixed no-instance. Otherwise, for
every \(w\in W\), contract every vertex of
\[
    B(P_W(G,\{w\}))\setminus W
\]
into \(w\), and update terminal pairs accordingly.

\begin{lemma}[Near-linear implementation of the final contraction]%
\label{lem:final-contraction-implementation}%
Given a stored candidate and its already computed regions \(FC_G(w)\) for all
\(w\in W\), the final region contraction can be implemented in
\[
    O(|W||V|+|E|+|\mathcal T|)
\]
time, up to sorting or hashing factors.
\end{lemma}

\begin{proof}
Compute for every vertex its membership mask
\[
    M(v)=\{w\in W : v\in FC_G(w)\}
\]
by scanning all stored regions. This costs \(O(|W||V|)\) time. A vertex lies in
\(P_W(G,\emptyset)\) exactly when \(M(v)=\emptyset\); if this happens, return the
fixed no-instance.

Otherwise a nonterminal vertex \(v\) lies in \(B(P_W(G,\{w\}))\setminus W\)
exactly when \(M(v)=\{w\}\) and some neighbor \(u\) has \(M(u)\neq\{w\}\). A
single edge scan marks precisely these vertices and records their image terminal
\(w\). A second edge scan constructs the contracted graph by replacing each edge
\(uv\) by the edge between the images of \(u\) and \(v\), discarding self-loops
and parallel edges. One scan of \(\mathcal T\) updates terminal pairs in the same way,
retaining equal-endpoint pairs for Lemma~\ref{lem:clean-equal-endpoint-pairs}. The total time after the masks are
known is
\[
    \widetilde O(|E|+|\mathcal T|).
\]
\end{proof}

\begin{theorem}[Sparse bipedalization]%
\label{thm:sparse-bipedalization}%
Given a compression instance $(G, \mathcal{T}, W, k)$ that may be assumed to have a shadowless solution in
the yes case, one can output \(k^{O(k)}\) explicit bipedal instances in
\[
    \widetilde O(k^{O(k)}|W|(|V(G)|+|E(G)|+|\mathcal{T}|))
\]
time such that a shadowless yes-case is preserved in one output and every output
solution lifts to a solution of the input.
\end{theorem}

\begin{proof}
Apply Lemma \ref{lemma:branchingcontractible}. For every stored candidate, run the final region contraction
implemented in
Lemma~\ref{lem:final-contraction-implementation}; then clean equal-endpoint pairs by Lemma~\ref{lem:clean-equal-endpoint-pairs}. Lemma~22 of~\cite{chu2026faster} gives the three correctness
properties of this final contraction: a shadowless
contractible solution is preserved, every output solution lifts to the candidate,
and every nontrivial output is bipedal. Lemma~\ref{lem:clean-equal-endpoint-pairs} preserves these properties and
keeps bipedal instances bipedal. Combining with \Cref{lemma:branchingcontractible} gives the theorem.
\end{proof}

\section{Solving Bipedal Instances by \textsc{Almost 2-SAT}.}\label{sec:almost2sat}

Let
\[
    I=(G,\mathcal T,W,k)
\]
be a bipedal compression instance in which every vertex outside \(W\) is deletable.
The solver below has the exact one-sided guarantee needed after shadow removal and
bipedalization: if \(I\) has a shadowless solution, then the formula is a yes-instance;
if the formula is a yes-instance, then \(I\) has an actual multicut solution.

Before constructing the formula, reject if two distinct vertices of \(W\) are adjacent.
Such an edge cannot be hit by deleting vertices outside \(W\), so no feasible compression
solution exists in that case.

For every component \(C\) of \(G-W\), let
\[
    L(C)=N_G(C)\cap W,
\]
so \(|L(C)|\le 2\). For each \(v\in C\) and each \(w\in L(C)\), introduce a variable
\[
    p_{v,w},
\]
interpreted as ``\(v\) is kept and assigned to terminal \(w\).''

First preprocess zero-legged components. If \(L(C)=\emptyset\), then every vertex of
\(C\) must be deleted by any shadowless solution. We therefore add all vertices of \(C\)
to a remembered forced set \(Q_0\), delete them from the graph, decrease the budget by
\(|C|\), and remove terminal pairs incident with these vertices. If the budget becomes
negative, this branch rejects. This preprocessing cannot create a false yes: any solution
found later is combined with \(Q_0\), and the removed terminal pairs are separated by the
deletion of their incident endpoint. It preserves every shadowless yes-case because every
shadowless solution contains all vertices of every zero-legged component. After this
preprocessing, we rename the reduced instance and remaining budget back to
\[
    I=(G,\mathcal T,W,k);
\]
every remaining component has one or two labels.

We build a 2-CNF formula with soft clauses and hard clauses. \textsc{Almost 2-SAT}
deletes clauses; a hard clause is represented by \(k+1\) identical copies, so a deletion
set of size at most \(k\) cannot remove all copies.

\paragraph{Soft clauses.}
For each vertex \(v\in C\):
\begin{itemize}
    \item if \(L(C)=\{a,b\}\), add the soft clause
    \[
        (p_{v,a}\vee p_{v,b});
    \]
    \item if \(L(C)=\{a\}\), add the soft unit clause
    \[
        (p_{v,a}).
    \]
\end{itemize}

\paragraph{Hard clauses.}
For a two-legged component with \(L(C)=\{a,b\}\), add for every \(v\in C\) the hard
clause
\[
    (\neg p_{v,a}\vee \neg p_{v,b}),
\]
so a kept vertex has at most one label. If \(v\) is adjacent to \(a\in W\), add the hard
unit clause
\[
    (\neg p_{v,b});
\]
if \(v\) is adjacent to \(b\), add
\[
    (\neg p_{v,a}).
\]

For every edge \(uv\) inside a component with labels \(\{a,b\}\), add
\[
    (\neg p_{u,a}\vee \neg p_{v,b})
    \qquad\text{and}\qquad
    (\neg p_{u,b}\vee \neg p_{v,a}).
\]
These forbid an undeleted edge from crossing between the two terminal labels. One-legged
components need no internal edge clauses.

For every terminal pair \((s,t)\in \mathcal T\), add clauses forbidding the two endpoints from
being kept with a common label. A terminal \(w\in W\) is treated as a fixed kept object
with label \(w\). Thus:
\begin{itemize}
    \item if \(s,t\notin W\), add
    \[
        (\neg p_{s,w}\vee \neg p_{t,w})
    \]
    for every common label \(w\);
    \item if \(s=w\in W\) and \(t\notin W\), add
    \[
        (\neg p_{t,w})
    \]
    if \(w\) is a label of \(t\);
    \item if \(s,t\in W\) and \(s=t\), discard as no; if \(s\neq t\), no clause is needed
    because distinct \(W\)-terminals are already separated by the multiway requirement.
\end{itemize}
Equal-endpoint pairs \((v,v)\) with \(v\notin W\) are handled by hard unit clauses
forbidding every label of \(v\), forcing deletion of the soft clause of \(v\).

\begin{lemma}[Formula size]%
\label{lem:bipedal-formula-size}%
The formula has \(O(|V|)\) variables and
\[
    O(k(|V|+|E|+|\mathcal T|))
\]
clauses after hard-clause copying.
\end{lemma}

\begin{proof}
Each nonterminal vertex has at most two label variables. Before copying, each vertex,
edge, and terminal pair contributes \(O(1)\) clauses. Each hard clause is copied \(k+1\)
times.
\end{proof}

\begin{lemma}[Bipedal formula correctness]%
\label{lem:bipedal-formula-correctness}%
If \(I\) has a shadowless solution of size at most \(k\), then the constructed
\textsc{Almost 2-SAT} instance can be made satisfiable by deleting at most \(k\) soft
clauses. Conversely, if the formula can be made satisfiable by deleting at most \(k\)
clauses, then \(I\) has a solution of size at most \(k\).
\end{lemma}

\begin{proof}
Let \(X\) be a shadowless solution. For every \(v\notin X\), shadowlessness gives a
component of \(G-X\) containing \(v\) and some terminal \(w\in W\). If \(C\) is the
component of \(G-W\) containing \(v\), then
\[
    w\in N_G(C)\cap W=L(C),
\]
because any path from \(v\) to \(w\) first leaves \(C\) through a neighbor in \(W\). Set
\(p_{v,w}=1\) and all other label variables of \(v\) to \(0\). For \(v\in X\), set all
label variables to \(0\) and delete the soft clause of \(v\).

Every undeleted vertex satisfies its soft clause. A two-legged kept vertex receives
exactly one label, so the at-most-one hard clauses hold. If \(v\) is adjacent to terminal
\(a\), then in \(G-X\) the vertex \(v\) is in the component of \(a\); assigning it the
other label would connect two terminals of \(W\), contradicting feasibility of \(X\).
Thus the terminal-edge hard clauses hold. If an edge \(uv\) joined two kept vertices
assigned to different labels \(a\) and \(b\), then the component of \(a\) in \(G-X\)
would contain \(u\), the edge \(uv\), and \(v\), and hence would meet the component of
\(b\), connecting distinct terminals of \(W\); this is impossible. Thus the edge hard
clauses hold. Finally, if a terminal pair clause failed, then both endpoints of that
terminal pair would be kept with a common label \(w\); the endpoint-to-\(w\) paths inside
\(G-X\) would connect the terminal pair, contradicting feasibility of \(X\).

Conversely, suppose deleting at most \(k\) clauses makes the formula satisfiable. Since
each hard clause has \(k+1\) copies, deleting at most \(k\) clauses leaves at least one
copy of each hard clause, so every hard constraint is satisfied. Let \(X\) be the set of
vertices whose soft clause is deleted. There is one soft clause per nonterminal vertex,
so
\[
    |X|\le k.
\]

We first show that distinct terminals of \(W\) are separated in \(G-X\). Suppose there is
a path from \(a\in W\) to \(b\in W\), where \(a\neq b\). Delete repeated vertices to make
it simple. The path has at least one nonterminal vertex because adjacent distinct terminals
were rejected before the formula was built. The first nonterminal vertex after \(a\) must
have label \(a\) by the terminal-edge hard clauses. Along every internal edge, the edge
hard clauses forbid changing from label \(a\) to a different label. Therefore the last
nonterminal vertex before \(b\) must still have label \(a\), while the terminal-edge hard
clauses force it to have label \(b\), a contradiction. Hence no such path exists.

Now suppose a terminal pair \((s,t)\in \mathcal T\) is connected in \(G-X\). Traverse a simple
\(s\)-\(t\) path. If the path meets a terminal \(w\in W\), then the preceding paragraph
implies it cannot meet any other terminal of \(W\), and all adjacent nonterminal vertices
on the path are forced to carry label \(w\). If the path does not meet \(W\), all its
vertices lie in one component of \(G-W\), which has one or two labels; the edge hard
clauses prevent a change of label along the path. In either case, the two endpoints have
a common label \(w\) in the sense used by the terminal-pair clauses, with terminals
treated as fixed labels. The corresponding hard pair clause forbids this. Thus no terminal
pair is connected in \(G-X\), and \(X\) solves \(I\).
\end{proof}

\begin{theorem}[Near-linear bipedal solver]%
\label{thm:near-linear-bipedal-solver}%
A bipedal compression instance can be processed in
\[
    2^{O(k)}\widetilde O(|V|+|E|+|\mathcal T|)
\]
time with the following guarantees: it outputs yes if there is a shadowless solution of
size at most \(k\), and every yes output certifies an actual solution of size at most \(k\).
\end{theorem}

\begin{proof}
Construct the formula described above. By Lemma~\ref{lem:bipedal-formula-size}, its number of variables plus clauses
is
\[
    O(k(|V|+|E|+|\mathcal T|)).
\]
Invoking the \textsc{Almost 2-SAT} algorithm of~\cite{iwata2014linear} solves it in 
\[
    2^{O(k)}\widetilde O(|V|+|E|+|\mathcal T|)
\]
time. If the bipedal instance has a shadowless solution, Lemma~\ref{lem:bipedal-formula-correctness} supplies a deletion set
of at most \(k\) soft clauses, so the \textsc{Almost 2-SAT} algorithm returns yes. If the
\textsc{Almost 2-SAT} algorithm returns yes, Lemma~\ref{lem:bipedal-formula-correctness} converts the deleted soft clauses
into a vertex set of size at most \(k\) solving the bipedal compression instance.
\end{proof}

\section{AI Disclosure.}
We used GPT 5.5 Thinking/Pro and Codex to fill in the details and help write parts of~\Cref{sec:multicut} and the related appendices, and give feedback on the paper. All approaches were devised by the authors, and the AI-assisted proofs were revised and verified by the authors.